\pdfoutput=1  

\documentclass{article}
\usepackage[margin=.8in,left=.8in]{geometry}

\usepackage{amsthm}
\usepackage{amssymb}
\usepackage{amsmath}
\usepackage{mathtools}
\usepackage{adjustbox}
\usepackage{xcolor}
\usepackage{stmaryrd}    
\usepackage{hyperref}
\hypersetup{
  colorlinks = true,
  allcolors=.
}
\usepackage[all]{xy}
\usepackage{rotating}
\usepackage{xfrac}
\usepackage{enumitem}

\usepackage{tikz}
\usetikzlibrary{backgrounds, decorations.pathmorphing}
\usetikzlibrary{cd}

\usepackage[titletoc]{appendix}

\usepackage[safe]{tipa} 

\definecolor{darkblue}{rgb}{0.05,0.25,0.65}
\definecolor{greenii}{RGB}{20,140,10}
\definecolor{darkgreen}{rgb}{0.00,0.85,0.1}
\definecolor{lightgray}{rgb}{0.9,0.9,0.9}
\definecolor{orangeii}{RGB}{200,100,5}
\definecolor{darkyellow}{rgb}{.91,.91,0}

\usepackage{multirow}

\usepackage{enumerate} 

\usepackage{mathptmx}
\usepackage{amsmath}
\usepackage{graphicx}
\DeclareRobustCommand{\coprod}{\mathop{\text{\fakecoprod}}}
\newcommand{\fakecoprod}{%
  \sbox0{$\prod$}%
  \smash{\raisebox{\dimexpr.9625\depth-\dp0}{\scalebox{1}[-1]{$\prod$}}}%
  \vphantom{$\prod$}%
}

\newcommand{\Differential}{
  \mathrm{d}
}

\newcommand{\DeRham}{\mathrm{dR}}

\newcommand{\DifferentialForms}[3]{
  \Omega^{#2}_{\mathrm{dR}}
  #1(
    #3
  #1)
}

\newcommand{\deRhamDifferential}{\mathrm{d}}
\newcommand{\DeRhamDifferential}{\deRhamDifferential}

\newcommand{\DolbeaultDifferential}{\partial}

\newcommand{\DolbeaultDifferentialForms}[4]{
 \DifferentialForms{}{#2,#3}
   {#4}
}

\newcommand{\HolomorphicDifferentialForms}[3]{
  \DolbeaultDifferentialForms{#1}{#2}{0}{#3}
    \vert_{ \overline{\DolbeaultDifferential}=0}
}
\newcommand{\HolomorphicForms}[3]{\HolomorphicDifferentialForms{#1}{#2}{#3}}

\newcommand{\DeRhamComplex}[2]{
  \DifferentialForms{#1}{\bullet}
    {#2}
}

\newcommand{\HolomorphicDeRhamComplex}[2]{\HolomorphicDifferentialForms{#1}{\bullet}{#2}}

\newcommand{\HypothesisK}{\hyperlink{HypothesisK}{\it Hypothesis K}}
\newcommand{\HypothesisH}{\hyperlink{HypothesisH}{\it Hypothesis H}}

\newcommand{\EquivarianceGroup}{G}

\newcommand{\Order}{\mathrm{ord}}

\newcommand{\CharacterGroup}[1]{
  #1^\ast
}

\newcommand{\Irreps}[1]{\mathrm{Irr}(#1)}

\newcommand{\RepresentationRing}[2]{
  R_{\scalebox{.67}{$#2$}}({#1})
}

\newcommand{\ComplexRepresentationRing}[1]{
  \RepresentationRing{#1}{\ComplexNumbers}
}

\newcommand{\NaturalNumbers}{
  \mathbb{N}
}

\newcommand{\Integers}{
  \mathbb{Z}
}

\newcommand{\ProfiniteIntegers}{
  \widehat{\Integers}
}

\newcommand{\Rationals}{
  \RationalNumbers
}

\newcommand{\RationalNumbers}{
  \mathbb{Q}
}

\newcommand{\RealNumbers}{
  \mathbb{R}
}

\newcommand{\Torus}[1]{\mathbb{T}^{#1}}

\newcommand{\MFTorus}{\Torus{2}_{\scalebox{.6}{M/F}}}

\newcommand{\ImaginaryUnit}{
  \mathrm{i}
}

\newcommand{\ComplexNumbers}{\mathbb{C}}

\newcommand{\ComplexNumber}{c}

\newcommand{\Quaternions}{
  \mathbb{H}
}

\newcommand{\MasterFunction}{
  \ell
}

\newcommand{\CyclicGroup}[1]{C_{#1}}

\newcommand{\IrrepOfCyclicGroup}[2]{\ComplexNumbers_{#1/#2}}

\newcommand{\modulo}{\,\mathrm{mod}\,}

\newcommand{\RationalCircleGroup}{
  \Quotient
    { \Rationals }
    { \Integers }
}

\newcommand{\UnitaryGroup}{
  \mathrm{U}
}

\newcommand{\SpecialUnitaryGroup}{
  \mathrm{S}\UnitaryGroup
}

\newcommand{\QuaternionicGroup}{
  \mathrm{Sp}
}

\newcommand{\SpOne}{
  \QuaternionicGroup(1)
}

\newcommand{\SpinGroup}{
  \mathrm{Spin}
}

\newcommand{\Spin}{\SpinGroup}

\newcommand{\UnitaryLieAlgebra}[1]{\mathfrak{u}_{#1}}

\newcommand{\SpecialUnitaryLieAlgebra}[1]{\mathfrak{s}\UnitaryLieAlgebra{#1}}
\newcommand{\suTwo}{\SpecialUnitaryLieAlgebra{2}}

\newcommand{\SpecialLinearLieAlgebra}[1]{\mathfrak{sl}_{#1}}
\newcommand{\slTwo}{\SpecialLinearLieAlgebra{2}}
\newcommand{\RaisingOperator}{e}
\newcommand{\LoweringOperator}{f}
\newcommand{\WeightOperator}{h}
\newcommand{\HighestWeightVector}{v^0}

\newcommand{\HighestWeightIrrep}[1]{L^{#1}}

\newcommand{\AffineHighestWeightIrrep}[2]{\widehat{{L}_{#1}\hspace{-1.38pt}}^{\raisebox{-2.5pt}{\scalebox{.72}{\hspace{+.2pt}$#2$}}}}

\newcommand{\AffineCharacter}[2]{\mathrm{ch}_{#1}^{#2}}

\newcommand{\IncomingHighestWeightReps}
{L_{\mathrm{in}}}

\newcommand{\SpecialLinearGroup}{\mathrm{SL}}

\newcommand{\SLTwo}{\SpecialLinearGroup(2,\ComplexNumbers)}
\newcommand{\SLTwoZ}{\SpecialLinearGroup(2,\Integers)}

\newcommand{\TOperator}{\mathrm{T}}
\newcommand{\SOperator}{\mathrm{S}}

\newcommand{\slAffine}[2]{\widehat{\SpecialLinearLieAlgebra{#1}}^{\raisebox{-2.5pt}{\scalebox{.73}{\hspace{-1pt}$#2$}}}}
\newcommand{\suAffine}[2]{\widehat{\SpecialUnitaryLieAlgebra{#1}}^{\raisebox{-2.5pt}{\scalebox{.73}{\hspace{-1pt}$#2$}}}}

\newcommand{\suTwoAffine}[1]{\suAffine{2}{#1}}
\newcommand{\slTwoAffine}[1]{\slAffine{2}{#1}}

\newcommand{\HilbertSpace}{
  \ell^2(\ComplexNumbers)
}

\newcommand{\BoundedOperators}{\mathcal{B}}

\newcommand{\CompactOperators}{\mathcal{K}}

\newcommand{\UnitaryOperator}{u}

\newcommand{\FredholmOperators}{
  \mathrm{Fred}
}

\newcommand{\FredholmOperator}{
  F
}

\newcommand{\UH}{
  \UnitaryGroup_{\!\omega}
}

\newcommand{\CircleGroup}{
  {\UnitaryGroup_{\!1}}
}

\newcommand{\PUH}{
  \mathrm{P}\UH
}

\newcommand{\GradedUH}{\UH^{\mathrm{gr}}}

\newcommand{\GradedPUH}{
  \PUH^{\mathrm{gr}}
}

\newcommand{\SUTwo}{
  \SpecialUnitaryGroup(2)
}

\newcommand{\simplicial}{
  \Delta
}

\newcommand{\Groupoids}[1]{
  \mathrm{Grpd}_{#1}
}

\newcommand{\limit}[1]{
  \underset{
    \underset{#1}{\longleftarrow}
  }{\lim}
  \,
}

\newcommand{\colimit}[1]{
  \underset{
    \underset{#1}{\longrightarrow}
  }{\lim}
  \,
}

\newcommand{\Complex}{\mathrm{Cplx}}

\newcommand{\SmoothManifold}{
  \TopologicalSpace
}

\newcommand{\ComplexManifold}{
  \SmoothManifold
}

\newcommand{\ComplexPlane}{\ComplexNumbers}

\newcommand{\RiemannSphere}{{\ComplexNumbers}P^1}

\newcommand{\FlatConnectionForm}{\omega_1}

\newcommand{\NumberOfPunctures}{N}
\newcommand{\NumberOfProbeBranes}{n}

\newcommand{\ConfigurationSpace}[1]{  \underset{
    \scalebox{.65}{$
      \{1,\cdots,#1\}
    $}
  }
  {\mathrm{Conf}}
}

\newcommand{\weight}{\mathrm{w}}

\newcommand{\IncomingWeight}{\weight_{\mathrm{in}}}
\newcommand{\OutgoingWeight}{\weight_{\mathrm{out}}}

\newcommand{\level}{k}
\newcommand{\Level}{\level}
\newcommand{\ShiftedLevel}{\kappa}
\newcommand{\Denominator}{r}
\newcommand{\DualDenominator}{s}

\newcommand{\ConformalBlocks}{\mathrm{CnfBlck}}

\newcommand{\compact}{\mathrm{cpt}}

\newcommand{\CartesianSpaces}{
  \mathrm{CartSp}
}

\newcommand{\TopologicalSpace}{
  \mathrm{X}
}

\newcommand{\InfinityGroupoids}{
  \Groupoids{\infty}
}

\newcommand{\WeakEquivalences}{
  \mathrm{WEq}
}

\newcommand{\Localization}[1]{
  \mathrm{Loc}^{\scalebox{.6}{$#1$}}
}

\newcommand{\CRationalization}{\Localization{\ComplexNumbers}}

\newcommand{\Local}{
  \mathrm{Lcl}
}

\newcommand{\LocalWeakEquivalences}{
  \Local\WeakEquivalences
}

\newcommand{\shape}{
  \raisebox{1pt}{\rm\normalfont\textesh}
}

\newcommand{\Smooth}{
  \mathrm{Smth}
}

\newcommand{\ShapeOfSphere}[1]{
  S^{#1}
}

\newcommand{\proofstep}[1]{
  \mbox{\small #1}
}

\newcommand{\SimplicialPresheaves}{
  \simplicial\Presheaves
}

\newcommand{\Sheaves}{
  \mathrm{Sh}
}

\newcommand{\Presheaves}{
  \mathrm{P}\Sheaves
}

\newcommand{\Truncation}[1]{
  \tau_{#1}
}

\newcommand{\Slice}[2]{
  {#1}_{\scalebox{.7}{$/#2$}}
}

\newcommand{\Maps}[3]{
  \mathrm{Map}
  #1(
    #2
    ,\,
    #3
  #1)
}

\newcommand{\stable}{\mathrm{stbl}}
\newcommand{\diff}{\mathrm{diff}}
\newcommand{\homotopy}{\mathrm{htpy}}

\newcommand{\ConstrainedMaps}[3]{
  \mathrm{Map}^{\mathrm{\ast}/}
  #1(
    #2
    ,\,
    #3
  #1)
}

\newcommand{\Homs}[3]{
  \mathrm{Hom}
  #1(
    #2
    ,\,
    #3
  #1)
}

\newcommand{\TED}{$\mathrm{TED}$}

\newcommand{\DifferentialKU}{\mathrm{KU}_{\diff}}

\newcommand{\SliceMaps}[4]{
  \Maps{#1}{#3}{#4}_{\!\scalebox{.7}{$#2$}}
}

\newcommand{\Quotient}[2]{
  #1 / #2
}

\newcommand{\HomotopyQuotient}[2]{
  #1 \!\sslash\! #2
}

\newcommand{\SmoothInfinityGroupoids}{
  \Smooth\Groupoids{\infty}
}

\newcommand{\ComplexInfinityGroupoids}{\Complex\InfinityGroupoids}

\newcommand{\DoldKanCorrespondence}{\mathrm{DK}}
\newcommand{\DoldKanConstruction}{\DoldKanCorrespondence}

\newcommand{\Modules}{\mathrm{Mod}}

\newcommand{\mysetminus}{\setminus}

\usepackage[new]{old-arrows}   

\newdir{> }{{}*!/10pt/@{>}}

\usepackage{amssymb}

\def\acts{\raisebox{1.4pt}{\;\rotatebox[origin=c]{90}{$\curvearrowright$}}\hspace{.5pt}}

\makeatletter
\newif\if@sup
\newtoks\@sups
\def\append@sup#1{\edef\act{\noexpand\@sups={\the\@sups #1}}\act}%
\def\reset@sup{\@supfalse\@sups={}}%
\def\mk@scripts#1#2{\if #2/ \if@sup ^{\the\@sups}\fi \else%
  \ifx #1_ \if@sup ^{\the\@sups}\reset@sup \fi {}_{#2}%
  \else \append@sup#2 \@suptrue \fi%
  \expandafter\mk@scripts\fi}
\def\tensor#1#2{\reset@sup#1\mk@scripts#2_/}
\def\multiscripts#1#2#3{\reset@sup{}\mk@scripts#1_/#2%
  \reset@sup\mk@scripts#3_/}
\makeatother

\makeatletter
\newbox\slashbox \setbox\slashbox=\hbox{$/$}
\def\itex@pslash#1{\setbox\@tempboxa=\hbox{$#1$}
  \@tempdima=0.5\wd\slashbox \advance\@tempdima 0.5\wd\@tempboxa
  \copy\slashbox \kern-\@tempdima \box\@tempboxa}
\def\slash{\protect\itex@pslash}
\makeatother

\def\clap#1{\hbox to 0pt{\hss#1\hss}}
\def\mathllap{\mathpalette\mathllapinternal}
\def\mathrlap{\mathpalette\mathrlapinternal}
\def\mathclap{\mathpalette\mathclapinternal}
\def\mathllapinternal#1#2{\llap{$\mathsurround=0pt#1{#2}$}}
\def\mathrlapinternal#1#2{\rlap{$\mathsurround=0pt#1{#2}$}}
\def\mathclapinternal#1#2{\clap{$\mathsurround=0pt#1{#2}$}}

\let\oldroot\root
\def\root#1#2{\oldroot #1 \of{#2}}
\renewcommand{\sqrt}[2][]{\oldroot #1 \of{#2}}

\DeclareSymbolFont{symbolsC}{U}{txsyc}{m}{n}
\SetSymbolFont{symbolsC}{bold}{U}{txsyc}{bx}{n}
\DeclareFontSubstitution{U}{txsyc}{m}{n}

\DeclareSymbolFont{stmry}{U}{stmry}{m}{n}
\SetSymbolFont{stmry}{bold}{U}{stmry}{b}{n}

\DeclareFontFamily{OMX}{MnSymbolE}{}
\DeclareSymbolFont{mnomx}{OMX}{MnSymbolE}{m}{n}
\SetSymbolFont{mnomx}{bold}{OMX}{MnSymbolE}{b}{n}
\DeclareFontShape{OMX}{MnSymbolE}{m}{n}{
    <-6>  MnSymbolE5
   <6-7>  MnSymbolE6
   <7-8>  MnSymbolE7
   <8-9>  MnSymbolE8
   <9-10> MnSymbolE9
  <10-12> MnSymbolE10
  <12->   MnSymbolE12}{}

\usepackage{cleveref}

\crefformat{section}{\S#2#1#3} 
\crefformat{subsection}{\S#2#1#3}
\crefformat{subsubsection}{\S#2#1#3}

\newtheorem{theorem}{Theorem}[section]
\newtheorem{lemma}[theorem]{Lemma}
\newtheorem{conjecture}[theorem]{Conjecture}
\newtheorem{proposition}[theorem]{Proposition}

\theoremstyle{definition}
\newtheorem{definition}[theorem]{Definition}

\newtheorem{example}[theorem]{Example}

\newtheorem{remark}[theorem]{Remark}

\usepackage{amsfonts}
\usepackage{colortbl}

\renewcommand{\emph}{\textit}

\begin{document}

\title{Anyonic Defect Branes and Conformal Blocks
in
\\
Twisted Equivariant Differential (TED) K-theory}

\author{Hisham Sati, \qquad Urs Schreiber}

\maketitle

\begin{abstract}
We demonstrate that
twisted equivariant differential K-theory of transverse complex curves
accommodates  exotic charges of the form  expected of
   codimension=2 defect branes,
   such as
   of $\mathrm{D7}$-branes in IIB/F-theory
   on $\mathbb{A}$-type orbifold singularities,
   but also of their dual
   3-brane defects of class-S theories on
   M5-branes.
   These branes have been argued,
   within F-theory and the AGT correspondence,
   to carry special
    $\mathrm{SL}(2)$-monodromy charges not seen for other branes,
    at least partially
    reflected in conformal blocks of the $\suTwo$-WZW model over their transverse punctured complex curve.
  Indeed, it has been argued that
  all ``exotic'' branes of string theory are defect branes carrying such U-duality monodromy charges
  --
  but none of these had previously been identified
  in the
  expected brane charge quantization law
  given by K-theory.

  Here we observe that it is the
  subtle
  (and previously somewhat neglected)
  twisting of equivariant K-theory
  by flat complex line bundles
  appearing
  inside orbi-singularities
  (``inner local systems'')
  that makes the
  secondary
  Chern character
  on a punctured plane inside an $\mathbb{A}$-type singularity
  evaluate to the twisted
  holomorphic de Rham cohomology which
  Feigin, Schechtman \& Varchenko
  showed
  realizes
  $\slTwoAffine{\Level}$-conformal blocks,
  here in degree 1
  --
  in
  fact it gives the direct sum of these over all
  admissible
  fractional levels $k = - 2 + \ShiftedLevel/\Denominator$.
  The remaining higher-degree
  $\slTwoAffine{\Level}$-conformal blocks
  appear similarly
  if we assume our previously discussed ``Hypothesis H'' about brane charge quantization in M-theory.
  Since conformal blocks -- and hence these twisted equivariant secondary Chern characters --
  solve the Knizhnik-Zamolodchikov equation and thus
  constitute representations of the braid group of motions of defect branes inside their transverse space,
  this provides a
  concrete
  first-principles
  realization
  of anyon statistics of
  --
   and hence of topological quantum computation on --
  defect branes in string/M-theory.
\end{abstract}

\medskip

\tableofcontents

\newpage

\section{Introduction}

It is a classical and famous fact that the Chern character on twisted K-theory classes over smooth manifolds may be understood
as taking values in (complex) de Rham cohomology twisted by a closed differential 3-form $H_3$, i.e., in the cohomology of the shifted de Rham differential $\DeRhamDifferential + H_3 \wedge$ acting on (complex) differential forms of even or odd degrees (\cite{RohmWitten86}\cite{BCMMS02},
review in \cite[Prop. 3.109 \& 5.6]{FSS20CharacterMap}). What has received much less attention than this ``3-twist'' of the Chern character is the curious fact that -- in further generalization to {\it equivariant} twisted K-theory and hence to the twisted K-theory of smooth orbifolds -- there appears
(\cite{TuXu06}\cite{FreedHopkinsTeleman02ComplexCoefficients}, see Prop. \ref{TheBareBProfiniteIntegersTwistOfComplexRationalizedAEquivariantKTheory} below)
``inside'' the  orbi-singularities a system of further shifts of the de Rham differential by closed differential 1-forms $\FlatConnectionForm$, hence by flat connections on (complex) line bundles. In \cref{OneTwistedCohomologyAsTwistedEquivariantKTheory} we give a detailed outline of a transparent proof of this fact.

\medskip
Of course, in themselves such ``1-twists'' $\DeRhamDifferential + \FlatConnectionForm \wedge$ of the de Rham differential are even more classical (\cite[\S 2, 6]{Deligne70}, review in \cite[\S 5.1.1]{Voisin03}, see also \cite{Witten82}\cite{GS-Deligne}),
often known as or implicitly identified with ``local systems'' (\cite{Steenrod43}, whence their orbifold version has been named ``inner local systems'' \cite{Ruan00}).
A particularly rich class of examples is the 1-twisted de Rham cohomology already of the simple example of $\NumberOfPunctures$-punctured planes, which is known to be the source of the ``hypergeometric'' solutions
(\cite{SchechtmanVarchenko89}\cite{SchechtmanVarchenko90}\cite{SchechtmanVarchenko91}, review in \cite[\S 7]{EtingofFrenkelKirillov98})
of the celebrated Knizhnik-Zamolodchikov (KZ) equation (\cite{KZ84}, review in \cite[\S 1.5]{Kohno02}). In the special case
that the holonomy of the twisting connection takes values in the group of rational phases $\RationalCircleGroup \xhookrightarrow{\;} \CircleGroup$ (i.e., roots of unity), these KZ solutions and hence these 1-twisted de Rham cohomology spaces constitute (\cite{FeiginSchechtmanVarchenko90}\cite{FeiginSchechtmanVarchenko94}\cite{FeiginSchechtmanVarchenko95})
the conformal blocks
(e.g. \cite{Beauville94}\cite[\S 1.4]{Kohno02})
of chiral WZW models (e.g. \cite[\S C]{DMS97}\cite{Walton00}), thus providing a {\it twisted cohomological re-incarnation of the core structure of conformal field theory} (CFT). Specifically, if the holonomy of the 1-twist is inside the cyclic group of $\ShiftedLevel$-fold roots of unity $\CyclicGroup{\ShiftedLevel} \xhookrightarrow{\;} \CircleGroup$, then the corresponding 1-twisted de Rham cohomology of the $\NumberOfPunctures$-punctured plane reflects the $\slTwoAffine{\level}$-conformal blocks at level $\Level = -2 + \ShiftedLevel$. We review this as Prop. \ref{Degree1ConformalBlocksInTwistedCohomology}, \ref{ConformalBlocksInTwistedCohomology} below.
 It is expected that this remains the case also for the fractional levels $\Level = - 2 + \ShiftedLevel/\Denominator$ at which the $\slTwo$-WZW model exists as a ``logarithmic'' CFT (see Rem. \ref{FractionalLevelWZWModels} below).

\medskip
One of the main observations of this paper is the unification of these two threads by regarding the 1-twisted de Rham cohomology which constitutes $\slTwoAffine{\Level}$-conformal blocks as the (secondary) twisted equivariant Chern characters of the twisted equivariant differential (\TED) K-theory of the punctured plane regarded as sitting inside an $\mathbb{A}_{\ShiftedLevel-1}$-orbi-singularity (Prop. \ref{Degree1ConformalBlockInTEdKTheory}, Thm. \ref{ConformalBlockInTEdKTheory} below). We observe that this K-theoretic perspective naturally enforces the holonomy in $\CyclicGroup{\ShiftedLevel} \xhookrightarrow{\;} \CircleGroup$ as well as the ``admissible'' fractional levels $\Level = -2 + \ShiftedLevel/\Denominator$ (see Rem. \ref{RationalLevelsInKTheory} below), both of which are necessary for ensuring that the general hypergeometric construction of KZ-solutions specializes to conformal blocks of WZW models.

\medskip
Moreover, under the widely expected {\HypothesisK} (see Rem. \ref{HypothesesAboutBraneChargeQuantization}) that twisted
equivariant
differential
K-theory
(henceforth \TED-K-theory) classifies the RR-fluxes and D-brane charges in type-II string theory, this identification implies (discussed in \cref{QuantumStatesAndObservables})
that charges of codim=2 ``defect branes'' on $\mathbb{A}_{\ShiftedLevel-1}$-type singularities have exotic properties not previously identified in K-theory, such as transformation laws under $\SLTwoZ$ and under the braid group of motions of the defect branes in their transverse space. Both of these are hallmarks expected of D7/D3-branes in F-theory but have not previously been derived from first principles of charge quantization. Finally, assuming also our more recently introduced {\HypothesisH} about the charge quantization of M-branes, it follows that charges of defect M-branes reflect the full structure of $\slTwoAffine{\Level}$-conformal blocks, which is a key property expected (\cite{NishiokaTachikawa11}\cite{FMMW20}\cite{Manabe20}) of the ``dual'' incarnation of D3-branes at IIB $\mathbb{A}$-type singularities as codim=2 defects inside M5-branes, according to the AGT correspondence (\cite{AGT10}, review in \cite{Tachikawa16}\cite{LeFloch20}\cite{Akhond21}).

\vspace{-.1cm}
\hspace{-.8cm}
\begin{tikzpicture}
\draw (0,0) node {
\begin{tikzcd}[
  column sep={between origins, 190pt},
  row sep=07pt
]
&
|[alias=top]|
\fbox{\!\!\!\!\!\! \small
     \def\arraystretch{.9}
     \begin{tabular}{c}
       \TED-K-Theory of
       \\
       (configuration space of points in)
       \\
       $\NumberOfPunctures$-punctured plane
       \\
       inside $\mathbb{A}_{\ShiftedLevel-1}$-singularity
     \end{tabular}
   \!\!\!\!\!}
&
\\
|[alias=left]|
 \fbox{\!\!\!\!\!\!  \small
    \def\arraystretch{.9}
    \begin{tabular}{c}
      1-twisted holomorphic
      \\
      deRham-cohomology of
      \\
      (configuration space of points in)
      \\
      $\NumberOfPunctures$-punctured plane
    \end{tabular}
  \!\!\!\!\!}
&&
|[alias=right]|
\fbox{\!\!\!\! \small
     \begin{tabular}{c}
       Exotic charges of $\NumberOfPunctures$
       \\
       codim=2 defect branes
       \\
       in F-theory (M-theory)
     \end{tabular}
   \!\!\!\!}
\\
&
|[alias=bottom]|
\fbox{\!\!\!\!\! \small
     \def\arraystretch{.9}
     \begin{tabular}{c}
       $\slTwoAffine{\Level}$-conformal blocks
       \\
       with $\NumberOfPunctures$ field insertions
       \\
       of integrable weights
       \\
       transforming under $\SLTwoZ$
     \end{tabular}
  \!\!\!\!\!\! }
  %
  %
  \ar[
    from=top, to=right,
    -,
    shorten=-9pt,
    "{
      \color{greenii}
      \scalebox{.8}{
        \begin{tabular}{c}
          {\HypothesisK}
          \\
          ({\HypothesisH})
        \end{tabular}
      }
    }"{sloped},
    "{
      \scalebox{.8}{
        Rem. \ref{HypothesesAboutBraneChargeQuantization}
      }
    }"{swap, sloped}
  ]
   \ar[
    from=top, to=left,
    -,
    shorten=-11pt,
    "{
      \hspace{-4pt}
      \color{greenii}
      \scalebox{.8}{
        \def\arraystretch{.9}
        \begin{tabular}{c}
          (secondary)
          \\
          \TED-Chern character
        \end{tabular}
      }
    }"{sloped},
    "{
      \scalebox{.8}{
        Prop. \ref{TheBareBProfiniteIntegersTwistOfComplexRationalizedAEquivariantKTheory},
        and
        \cref{OneTwistedCohomologyAsTwistedEquivariantKTheory}
      }
    }"{swap, sloped}
  ]
  \ar[
    from=top, to=bottom,
    -,
    shorten=-10,
    "{
      \scalebox{.8}{
        \cref{ConformalBlocksAsTEdKTheory}
      }
    }"{swap}
  ]
  \ar[
    from=left, to=bottom,
    -,
    shorten=-7pt,
    "{
      \color{greenii}
      \scalebox{.8}{
        \def\arraystretch{.9}
        \begin{tabular}{c}
          hypergeometric
          \\
          construction
        \end{tabular}
      }
    }"{sloped},
    "{
      \scalebox{.8}{
        Prop. \ref{Degree1ConformalBlocksInTwistedCohomology},
        \ref{ConformalBlocksInTwistedCohomology}
      }
    }"{swap, sloped}
  ]
  \ar[
    from=bottom, to=right,
    -,
    shorten=-7pt,
    dashed,
    "{
      \color{greenii}
      \scalebox{.8}{
        \def\arraystretch{.9}
        \begin{tabular}{c}
        F-theory prediction
        \\
        (AGT correspondence)
        \end{tabular}
      }
    }"{sloped},
    "{
      \scalebox{.8}{
        \cref{QuantumStatesAndObservables}
      }
    }"{swap, sloped}
  ]
\end{tikzcd}
};
%
%
\draw (0,0) node {
\begin{tikzcd}[
  column sep={between origins, 190pt},
  row sep=07pt
]
&
|[alias=top]|
\fcolorbox{black}{white}{\!\!\!\!\!\! \small
     \def\arraystretch{.9}
     \begin{tabular}{c}
       \TED-K-Theory of
       \\
       (configuration space of points in)
       \\
       $\NumberOfPunctures$-punctured plane
       \\
       inside $\mathbb{A}_{\ShiftedLevel-1}$-singularity
     \end{tabular}
   \!\!\!\!\!}
&
\\
|[alias=left]|
\fcolorbox{black}{white}{\!\!\!\!\!\!  \small
    \def\arraystretch{.9}
    \begin{tabular}{c}
      1-twisted holomorphic
      \\
      deRham-cohomology of
      \\
      (configuration space of points in)
      \\
      $\NumberOfPunctures$-punctured plane
    \end{tabular}
  \!\!\!\!\!}
&&
|[alias=right]|
\fcolorbox{black}{white}{\!\!\!\! \small
     \begin{tabular}{c}
       Exotic charges of $\NumberOfPunctures$
       \\
       codim=2 defect branes
       \\
       in F-theory (M-theory)
     \end{tabular}
   \!\!\!\!}
\\
&
|[alias=bottom]|
\fcolorbox{black}{white}{\!\!\!\!\! \small
     \def\arraystretch{.9}
     \begin{tabular}{c}
       $\slTwoAffine{\Level}$-conformal blocks
       \\
       with $\NumberOfPunctures$ field insertions
       \\
       of integrable weights
       \\
       transforming under $\SLTwoZ$
     \end{tabular}
  \!\!\!\!\!\! }
\end{tikzcd}
};
\end{tikzpicture}

\noindent
In conclusion, we describe a curious embedding of WZW conformal field theory and of anyonic braid representations into \TED-K-theory, compatibly so with expectations about defect brane charges in F-theory/M-theory. An outlook is given in \cref{Outlook}.

\medskip

%
%
%
%
%

\newpage

\noindent
{\bf Quantum states of M-branes.}
We suggest that this result further supports the {\HypothesisH} (Rem. \ref{HypothesesAboutBraneChargeQuantization}) that quantum states of M-branes in the elusive non-perturbative completion of string theory known as ``M-theory'' or ``F-theory'' are reflected in the cohomology of Cohomotopy cocycle spaces, which for the low-codimension intersecting ``flat branes'' of interest here are given by configuration spaces of points \eqref{OrderedConfigurationSpaceAsSpaceOfIntersections} in the plane. These form a direct system whose fibers
\eqref{ConfigurationSpaceAsFiberProduct}
are exactly the configuration spaces of points in the $\NumberOfPunctures$-punctured plane whose 1-twisted cohomology
reflects $\slTwoAffine{\Level}$-conformal blocks, by the above discussion.

\medskip
Previously we had discussed this (in \cite{SS19ConfigurationSpaces}\cite{CSS21}) for $\mathrm{NS5} = \mathrm{D6}\!\perp\!\mathrm{D}8$-brane intersections within the reduced M-theory bulk spacetime where charges are in 4-Cohomotopy, quantizing the degree=4 C-field flux density $G_4$ (incarnated as $F_4$).
The discussion here (\cref{QuantumStatesAndObservables}) concerns the
directly analogous situation (see \hyperlink{TableOfCohomologyOfCohomotopyCocycleSpaces}{\it Table 1}), now for $\mathrm{M3} = \mathrm{M5} \!\perp\! \mathrm{M5}$-intersections within the ambient $\mathrm{MK6}$-singularity where charges are in 3-cohomotopy (according to the discussion in \cite{FSS20TwistedString}\cite{FSS19HopfWZ}), quantizing the degree=3 flux density $H_3$.

\medskip

\begin{center}
\hypertarget{TableOfCohomologyOfCohomotopyCocycleSpaces}{}
\def\arraystretch{2.5}
\begin{tabular}{|l||l||l|c|}
  \hline
  &
  \def\arraystretch{.85}
  \begin{tabular}{c}
    \bf
    Hanany-Witten theory of
    \\
    \bf
    codim=3 branes
    \\
    (\cite{SS19ConfigurationSpaces}\cite{CSS21})
    $\mathclap{\phantom{\vert_{\vert_{\vert}}}}$
  \end{tabular}
  &
  \def\arraystretch{.85}
  \begin{tabular}{c}
    \bf
    Seiberg-Witten theory of
    \\
    \bf
    defect codim=2 branes
    \\
    (here)
    $\mathclap{\phantom{\vert_{\vert_{\vert}}}}$
  \end{tabular}
  &
  \\
  \hline
  \rowcolor{lightgray}
  \def\arraystretch{.85}
  \begin{tabular}{l}
    Intersecting branes
  \end{tabular}
  &
  \def\arraystretch{.9}
  \begin{tabular}{l}
    $\mathrm{NS5} = \mathrm{D6} \!\perp\! \mathrm{D8}$
    \\
    in 10d bulk spacetime
    $\mathclap{\phantom{\vert_{\vert_{\vert_{\vert}}}}}$
  \end{tabular}
  &
  \def\arraystretch{.9}
  \begin{tabular}{l}
    $
    \mathrm{D7}\Vert\mathrm{D3}
    \,\rightsquigarrow\,
    \mathrm{M3} = \mathrm{M5} \!\perp\! \mathrm{M5}$
    \\
    in 7d MK6 worldvolume
    $\mathclap{\phantom{\vert_{\vert_{\vert_{\vert}}}}}$
  \end{tabular}
  &
  \cref{QuantumStatesAndObservables}
  \\
  \def\arraystretch{.5}
  \begin{tabular}{l}
    Charge quantization law
  \end{tabular}
  &
  \def\arraystretch{.5}
  \begin{tabular}{c}
    4-Cohomotopy
  \end{tabular}
  &
  \def\arraystretch{.5}
  \begin{tabular}{c}
     3-Cohomotopy
  \end{tabular}
  &
  \cite{FSS20TwistedString}
  \\
\rowcolor{lightgray}  \def\arraystretch{.85}
  \begin{tabular}{l}
    Cocycle space /
    \\
    configuration space
  \end{tabular}
  &
  \begin{tabular}{c}
  $\underset{\NumberOfPunctures_{\mathrm{f}} \in \mathbb{N}}{\coprod}
  \;
  \ConfigurationSpace{\NumberOfPunctures_{\mathrm{f}}}
  \big(
    \RealNumbers^3
  \big)
  \mathclap{\phantom{\vert_{\vert_{\vert_{\vert_{\vert_{\vert_{\vert_{\vert}}}}}}}}}
  $
  \end{tabular}
  &
  \begin{tabular}{l}
  $\underset{\NumberOfPunctures \in \mathbb{N}}{\coprod}
  \;
  \ConfigurationSpace{\NumberOfPunctures}
  \big(
    \RealNumbers^2
  \big)
  \mathclap{\phantom{\vert_{\vert_{\vert_{\vert_{\vert_{\vert_{\vert_{\vert}}}}}}}}}
  $
  \end{tabular}
  &
  \eqref{OrderedConfigurationSpaceAsSpaceOfIntersections}
  \\
  \def\arraystretch{.85}
  \begin{tabular}{l}
    (Twisted, fiberwise)
    \\
    (Co)Homology / \\
    quantum states + observables
  \end{tabular}
  &
  \def\arraystretch{1}
  \begin{tabular}{c}
    Weight systems/
    \\
    Chord diagrams
  \end{tabular}
  &
  \def\arraystretch{1}
  \begin{tabular}{l}
    $\slTwoAffine{\Level}$-Conformal blocks/
    \\
    Braid group representations
  \end{tabular}
  &
  \cref{ConformalBlocksAsTEdKTheory}
  $\mathclap{\phantom{\vert_{\vert_{\vert_{\vert_{\vert_{\vert_{\vert_{\vert}}}}}}}}}$
  \\
  \hline
  \end{tabular}

  \vspace{.2cm}

  {\bf Table 1.}
  Quantum states of branes as cohomology of
  Cohomotopy cocycle spaces
  according to {\HypothesisH}.
\end{center}

\vspace{.1cm}

Independently of this relation to M-theory, the results we present now indicate a new way in which generalized cohomology theory and tools from algebraic topology are usefully brought to bear on questions of quantum physics in general and of topologival quantum computation, in particular. We come back to this at the end in \cref{Outlook}.

\medskip
\noindent
{\bf Acknowledgements.} We thank David Ridout, Eric Sharpe and Guo Chuan Thiang for useful discussions.

\medskip

\section{Conformal blocks in TED-K-Theory}
\label{ConformalBlocksAsTEdKTheory}

This section discusses how the
secondary sector of
twisted equivariant differential K-theory
(``\TED-K-theory'' for short, see \cref{OneTwistedCohomologyAsTwistedEquivariantKTheory})
of configuration spaces of points in the punctured plane
inside $\mathbb{A}_{\ShiftedLevel-1}$-singularities naturally reflects
the spaces of conformal blocks of the $\slTwoAffine{\Level}$-WZW model
(Prop. \ref{Degree1ConformalBlockInTEdKTheory}, Thm. \ref{ConformalBlockInTEdKTheory} below)
at the usual integer level $\Level = - 2 + \ShiftedLevel$,
as well as those expected for the fractional-level WZW models which are conjectured
to exist (Rem. \ref{FractionalLevelWZWModels} below).

\medskip

Our Theorem \ref{ConformalBlockInTEdKTheory} is a consequence of
Prop.  \ref{TheBareBProfiniteIntegersTwistOfComplexRationalizedAEquivariantKTheory} and Prop. \ref{ConformalBlocksInTwistedCohomology}
below, for which we now provide
a quick introduction of all the notions and notation that go into the them,
together with pointers to full details.

\medskip


\noindent
-- Prop. \ref{ConformalBlocksInTwistedCohomology} below is a digest of
Feigin, Schechtman \& Varchenko's
identification
\cite{FeiginSchechtmanVarchenko94}
of $\slTwoAffine{\Level}$-conformal blocks inside the twisted holomorphic de Rham cohomology of configuration spaces of points in the complex plane.
This is a special case of a more general construction, via ``hypergeometric integrals'', of solutions to the Knizhnik-Zamolodchikov equation due to \cite{SchechtmanVarchenko89}\cite{SchechtmanVarchenko90}\cite{SchechtmanVarchenko91}, whose specialization to level-$\Level$ conformal blocks, when the twisting holonomy is inside $\CyclicGroup{\Level+2} \xhookrightarrow{\;} \CircleGroup$,
was established
 in \cite{FeiginSchechtmanVarchenko90}\cite{FeiginSchechtmanVarchenko94}\cite{FeiginSchechtmanVarchenko95}
 (see also \cite{Varchenko95}),
following special cases established in
\cite{DotsenkoFateev84}\cite{ChristeFlume87}.
Reviews include
\cite[\S 7]{EtingofFrenkelKirillov98}\cite[\S 14.3]{FrenkelBenZvi04}\cite{Kohno12}\cite{Kohno14}.

\medskip
\noindent
-- Prop. \ref{TheBareBProfiniteIntegersTwistOfComplexRationalizedAEquivariantKTheory} highlights the structure of \TED-K-theory of $\mathbb{A}$-type singularities, in order to point out (Prop. \ref{Degree1ConformalBlockInTEdKTheory}, Thm. \ref{ConformalBlockInTEdKTheory}) that this is a natural home for the twisted cohomology groups in the FSV-construction of conformal blocks, in that it features twists by
just the relevant local systems with
exactly the necessary restriction to rational phases and to admissible levels, and automatically sums up the conformal blocks at all the commensurate fractional levels at which the $\slTwoAffine{\level}$-CFT is thought to make sense.

\newpage

\noindent
{\bf The level.}

\vspace{-2mm}
\begin{itemize}[leftmargin=*]
\setlength\itemsep{-1pt}

\item $\kappa \,\in\, \NaturalNumbers_{\geq 2}$ denotes a natural number,
which {\it a priori} is the order of an $\mathbb{A}_{\kappa-1}$-type orbi-singularity
(beginning in Prop. \ref{TheBareBProfiniteIntegersTwistOfComplexRationalizedAEquivariantKTheory} below)
and as such necessarily $\geq 2$,
but which eventually is being identified with
the shift of:

\item
$\Level \,\in\, \RationalNumbers$, denoting a rational number that is eventually identified (see Cor. \ref{Degree1ConformalBlockInTEdKTheory}, Thm. \ref{ConformalBlockInTEdKTheory}, Rem. \ref{FractionalLevelWZWModels} below) with the {\it level} of an affine Kac-Moody algebra $\slTwoAffine{\Level}$, such that $\ShiftedLevel$ is its {\it shift} by the dual Coxeter number,
which for $\slTwo$ means:
\vspace{-1mm}
\begin{equation}
  \label{ShiftedLevel}
  \ShiftedLevel
    \,:=\,
  \Level
    \; +\;\;\;
  \underset{
    \mathclap{
    \raisebox{-2pt}{
     \tiny
     \color{darkblue}
     \bf
     \begin{tabular}{c}
       dual Coxeter number
       \\
        of $\slTwo$
     \end{tabular}
    }
    }
  }{
    \underbrace{
      2
    }
  }
  \;\;\;
  \in
  \;
  \NaturalNumbers_{\geq 2}
  \,.
\end{equation}
\vspace{-.4cm}

Eventually $\slTwoAffine{\Level}$ will be identified
with the symmetry algebra of the chiral $\slTwo$-WZW model at level $\Level$, arising as the complexification
\eqref{ComplexificationOfsuTwo}
of the affine $\suTwo$-algebra which, in turn, appears as a summand in the super-conformal algebra on the worldsheet of strings probing an A-type singularity \eqref{ShiftOfLevelViaStringsOnALE};
and the shift \eqref{ShiftedLevel} will arise from a cancellation with a fermionic level.
After the fact, this shift is also the {\it renormalization} of the level (see \cite{AGLR90}\cite{Shifman91}) of the corresponding $\suTwo$ Chern-Simons theory (\cite{Witten89}): the dual Coxeter number constitutes the quantum correction to the classical Chern-Simons level $\Level$.

In fact, more generally:

\item
$ \Denominator \,\in\, \{1, 2,  \cdots, \ShiftedLevel\}$ serves as
a denominator in the more general identification between these two numbers
\begin{equation}
  \label{FractionLevelInCastOfCharacters}
  k
  \,+\,
  2
  \,:=\,
  {\ShiftedLevel}/r
  \;\;\;\;\;\;\;\;
  \Leftrightarrow
  \;\;\;\;\;\;\;\;
  \ShiftedLevel
  \,=\,
  r(\Level + 2)
  \;\in\;
  \NaturalNumbers_{\geq 2}
  \,,
\end{equation}
that {\it a priori} appears in the secondary \TED-K-theory of the punctured plane inside an
$\mathbb{A}_{\ShiftedLevel - 1}$-singularity (see \eqref{TEdKTheoryOfPuncturedPlaneAsDirectSum} below),
but eventually characterizes the (admissible) {\it fractional level} of an $\slTwo$-CFT
(Rem. \ref{FurtherContentSeenInTheTEdKTheory} below).

\end{itemize}

\vspace{0mm}
\noindent
{\bf The groups.}

\vspace{-2mm}
\begin{itemize}[leftmargin=*]
\setlength\itemsep{-1pt}

\item $\CyclicGroup{\kappa} \,\subset\, \CircleGroup$ denotes the cyclic group of order $\kappa$, regarded here as the subgroup
of $\kappa$th roots of unity inside the circle group $\CircleGroup$.

\item $\IrrepOfCyclicGroup{\DualDenominator}{\ShiftedLevel}$ denotes its $\DualDenominator$-th irrep with action:
\begin{equation}
  \label{IrrepsOfCyclicGroup}
    \IrrepOfCyclicGroup
      { \DualDenominator }
      { \ShiftedLevel }
    \,:=\,
    e^{ 2 \pi \ImaginaryUnit (-)\DualDenominator/\ShiftedLevel }
    \!\cdot\!
    \acts \;
    \ComplexNumbers
  \,.
\end{equation}

\item $\CharacterGroup{\CyclicGroup{\kappa}} := \Homs{}{\CyclicGroup{\kappa}}{\CircleGroup}$ denotes its character group.
Notice that this is isomorphic to the original cyclic group, $\CharacterGroup{\CyclicGroup{\kappa}} \,\simeq\, \CyclicGroup{\kappa}$,
but {\it not naturally} so; isomorphisms are in bijection to {\it choices} of generators $g \,\in\, \CyclicGroup{\kappa}$, i.e.,
of elements with $\Order(g) \,=\, \kappa$:
\vspace{-2mm}
$$
  \begin{tikzcd}[row sep=3pt]
    \mathllap{
      \Homs{}{\CyclicGroup{\kappa}}{\CircleGroup}
      \;=\;\;
    }
    \CharacterGroup{\CyclicGroup{\kappa}}
    \ar[r,  "\mathrm{ev}_{g}", "{\sim}"{swap}]
    &
    \CyclicGroup{\kappa}
    \mathrlap{
      \;\;\subset\;
      \CircleGroup
      \,.
    }
  \end{tikzcd}
$$

\vspace{-2mm}
\noindent This little subtlety leads to the following major distinction:

\item $\RationalCircleGroup \;\simeq\; \colimit{\kappa \,\in\, \NaturalNumbers_+} \CyclicGroup{\kappa}$ is the {\it colimit} over the inclusions
$\CyclicGroup{\kappa} \xhookrightarrow{ \cdot n} \CyclicGroup{ n \kappa  }$ for $n \in \NaturalNumbers_+$. We may think of this as the
``rational circle group'' $\RationalCircleGroup \,\xhookrightarrow{\;}\,\CircleGroup$ or the group ``of rational phases'' or ``of any roots of unity'':
$\RationalCircleGroup \,\simeq\, \big\{ \exp(2\pi\ImaginaryUnit \, {\weight/\kappa}) \; \vert \; \weight \,\in\, \Integers, \, \kappa \geq 2  \big\}$.

\item $\ProfiniteIntegers \;\simeq\; \limit{\kappa \,\in\, \NaturalNumbers_+}$ is the {\it limit} over the surjections
$\begin{tikzcd}[column sep=50pt]\CharacterGroup{\CyclicGroup{n \kappa}} \ar[r,->>, "{ \CharacterGroup{(\cdot n)}  =\; \modulo \kappa \;\;\;\; }"] & \CharacterGroup{\CyclicGroup{\kappa}},\end{tikzcd}$known as the {\it profinite integers}. Here we wrote surjections of character groups
to make manifest that the profinite integers constitute the character group of the rational circle group:
$$
  \CharacterGroup{
    \big(
      \RationalCircleGroup
    \big)
  }
  \,=\,
  \CharacterGroup{
    \Big(\;
      \colimit{\kappa \,\in\, \NaturalNumbers_+}
      \CyclicGroup{\kappa}
    \Big)
  }
  \;=\;
  \Homs{\Big}{ \;
    \colimit{\kappa \,\in\, \NaturalNumbers_+}
    \CyclicGroup{\kappa}
  }{\CircleGroup}
  \;\simeq\;
  \limit{n \,\in\, \NaturalNumbers_+}
  \Homs{\big}{
    \CyclicGroup{\kappa}
  }{\CircleGroup}
  \;\simeq\;
  \limit{n \,\in\, \NaturalNumbers_+}
  \CharacterGroup{\CyclicGroup{\kappa}}
  \;\simeq\;
  \ProfiniteIntegers
  \,.
$$

\end{itemize}

\medskip

\noindent
{\bf The $\mathbb{A}$-type orbifold singularity.}

\vspace{-.2cm}
\begin{itemize}[leftmargin=*]
\setlength\itemsep{-1pt}

\item $\Quaternions$ denotes the space of quaternions, to be understood as a complex manifold $\Quaternions \,\simeq\, \ComplexNumbers^2$.

\item $\SpOne \,:=\,S(\Quaternions)\, \,\subset\, \Quaternions$ denotes the multiplicative group of unit quaternions (which, as a Lie group,
is $\simeq\, \SUTwo \,\simeq\, \Spin(3)$).

Its left multiplication action on $\Quaternions$ induces a canonical left action $\CyclicGroup{\kappa} \acts \, \Quaternions$ of
$\CyclicGroup{\kappa} \,\subset\, \CircleGroup \,\subset\, \SpOne$.

\item
$\HomotopyQuotient{\Quaternions}{\CyclicGroup{\kappa}}$ denotes the
complex {\it orbifold} of this action. As such, the singular point $[0] \,\in\, \Quotient{\Quaternions}{\CyclicGroup{\kappa}}$ in the corresponding
ordinary quotient space
is an $\mathbb{A}_{\ShiftedLevel - 1}$-type singularity.

\item $\ComplexManifold$ denotes any complex manifold, and $\ComplexManifold \times \HomotopyQuotient{\Quaternions}{\CyclicGroup{\kappa}}$
its direct product, as orbifolds, with the above complex orbifold.

\item The orbifold obtained from equipping $\ComplexManifold$ with the {\it trivial} action of a cyclic group of order $\kappa$ is denoted
$$
  \HomotopyQuotient
    {\ComplexManifold}
    {\CyclicGroup{\kappa}}
      \,\simeq\,
    \ComplexManifold
      \times
    \HomotopyQuotient
      {\ast}
      {\CyclicGroup{\kappa}}
   \;  \xhookrightarrow{ \qquad } \;
    \ComplexManifold
      \times
    \HomotopyQuotient
      {\Quaternions}
      {\CyclicGroup{\kappa}}\;,
$$
 to be thought of as the extended
$\mathbb{A}_{\kappa-1}$-type singular locus.
\end{itemize}

\medskip

\noindent
{\bf Cohomology.}
All domain spaces we consider here are connected, and we regard them all as pointed spaces, meaning that all (generalized) cohomology groups in the following are {\it reduced}.
For further pointers and discussion in our context see \cite{FSS20CharacterMap}.

\newpage

\noindent
{\bf Differential K-theory.}

\vspace{-.2cm}
\begin{itemize}[leftmargin=*]
\setlength\itemsep{-1pt}

\item $\mathrm{KU}^{ \NumberOfProbeBranes - 1}(\ComplexManifold)$ denotes the complex topological K-cohomology
(e.g. \cite{Karoubi78})
of the manifold $\SmoothManifold$ in degree $\NumberOfProbeBranes - 1$ ($\modulo\, 2$). The shifted degree $\NumberOfProbeBranes$ will eventually be identified with the number of points (probe branes) in a configuration \eqref{ConfigurationSpaceOfPointsInThePlane},
as well as with the degree of conformal blocks \eqref{slTwoConformalBlocksAsQuotientOfSlTwoCoinvariants}.

\item
$\mathrm{KU}^{\bullet}_{\diff}(\SmoothManifold)$,
$\mathrm{KU}^{\bullet}_{\flat}(\SmoothManifold)$
denote, respectively, {\it differential} and {\it flat differential} K-theory (e.g. \cite{BunkeSchick12}).

These cohomology groups arrange into a hexagonal commuting diagram, the lower part of which looks as follows
(\cite[\S 2]{SimonsSUllivan08}\cite[\S 6]{BNV13}), where the bottom sequence of groups is long exact (\cite[\S 7.21]{Karoubi87}\cite[Ex. 3]{Karoubi90}\cite[(16)]{Lott94}):

\vspace{-1mm}

\begin{equation}
\label{TheDifferentialCohomologyHexagonForKTheory}
\adjustbox{raise=-2cm}{
\begin{tikzpicture}
  \clip (-6.5,-1.9) rectangle (6.1,1.1);
  \draw (0,0) node {
  \begin{tikzcd}[column sep=15pt]
    &
    {}
    \ar[rr, lightgray]
    \ar[
      dr,
      gray,
      shorten >=-0
      ]
    {}
    &&
    {}
    \ar[
      dr,
      gray,
      shorten >=-0pt
    ]
    &
    \\
    \mathrm{KU}
      ^{\NumberOfProbeBranes -2}
      (\SmoothManifold;\, \ComplexNumbers)
  \ar[
    dr,
    shorten <=-0pt,
    "{
        \mbox{
          \tiny
          \color{greenii}
          \bf
          \def\arraystretch{.9}
          \begin{tabular}{c}
            secondary
            \\
            Chern character
          \end{tabular}
        }
    }"{swap, sloped, pos=.35}
  ]
  \ar[
    ur,
    gray
  ]
  \ar[
    rr,
    "{
      \mbox{
        \tiny
        \color{greenii}
        \bf
        \def\arraystretch{.9}
        \begin{tabular}{c}
        \end{tabular}
      }
    }"{pos=.4}
  ]
  &&
    \overset{
      {
      \mathclap{
      \raisebox{1pt}{
        \tiny
        \color{darkblue}
        \bf
        \def\arraystratch{.9}
        \begin{tabular}{c}
        differential K-theory
        \end{tabular}
      }
      }
      }
    }{
      \mathrm{KU}^{
        \NumberOfProbeBranes
          -
        1
      }
      _{\mathrm{diff}}
      \big(
        \SmoothManifold
      \big)
    }
    \ar[
      rr,
      "{
        \mbox{
          \tiny
          \color{greenii}
          \bf
          \def\arraystretch{.9}
          \begin{tabular}{c}
          \end{tabular}
        }
      }"{pos=.6}
    ]
    \ar[
      ur,
      gray,
      shorten <=-0
    ]
    \ar[
      dr,
      "{
        \mbox{
          \tiny
          \color{greenii}
          \bf
          forget
        }
      }"{swap, sloped},
      shorten <=-0
    ]
    &&
    \mathrm{KU}^{\NumberOfProbeBranes - 1}
    (\SmoothManifold;\, \ComplexNumbers)\;.
    \\
    &
    \underset{
      \raisebox{-1pt}{
        \tiny
        \color{darkblue}
        \bf
        flat K-theory
      }
    }{
    \mathrm{KU}
      ^{
        \NumberOfProbeBranes
          -
        1
      }
      _{\flat}
      (
        \SmoothManifold
      )
    }
    \ar[
      ur,
      shorten >=-0pt,
      "{
        \mbox{
          \tiny
          \color{greenii}
          \bf
          include
        }
      }"{swap, sloped}
    ]
    \ar[
      rr,
      "{
        \mbox{
          \tiny
          \color{greenii}
          \bf
          \def\arraystretch{.9}
          \begin{tabular}{c}
          \end{tabular}
        }
      }"{swap, yshift=-0pt, pos=.4}
    ]
    &&
    \underset{
      \mathclap{
      \raisebox{-2pt}{
        \tiny
        \color{darkblue}
        \bf
        \begin{tabular}{c}
          underlying
          plain K-theory
        \end{tabular}
      }
      }
    }{
    \mathrm{KU}
      ^{
        \NumberOfProbeBranes
          -
        1
      }
    (
      \SmoothManifold
    )
    }
        \ar[
      ur,
      shorten >=-0pt,
      "{
        \mbox{
          \tiny
          \color{greenii}
          \bf
          \def\arraystretch{.9}
          \begin{tabular}{c}
            Chern
            \\
            character
          \end{tabular}
        }
      }"{sloped, swap, pos=.5}
    ]
  \end{tikzcd}
  };
  \end{tikzpicture}
  }
\end{equation}

\end{itemize}

\noindent
Here the items on the left and right denote
$\ComplexNumbers$-rationalized
K-theory (i.e., ``with complex coefficients'', e.g. \cite[Def. 4]{Lott94}),
which is  equivalently even-periodic ordinary cohomology (cf. \cite{GS-Higher}).
For making contact to the hypergeometric formalism, it is convenient to represent this,
and its variants below, via de Rham cohomology (cf. Rem. \ref{TwistedCohomologyOverSteinManifolds} below):
$$
  \mathrm{KU}^{\NumberOfProbeBranes}
  (\SmoothManifold;\, \ComplexNumbers)
  \;\;
  :=
  \;\;
    {
    \underset{
      \mathclap{
        { d \,\in\, \Integers }
      }
    }{\bigoplus}
    \;
    }
    \overset{
      \mathclap{
      \hspace{-10pt}
      \raisebox{1pt}{
        \tiny
        \color{darkblue}
        \bf
        \def\arraystretch{.9}
        \begin{tabular}{c}
          even-periodic
          de Rham cohomology
        \end{tabular}
      }
      }
    }{
    H^{
      \NumberOfProbeBranes
      +
      2d
    }
    \big(
      \DeRhamComplex{}
        { \SmoothManifold;
        \, \ComplexNumbers}
      ,\,
      \DeRhamDifferential
    }
    \big).
$$

\noindent
{\bf Twisted Equivariant Differential K-theory.}

\vspace{-.2cm}
\begin{itemize}[leftmargin=*]
\setlength\itemsep{-1pt}

\item
$\mathrm{KU}^{\NumberOfProbeBranes - 1 }
\big(
  \HomotopyQuotient
    {\TopologicalSpace}
    {\CyclicGroup{\kappa}})
$
denotes the $\CyclicGroup{\kappa}$-{\it equivariant} K-theory of $\CyclicGroup{\kappa} \acts \, \TopologicalSpace$ (e.g. \cite{Greenlees05}\cite{SzaboValentino07}).
The notation suggestively follows that for orbifold K-theory (\cite{LupercioUribe01}\cite{AdemRuan01}), which coincides with equivariant K-theory for our case of global quotient orbifolds.

\item
$
  \FlatConnectionForm
  \,\in\,
  \DifferentialForms{}{1}{\SmoothManifold}
      \vert_{
    \scalebox{.8}{$
      \DeRhamDifferential = 0
    $}
   }
$
denotes a closed differential 1-form
 such that its holonomy in $\RealNumbers \twoheadrightarrow \CircleGroup$ is in fact in $\CyclicGroup{\kappa} \xhookrightarrow{\;} \CircleGroup$,
hence such that its class is that of a flat connection on a complex line bundle whose holonomy is in $\CyclicGroup{\kappa}$:
\vspace{-2mm}
\begin{equation}
\label{ClassofTheFlatConnectionOneForm}
\hspace{-3cm}
\begin{tikzcd}[column sep=70pt]
\overset{
  \mathclap{
 \raisebox{4pt}{
   \tiny
   \color{darkblue}
   \bf
   \def\arraystretch{.9}
   \begin{tabular}{c}
    a flat connection on
    \\
    a complex line bundle
   \end{tabular}
 }
 }
}{
  \big\{[\FlatConnectionForm]\big\}
}
\ar[
  r,
  "{
    \mbox{
      \tiny
      \color{orangeii}
      \bf
      \begin{tabular}{c}
        with holonomy in
        \\
        $\kappa$th roots of unity
      \end{tabular}
    }
  }"{swap}
]
\ar[
  d,
  "{
    \mbox{
      \tiny
      \color{orangeii}
      \bf
      \begin{tabular}{c}
        with globally defined
        \\
        connection 1-form
      \end{tabular}
    }
  }"{swap, xshift=5pt}
]
&
H^1(\SmoothManifold;\, \CyclicGroup{\kappa})
\mathrlap{
  \; \simeq\,
H^1(
  { \SmoothManifold }
  ;\,
  { \CharacterGroup{\CyclicGroup{\kappa}} }
)
}
\ar[d]
\\
H^1_{\DeRham}(\SmoothManifold)
\ar[
  r,
  "{
    \mbox{
      \tiny
      \color{greenii}
      \bf
      \def\arraystretch{.9}
      \begin{tabular}{c}
        de Rham theorem
        \\
        mod $\Integers$
      \end{tabular}
    }
  }"{swap}
]
&
\underset{
  \mathclap{
  \raisebox{-3pt}{
    \tiny
    \color{darkblue}
    \bf
    \begin{tabular}{c}
      equivalence classes of
      \\
      all flat $\CircleGroup$-connections
    \end{tabular}
  }
  }
}{
  H^1(\SmoothManifold;\, \CircleGroup)
}
\end{tikzcd}
\end{equation}

\vspace{-.4cm}
\item
$
  \mathrm{KU}
    ^{
      \NumberOfProbeBranes
      +
      1
      +
      [\FlatConnectionForm]
    }
    _{(\diff)}
  \big(
    \HomotopyQuotient
      { \TopologicalSpace }
      { \CyclicGroup{\kappa} }
  \big)
$
denotes $[\FlatConnectionForm]$-twisted $\CyclicGroup{\kappa}$-equivariant
(differential) K-cohomology, according to the following:
\end{itemize}

\begin{proposition}[Secondary Chern characters in \TED-K-theory of an $\mathbb{A}$-type singularity]
\label{TheBareBProfiniteIntegersTwistOfComplexRationalizedAEquivariantKTheory}
The closed differential 1-forms $\FlatConnectionForm$
\eqref{ClassofTheFlatConnectionOneForm}
twist
the $\CyclicGroup{\kappa}$-equivariant K-cohomology
of
smooth manifolds $\SmoothManifold$ inside
$\CyclicGroup{\ShiftedLevel}$-singularities
such that the
corresponding
(secondary) twisted equivariant Chern character classes
fill the direct sum
over $r \,\in\, \{0, \cdots, \kappa-1 \}$
of the
$r \cdot \FlatConnectionForm$-twisted de Rham cohomology groups of $\SmoothManifold$;
in that the
 $B \CharacterGroup{\CyclicGroup{\ShiftedLevel}}$-twisted $\CyclicGroup{\ShiftedLevel}$-equivariant enhancement of the differential cohomology hexagon
 \eqref{TheDifferentialCohomologyHexagonForKTheory}
looks as follows:

\vspace{-.2cm}
$$
  \label{TheSpecializedTwistedEquivariantCharacterMap}
  \begin{tikzpicture}
  \clip (-8.3,-2.4) rectangle (9.3,1.5);
  \draw (0,0) node {
  \begin{tikzcd}[column sep=-28pt]
    &
    {}
    \ar[rr, lightgray]
    \ar[
      dr,
      gray,
      shorten >=-16
      ]
    {}
    &[-4pt]&[-10pt]
    {}
    \ar[dr, gray]
    \\
    \mathllap{
    \underset{
      \mathclap{
        \scalebox{.61}{$
        \begin{array}{c}
          { d \,\in\, \Integers }
          \\
          { 1 \,\leq\, \Denominator \,\leq\, \ShiftedLevel  }
        \end{array}
        $}
      }
    }{\bigoplus}
    \;
    }
    \overset{
      \mathclap{
      \hspace{-40pt}
      \raisebox{3pt}{
        \tiny
        \color{darkblue}
        \bf
        \def\arraystretch{.9}
        \begin{tabular}{c}
          even-periodic
          de Rham cohomology
        \end{tabular}
      }
      }
    }{
    H^{
      \NumberOfProbeBranes
      +
      2d
    }
    \big(
      \DeRhamComplex{}
        { \SmoothManifold; \, \ComplexNumbers}
      ,\,
      \DeRhamDifferential
    }
    \overset{
      \mathclap{
      \hspace{-13pt}
      \raisebox{+6pt}{
        \tiny
        \color{orangeii}
        \bf
        \begin{tabular}{c}
          twisted by tensor powers of
          \\
           the flat complex line bundle
        \end{tabular}
      }
      }
    }{
      +
      \;
      \Denominator
      \cdot
      \FlatConnectionForm \wedge
    }
    \big)
  \ar[
    dr,
    shorten <=-20pt,
    "{
      \mbox{
        \tiny
        \color{greenii}
        \bf
        \def\arraystretch{.9}
        \begin{tabular}{c}
          secondary
          \\
          Chern character
        \end{tabular}
      }
    }"{sloped, swap, pos=.0}
  ]
  \ar[
    ur,
    gray
  ]
  \ar[
    rr,
    "{
      \mbox{
        \tiny
        \color{greenii}
        \bf
        \def\arraystretch{.9}
        \begin{tabular}{c}
        \end{tabular}
      }
    }"
  ]
  &&
    \overset{
      \mathclap{
      \raisebox{3pt}{
      \scalebox{.65}{
        \color{darkblue}
        \bf
        \def\arraystratch{.9}
        \begin{tabular}{c}
        \color{orangeii}
        $B \CharacterGroup{\CyclicGroup{\kappa}}$-twisted
        \\
        \color{orangeii}
        $\CyclicGroup{\kappa}$-equivariant
        \\
        differential K-theory
        \end{tabular}
      }
      }
      }
    }{
      \mathrm{KU}^{
        \NumberOfProbeBranes
          -
        1
          +
        [\FlatConnectionForm]
      }
      _{\mathrm{diff}}
      \big(
        \underset{
          \mathclap{
          \raisebox{-5pt}{
            \scalebox{.59}{
            \color{darkblue}
            \bf
            \def\arraystretch{.9}
            \begin{tabular}{c}
              smooth manifold inside
              \\
              $\CyclicGroup{\kappa}$-orbi-singularity
            \end{tabular}
            }
          }
          }
        }{
        \SmoothManifold
        \times
        \HomotopyQuotient
          { \ast }
          { \CyclicGroup{\kappa} }
        }
      \big)
    }
    \ar[
      rr,
      "{
        \mbox{
          \tiny
          \color{greenii}
          \bf
          \def\arraystretch{.9}
          \begin{tabular}{c}
          \end{tabular}
        }
      }"{xshift=-3pt}
    ]
    \ar[
      ur,
      gray,
      shorten <=-20
    ]
    \ar[
      dr,
      "{
      }",
      shorten <=-5,
      "{
        \mbox{
          \tiny
          \color{greenii}
          \bf
          forget
        }
      }"{swap, sloped, pos=.3}
    ]
    &&
    \underset{
      \mathclap{
        \scalebox{.61}{$
        \begin{array}{c}
          { d \,\in\, \Integers }
          \\
          { 1 \,\leq\, \Denominator \,\leq\, \ShiftedLevel  }
        \end{array}
        $}
      }
    }{\bigoplus}
    \;
    \overset{
      \mathclap{
      \hspace{+27pt}
      \raisebox{3pt}{
        \tiny
        \color{darkblue}
        \bf
        \def\arraystretch{.9}
        \begin{tabular}{c}
          even-periodic
          twisted
          de Rham cohomology
          in opposite degree
        \end{tabular}
      }
      }
    }{
    H^{
      \NumberOfProbeBranes
        -
      1
        +
      2d
    }
    \big(
      \DeRhamComplex{}
       { \SmoothManifold; \, \ComplexNumbers }
      ,\,
      \DeRhamDifferential
    }
    \mathrlap{
      +
      \;
      \Denominator
      \cdot
      \FlatConnectionForm \wedge
    \big).
    }
    \\
    &
    \underset{
      \raisebox{-1pt}{
        \tiny
        \color{darkblue}
        \bf
        flat \TED-K-theory
      }
    }{
    \mathrm{KU}
      ^{
        \NumberOfProbeBranes
          -
        1
          +
        [\FlatConnectionForm]
      }
      _{\flat}
      (
        \SmoothManifold
        \times
        \HomotopyQuotient
          {\ast}
          {\CyclicGroup{\ShiftedLevel}}
      )
    }
    \ar[
      ur,
      shorten >=-18pt,
      "{
        \mbox{
          \tiny
          \color{greenii}
          \bf
          include
        }
      }"{swap, sloped, pos=.7}
    ]
    \ar[
      rr,
      "{
        \mbox{
          \tiny
          \color{greenii}
          \bf
          \def\arraystretch{.9}
          \begin{tabular}{c}
          \end{tabular}
        }
      }"{swap, yshift=-2pt, pos=.4}
    ]
    &&
    \underset{
      \mathclap{
      \raisebox{-2pt}{
        \tiny
        \color{darkblue}
        \bf
        \begin{tabular}{c}
          underlying
          twisted equivariant K-theory
        \end{tabular}
      }
      }
    }{
    \mathrm{KU}
      ^{
        \NumberOfProbeBranes
          -
        1
          +
        [\FlatConnectionForm]
      }
    (
      \SmoothManifold
      \times
      \HomotopyQuotient
        {\ast}
        {\CyclicGroup{\ShiftedLevel}}
    )
    }
    \ar[
      ur,
      shorten >=-16pt,
      "{
        \mbox{
          \tiny
          \color{greenii}
          \bf
          Chern character
        }
      }"{sloped, swap, pos=.7}
    ]
  \end{tikzcd}
  };
  \end{tikzpicture}
$$
\end{proposition}
\noindent {\it Outline proof.}
This follows essentially as a highly specialized case of a general expression for the twisted equivariant Chern character which may be extracted from
\cite[Def. 3.10, Thm. 1.1]{TuXu06}\cite[Def. 3.6, Thm. 3.9]{FreedHopkinsTeleman02ComplexCoefficients};
in fact it is that special case which isolates exactly the peculiar extra $B \CharacterGroup{\CyclicGroup{\ShiftedLevel}}$-twisting of equivariant K-theory
that appears on $\CyclicGroup{\ShiftedLevel}$-fixed loci.
Below in \cref{OneTwistedCohomologyAsTwistedEquivariantKTheory} we outline a transparent proof,
the formalities of which are relegated to \cite{SS22TEC}.
\hfill $\square$

\begin{remark}[Twisted de Rham cohomology over complex Stein manifolds]
\label{TwistedCohomologyOverSteinManifolds}
$\,$

\noindent {\bf (i)} When $\SmoothManifold$ in Prop. \ref{TheBareBProfiniteIntegersTwistOfComplexRationalizedAEquivariantKTheory} happens to be a complex manifold which is Stein (see \cite{GrauertRemmert04})
and $\FlatConnectionForm$ is a holomorphic form representative, then the twisted cohomology groups are equivalently twisted {\it holomorphic} de Rham cohomology groups:

\vspace{-.4cm}
\begin{align}
  \label{HolomorphicDeRhamComplexOverSteinManifolds}
  \mbox{\small $\SmoothManifold$ is complex Stein domain}
  \;\;\;
  \Rightarrow
  \;\;\;
  H^{\NumberOfProbeBranes + [\omega_1]}
  \big(
    \SmoothManifold
    ;\,
    \ComplexNumbers
  \big)
  \,
 & =
  \,
  H^{\NumberOfProbeBranes}
  \big(
    \DeRhamComplex{\big}
      { \SmoothManifold ;\, \ComplexNumbers }
    ,\,
    \DeRhamDifferential +
    \Denominator
      \cdot
    \FlatConnectionForm
  \big)
 \qquad   \qquad
 \raisebox{2pt}{
      \tiny
      \color{darkblue}
      \bf
      twisted
      ordinary/de Rham
      cohomology
    }
    \nonumber
  \\
&  \simeq
  \;\;
       H^{\NumberOfProbeBranes}
  \big(
    \HolomorphicDeRhamComplex{\big}
      { \SmoothManifold;\, \ComplexNumbers }
    ,\,
    \DolbeaultDifferential +
    \Denominator
      \cdot
    \FlatConnectionForm
    \wedge
  \big)
  \quad
   \raisebox{2pt}{
      \tiny
      \color{darkblue}
      \bf
      twisted
      {\color{orangeii}holomorphic}
      de Rham cohomology.
    }
  \end{align}
\noindent {\bf (ii)} This follows, as in the case over affines
  \cite[Cor. 6.3]{Deligne70} (review in \cite[p. 171]{LibgoberYuzvinsky00}\cite[Cor. 5.4]{Voisin03}),
  using that sheaves of positive-degree holomorphic forms are still acyclic over Stein manifolds.

\noindent {\bf (iii)} Examples of Stein manifolds include punctured Riemann surfaces \cite{BehnkeStein47} and their configuration spaces of points \cite{DocquierGrauert60}. These are the domains to which we turn next.
\end{remark}

\medskip

\noindent
{\bf The punctured Riemann surface.}

\vspace{-.3cm}
\begin{itemize}[leftmargin=*]
\setlength\itemsep{-1pt}

\item
$\NumberOfPunctures \,\in\, \NaturalNumbers$ denotes a natural number, counting punctures in the complex plane,
eventually interpreted as the number of {\it defect branes}.

\item Denoted by
$$
\Sigma^2
  \,:=\,
\ComplexPlane
  \mysetminus
\{ z_1, \cdots, z_{N} \}
\;\simeq\;
\RiemannSphere
  \mysetminus
\{  z_1, \cdots, z_{N}, \infty \}
$$
this $\NumberOfPunctures$-punctured complex plane, equivalently the $(\NumberOfPunctures+1)$-punctured Riemann sphere,
equipped with its canonical holomorphic coordinate function $z : \Sigma^2 \xhookrightarrow{\;} \ComplexPlane$.

The ``incoming'' punctures $\vec z \,:=\, (z_1, \cdots, z_{\NumberOfPunctures})$ are taken to be ordered (and will carry labels $\weight_I$ in a moment) and the ``outgoing'' puncture of number $N+1$ is taken to be at $z = \infty$.

\vspace{-.5cm}
\begin{center}
\hypertarget{ThePuncturedPlane}{}
\begin{tikzpicture}

  \shade[right color=lightgray, left color=white]
    (3,-3)
      --
      node[above, yshift=-1pt, sloped]{
        \scalebox{.7}{
          \color{darkblue}
          \bf
          transverse plane
        }
      }
    (-1,-1)
      --
        (-1.21,1)
      --
    (2.3,3);

  \draw[dashed]
    (3,-3)
      --
    (-1,-1)
      --
    (-1.21,1)
      --
    (2.3,3)
      --
    (3,-3);

  \begin{scope}[rotate=(+8)]
   \draw[dashed]
     (1.5,-1)
     ellipse
     ({.2*1.85} and {.37*1.85});
   \begin{scope}[
     shift={(1.5-.2,{-1+.37*1.85-.1})}
   ]
     \draw[->, -Latex]
       (0,0)
       to
       (180+37:0.01);
   \end{scope}
   \begin{scope}[
     shift={(1.5+.2,{-1-.37*1.85+.1})}
   ]
     \draw[->, -Latex]
       (0,0)
       to
       (+37:0.01);
   \end{scope}
   \begin{scope}[shift={(1.5,-1)}]
     \draw (.43,.65) node
     { \scalebox{.8}{$\FlatConnectionForm$} };
  \end{scope}
  \draw[fill=white, draw=gray]
    (1.5,-1)
    ellipse
    ({.2*.3} and {.37*.3});
  \draw[line width=1.1]
   (1.5,-1)
   to node[above, yshift=-3pt, pos=.85]{
     \;\;\;\;\;\;\;\;\;\;\;\;\;
     \rotatebox[origin=c]{7}
     {
     \scalebox{.7}{
     \color{orangeii}
     \bf
     \colorbox{white}{defect} brane
     }
     }
   }
   (-2.2,-1);
  \draw[
    line width=1.1
  ]
   (1.5+1.2,-1)
   to
   (3.5,-1);
  \draw[
    line width=1.1,
    densely dashed
  ]
   (3.5,-1)
   to
   (4,-1);
  \end{scope}

  \begin{scope}[shift={(-.2,1.4)}, scale=(.96)]
  \begin{scope}[rotate=(+8)]
  \draw[dashed]
    (1.5,-1)
    ellipse
    (.2 and .37);
  \draw[fill=white, draw=gray]
    (1.5,-1)
    ellipse
    ({.2*.3} and {.37*.3});
  \draw[line width=1.1]
   (1.5,-1)
   to
   (-2.3,-1);
  \draw[line width=1.1]
   (1.5+1.35,-1)
   to
   (3.6,-1);
  \draw[
    line width=1.1,
    densely dashed
  ]
   (3.6,-1)
   to
   (4.1,-1);
  \end{scope}
  \end{scope}

  \begin{scope}[shift={(-1,.5)}, scale=(.7)]
  \begin{scope}[rotate=(+8)]
  \draw[dashed]
    (1.5,-1)
    ellipse
    (.2 and .32);
  \draw[fill=white, draw=gray]
    (1.5,-1)
    ellipse
    ({.2*.3} and {.32*.3});
  \draw
   (1.5,-1)
   to
   (-1.8,-1);
  \draw
    (5.23,-1)
    to
    (6.4-.6,-1);
  \draw[densely dashed]
    (6.4-.6,-1)
    to
    (6.4,-1);
  \end{scope}
  \end{scope}

\draw (2.41,-2.3) node
 {
  \scalebox{1}{
    $\Sigma^2$
  }
 };

\draw (3.05,-1.4) node
 {
  \scalebox{1}{
    $\infty$
  }
 };

\draw (1.73,-1.06) node
 {
  \scalebox{.8}{
    $z_{{}_{I}}$
  }
 };

\begin{scope}
[ shift={(-2,-.55)}, rotate=-82.2  ]

 \begin{scope}[shift={(0,-.15)}]

  \draw
    (-.2,.4)
    to
    (-.2,-2);

  \draw[
    white,
    line width=1.1+1.9
  ]
    (-.73,0)
    .. controls (-.73,-.5) and (+.73-.4,-.5) ..
    (+.73-.4,-1);
  \draw[
    line width=1.1
  ]
    (-.73+.01,0)
    .. controls (-.73+.01,-.5) and (+.73-.4,-.5) ..
    (+.73-.4,-1);

  \draw[
    white,
    line width=1.1+1.9
  ]
    (+.73-.1,0)
    .. controls (+.73,-.5) and (-.73+.4,-.5) ..
    (-.73+.4,-1);
  \draw[
    line width=1.1
  ]
    (+.73,0+.03)
    .. controls (+.73,-.5) and (-.73+.4,-.5) ..
    (-.73+.4,-1);

  \draw[
    line width=1.1+1.9,
    white
  ]
    (-.73+.4,-1)
    .. controls (-.73+.4,-1.5) and (+.73,-1.5) ..
    (+.73,-2);
  \draw[
    line width=1.1
  ]
    (-.73+.4,-1)
    .. controls (-.73+.4,-1.5) and (+.73,-1.5) ..
    (+.73,-2);

  \draw[
    white,
    line width=1.1+1.9
  ]
    (+.73-.4,-1)
    .. controls (+.73-.4,-1.5) and (-.73,-1.5) ..
    (-.73,-2);
  \draw[
    line width=1.1
  ]
    (+.73-.4,-1)
    .. controls (+.73-.4,-1.5) and (-.73,-1.5) ..
    (-.73,-2);

 \draw[
   densely dashed
 ]
   (-.2,-2)
   to
   (-.2,-2.5);
 \draw[
   line width=1.1,
   densely dashed
 ]
   (-.73,-2)
   to
   (-.73,-2.5);
 \draw[
   line width=1.1,
   densely dashed
 ]
   (+.73,-2)
   to
   (+.73,-2.5);

  \end{scope}
\end{scope}

\end{tikzpicture}

{\footnotesize
  \begin{minipage}{14cm}
  {\bf Figure 1.}
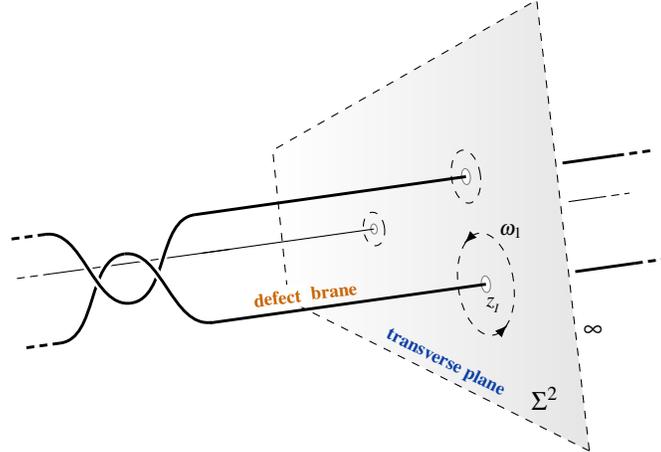
  The punctured complex plane.
  Compare
  the defect brane configurations
  in \hyperlink{FigureD7BraneConfiguration}{\it Figure 2}
  and \hyperlink{MBraneConfigurations}{\it Figure 3},
  where the solid lines correspond to brane worldvolumes.
  Concerning the physical meaning of the braiding that is indicated on the left above, see
  Rem. \ref{ExoticDefectBranesAndAnyonStatistics}
  and
  Footn. \ref{PhysicsOfBraiding} below.
 \end{minipage}
}

\end{center}

\vspace{-.1cm}
\begin{itemize}

\item
Since $H^2\big( \Sigma^2;\, \Integers\big) \,=\, 0$, all complex line bundles on $\Sigma^2$ are trivializable, and hence every element in $H^1\big( \Sigma^2;\, \CyclicGroup{\kappa} \big)$ is represented by a class $[\FlatConnectionForm]$
as in \eqref{ClassofTheFlatConnectionOneForm}.

\vspace{-.1cm}
\item Since $\Sigma^2$ is connected and  the abelianization of its fundamental group
is
$
\big(
  \pi_1
  (
    \Sigma^2
  )
  \big)^{\mathrm{ab}}
  \,\simeq\,
  \underset{
    \scalebox{0.5}{$ 1 \leq i \leq N$}
  }{\prod}
  \,
  \Integers
$, the Hurewicz theorem gives that
\vspace{-1mm}
\begin{equation}
  \label{AbelianizationOfFundamentalGroupOfSigma}
 H^1
 \big(
   \Sigma^2;
   \,
   \CharacterGroup{\CyclicGroup{\kappa}}
 \big)
 \;\simeq\;
 \Homs{\Big}
   { \;
     \underset{
      \scalebox{0.6}{$ 1 \leq i \leq \NumberOfPunctures $}
     }{\prod}
     \Integers
   }
   { \CharacterGroup{\CyclicGroup{\kappa}} }
 \;\simeq\;
  \underset{ \;
   \scalebox{0.6}{$ 1 \leq i \leq N$}
  }{\prod}
  \,
   \CharacterGroup{\CyclicGroup{\kappa}}
 \,.
\end{equation}

\vspace{-2mm}
\item
Therefore, flat connections $[\FlatConnectionForm]$
(see \eqref{ClassofTheFlatConnectionOneForm})
on $\Sigma^2$
may
be specified
by $N$-tuples of integers in
$\underset{ \scalebox{0.6}{$1 \leq i \leq N$}}{\prod} \, \Integers
\; \twoheadrightarrow \underset{ \scalebox{0.6}{$1 \leq i \leq N$}}{\prod} \, \CharacterGroup{\CyclicGroup{\kappa}}$.

These are going to be called {\it weights}:

\end{itemize}

\end{itemize}

\newpage

\noindent
{\bf The weights.}

\vspace{-.2cm}
\begin{itemize}[leftmargin=*]
\setlength\itemsep{-1pt}

\item
Such an $N$-tuple of integers is denoted
$$
\vec{\weight}
:=
(\weight_{{}_I})_{{}_{1 \leq I \leq N}} \;\; \in \;\;\; \underset{\mathclap{\scalebox{0.6}{$1 \leq I \leq N$}}}{\prod} \;\; \Integers
$$
and will label the $N$ punctures in the complex plane and eventually be identified with
{\it highest weights} that label $\mathfrak{sl}(2)$-representations.
Regarding these tuples as encoding connections $[\FlatConnectionForm]$
(see \eqref{ClassofTheFlatConnectionOneForm})
for any shifted level $\kappa$
means,
by \eqref{AbelianizationOfFundamentalGroupOfSigma},
to regard them as
$N$-tuples of profinite integers:
\vspace{-3mm}
\begin{equation}
  \label{WeightsMatterModuloShiftedLevel}
  \begin{tikzcd}[row sep=-5pt]
    \overset{
      \mathclap{
      \raisebox{3pt}{
        \tiny
        \color{darkblue}
        \bf
        weights
      }
      }
    }{
      \Integers
    }
    \;
    \ar[r, hook]
    &
    \quad
    \overset{
      \mathclap{
      \raisebox{1pt}{
        \tiny
        \color{darkblue}
        \bf
        \begin{tabular}{c}
          profinite
          weights
        \end{tabular}
      }
      }
    }{
      \ProfiniteIntegers
    }
    \quad
    \ar[r, hook]
    &
    \overset{
      \raisebox{+2pt}{
        \tiny
        \color{darkblue}
        \bf
        levels
      }
    }{
      \underset{
        \kappa
        \,\in\,
        \NaturalNumbers_+
      }
      {
        \prod
      }
    }
    \;
    \overset{
      \mathclap{
      \raisebox{2pt}{
        \tiny
        \color{darkblue}
        \bf
        phases
      }
      }
    }{
      \CharacterGroup
        {\CyclicGroup{\kappa}}
    }\;.
    \\
 \scalebox{0.8}{$  \weight  $}
      &\longmapsto&
  \scalebox{0.8}{$    \big(
      [\weight \modulo \kappa]
    \big)_{
      \kappa
      \,\in\,
      \NaturalNumbers_+
    }
    $}
  \end{tikzcd}
\end{equation}

\vspace{-.2cm}
\item
$
  \IncomingWeight
   \;:=\;\;\;
   \underset{
     \mathclap{
      \scalebox{0.5}{$
        1 \leq I \leq N
      $}
     }}
     {\sum}
  \;\;
  \weight_{{}_I}
$
denotes the sum of the ``incoming'' weights.

\vspace{-.0cm}
\item
$
  \FlatConnectionForm
  (
     \vec{\weight}
    ,
    \kappa
  )
    \;\in\;
  \HolomorphicDifferentialForms{\big}{1}
  {
    \ComplexPlane
    \setminus
    \{ z_1, \cdots, z_N \}
  }
$
denotes the corresponding canonical holomorphic
flat connection 1-form
with holonomy
$
  [ \weight_{{}_I} / \kappa ]$
  around the $I$th
  incoming puncture:
\begin{equation}
  \label{TheMasterForm}
  \overset{
    \mathclap{
    \raisebox{4pt}{
      \tiny
      \color{darkblue}
      \bf
      \def\arraystretch{.9}
      \begin{tabular}{c}
        flat holomorphic
        \\
        connection 1-form
      \end{tabular}
    }
    }
  }{
   \FlatConnectionForm
   (
      \vec{\weight}
     ,
     \kappa
   )
  }
  \quad:=\;
  \underset{
    \scalebox{0.6}{$
      1 \,\leq\, I \,\leq\, N
    $}
  }{\sum}
  \!\!
  -
  \frac{\weight_{{}_I}}{\kappa}
  \frac{ \Differential z }{(z-z_{{}_I})}
  \;\;\;\;
  =
  \;\;\;\;
  \underset{
   \scalebox{0.6}{$
     1 \,\leq\, I \,\leq\, N
   $}
  }{\sum}
  \!\!
  -
  \frac{\weight_{{}_I}}{\kappa}
  \,
  \Differential \log(z - z_{{}_I})
  \;\;\;\;\;
  \in
  \;
  \HolomorphicDifferentialForms{\big}{1}
  {
    \ComplexPlane
    \setminus
    \{ z_1, \cdots, z_N \}
  }
  \,.
\end{equation}

\vspace{-.2cm}
When
regarded as a connection on the $(\NumberOfPunctures+1)$-punctured Riemann sphere,
this has holonomy $-\sum_I -  {\weight_{{}_I}} \,=\, + {\IncomingWeight}$ around the outgoing puncture
at $\infty$, which we may notationally make explicit by writing (using $1/(z - \infty) = 0$):
\vspace{-2mm}
\begin{equation}
  \label{FlatConnectionWithExplicitOutgoingTerm}
  \FlatConnectionForm( \vec{\weight}, \ShiftedLevel)
  \;\;
  \;=\;
  \overset{
    \mathclap{
    \raisebox{2pt}{
      \tiny
      \color{darkblue}
      \bf
      \begin{tabular}{c}
        holonomy contributions
        \\
        around
        {\color{orangeii}incoming}
        punctures
      \end{tabular}
    }
    }
  }{
  \underset{
  \scalebox{0.6}{$ 1 \leq I \leq \NumberOfPunctures $}
  }{\sum}
  \!\!
  -
  \frac{\weight_{{}_I}}{\ShiftedLevel}
  \frac{ \Differential z }{(z - z_{{}_I})}
  }
  \;\;\;+\;\;\;
  \overset{
    \mathclap{
    \raisebox{2pt}{
      \tiny
      \color{darkblue}
      \bf
      \begin{tabular}{c}
       holonomy contribution
       \\
       around
       {\color{orangeii}outgoing}
       puncture
      \end{tabular}
    }
    }
  }{
  \frac{\IncomingWeight}{\ShiftedLevel}
  \frac{\Differential z}{(z - \infty)}
  }
  \;\;\;
  \in
  \;
  \HolomorphicDifferentialForms{\big}{1}
  {
    \RiemannSphere
    \setminus
    \{
      z_1, \cdots, z_{\NumberOfPunctures},
      \infty
    \}
  }
  \,.
\end{equation}

\vspace{-2mm}
\item $\widehat{\Sigma^2} \xrightarrow{p_\Sigma} \Sigma^2$ denotes the universal cover of $\Sigma^2$, i.e.,
the complex manifold isomorphic to the space of homotopy classes $[\gamma]$ of continuous paths
$\widehat{z}_{\gamma} : [0,1] \to \Sigma^2$, all based at a fixed start point $\gamma(0) \,\in\, \Sigma^2$
and modulo homotopies that fix this and their respective endpoint $\gamma(1)$.

Picking any open disk $U \subset \Sigma^2$ and a lift $\widehat{U} \subset \widehat{\Sigma^2}$ through $p$,
we write $\hat z \,:\, \widehat{\Sigma^2} \xrightarrow{\;} \ComplexNumbers$ for the unique holomorphic
function which agrees with $z := p^\ast z$ on $\widehat{U}$. This yields:

\item $\MasterFunction : \widehat{\Sigma^2} \xrightarrow{\;} \ComplexNumbers$ denotes a holomorphic function
on the universal cover $\widehat{\Sigma^2}$ given by any branch of the following expression
(cf. \cite[(2.1)]{SchechtmanVarchenko89}:
\begin{equation}
  \label{TheMasterFunction}
  \MasterFunction(
    \widehat{z}\,
  )
  \;:=\;
  \MasterFunction
  ( \vec{\weight}, \kappa)
  (
    \widehat{z} \,
  )
  \;:=\quad
  \underset{
    \mathclap{
     \scalebox{0.6}{$
       1 \leq I \leq \NumberOfPunctures
     $}
    }
  }{\prod}
  \;
  \big(
    \hat{z} - z_{{}_I}
  \big)^{ -\weight_{{}_I}/\kappa }
  \,.
\end{equation}
This has come to be called the ``master function'' \cite[\S 2.1]{SlinkinVarchenko18},
in appreciation of the fact that any branch of the logarithm of it trivializes the
pullback of $\FlatConnectionForm$ to the universal cover
$$
  p^\ast
  \FlatConnectionForm
  \;=\;
  \frac{
    \DeRhamDifferential
    \MasterFunction
  }{\MasterFunction}
  \;=\;
  \DeRhamDifferential
  \,
  \log \MasterFunction
  \;\;\;
  \in
  \;
  \HolomorphicForms{\Big}{1}
  {\;
    \widehat{\Sigma^2}
  \;
  }
  \,,
$$
so that
multiplication by this function transforms the $\FlatConnectionForm$-twisted
holomorphic de Rham cohomology on $\Sigma^2$
(Rem. \ref{TwistedCohomologyOverSteinManifolds})
into the
Deck-transformation-equivariant de Rham cohomology on $\widehat{\Sigma^2}$:
\vspace{-2mm}
\begin{equation}
  \label{TwistedDeRhamIsomorphicToEquivariantDeRhamOnCoveringSpace}
  \begin{tikzcd}[row sep=-3pt]
    \overset{
     \mathclap{
     \raisebox{2pt}{
       \tiny
       \color{darkblue}
       \bf
       \begin{tabular}{c}
         {\color{orangeii}twisted }
         holomorphic de Rham cohomology
         \\
         of the complex curve
       \end{tabular}
     }
     }
    }{
    H^{\NumberOfProbeBranes}
    \Big(
      \HolomorphicDeRhamComplex{\big}
        {\Sigma^2;\,\ComplexNumbers}
      ,\,
      \DeRhamDifferential
      +
      \FlatConnectionForm
    \Big)
    }
    \ar[rr, "{\sim}"]
    &&
    \overset{
      \mathclap{
      \raisebox{2pt}{
        \tiny
        \color{darkblue}
        \bf
        \begin{tabular}{c}
        {\color{orangeii}equivariant}
        holomorphic de Rham cohomology
        \\
        of its universal cover
        \end{tabular}
      }
      }
    }{
    H^{\NumberOfProbeBranes}
    \bigg(\!\!\!
    \Big(
      \HolomorphicDeRhamComplex{\Big}
      {\,
        \widehat{\Sigma^2}
        ;\,
        \IrrepOfCyclicGroup
          {\Denominator}
          {\ShiftedLevel}
      }
      \Big)
      ^{{}^{
        \scalebox{.75}{$
          \pi_1(\Sigma^2)
        $}
      }}
    \bigg).
    }
    \\
    \scalebox{0.8}{$
      [\alpha]
    $}
    &\longmapsto&
    \scalebox{0.8}{$
      \big[
      \MasterFunction
      \cdot
      p_{{}_{\Sigma}}^\ast \alpha
      \big]
    $}
  \end{tikzcd}
\end{equation}

Here the action of $\pi_1(\Sigma^2)$ on the right is by pullback of differential forms along its canonical operation on the universal covering space and by action on their coefficients through $\pi_1(\Sigma^2) \xrightarrow{ [\FlatConnectionForm] } \CharacterGroup{\CyclicGroup{\ShiftedLevel}}$.

\end{itemize}

\begin{remark}[\bf Admissible rational levels seen in \TED-K-theory]
  \label{RationalLevelsInKTheory}
  $\,$

\noindent
{\bf (i)}  It follows from expression \eqref{TheMasterForm}
  for $\FlatConnectionForm$ on $\Sigma^2$
  that
  $$
    \Denominator
      \cdot
    \FlatConnectionForm
    (
       \vec{\weight}
      ,
      \kappa )
    \;=\;
    \FlatConnectionForm
    (
       \vec{\weight},
      \ShiftedLevel/\Denominator
    )
    \,,
  $$
and hence that the direct sum
  appearing in the $[\omega_1]$-twisted equivariant K-theory of $\Sigma^2$, according to Prop. \ref{TheBareBProfiniteIntegersTwistOfComplexRationalizedAEquivariantKTheory}, is the direct sum
  of 1-twisted holomorphic cohomology groups
  over all the {\it rational levels}  (cf. \eqref{FractionLevelInCastOfCharacters})
  \begin{equation}
    \label{FractionalLevel}
    \Level
      =
      -2
        +
      \ShiftedLevel/\Denominator
    \;\;\;\;
    \in
    \;\;
    \RationalNumbers
    \,;
  \end{equation}

  \vspace{-2mm}
\noindent
this means that the secondary Chern character becomes:
\vspace{-1mm}
  \begin{equation}
    \label{TEdKTheoryOfPuncturedPlaneAsDirectSum}
    \begin{tikzcd}[column sep=large]
    \underset{
      \mathclap{
        \scalebox{.65}{$
          \begin{array}{c}
            { d \,\in\, \Integers }
            \\
            { 1 \,\leq\, {\color{purple}r} \,<\, \kappa  }
          \end{array}
        $}
      }
    }{\bigoplus}
    \;\;\;
    \overset{
      \mathclap{
      \hspace{-0pt}
      \raisebox{3pt}{
        \tiny
        \color{darkblue}
        \bf
        \def\arraystretch{.9}
        \begin{tabular}{c}
          direct sum of
          {holomorphic de Rham}
          cohomology groups
          \\
          \phantom{a}
        \end{tabular}
      }
      }
    }{
    H^{
      \NumberOfProbeBranes
      +
      2d
    }
    \Big(
      \HolomorphicForms{}{\bullet}
        {\ComplexManifold}
      ,\,
      \DolbeaultDifferential
    }
    \underset{
      \mathclap{
      \raisebox{-6pt}{
        \tiny
        \color{darkblue}
        \bf
        \begin{tabular}{c}
          twisted by
          the flat holomorphic 1-form
          \\
          at
          {
            \color{orangeii}
            any
            $\ShiftedLevel$-fractional
            level
          }
      \end{tabular}
      }
      }
    }{
      +
      \FlatConnectionForm( \vec{\weight}, \ShiftedLevel/{\color{purple}r}) \wedge
    }
    \Big)
    \ar[
      rr,
      "{
        \mbox{
          \tiny
          \color{greenii}
          \bf
          \begin{tabular}{c}
            secondary
            \\
            Chern character
          \end{tabular}
        }
      }"
    ]
    &&
    \overset{
      \mathclap{
      \raisebox{4pt}{
        \tiny
        \color{darkblue}
        \bf
        \TED-K-theory of
        punctured plane
        inside $\mathbb{A}_{\ShiftedLevel-1}$-singularity
      }
      }
    }{
    \mathrm{KU}
      ^{
        1 +
        \NumberOfProbeBranes
        +
        [\FlatConnectionForm( \vec{w}, \kappa)]
      }_{\diff}
    \Big(
      \big(
        \ComplexPlane
        \setminus
        \{\vec z\}
      \big)
      \times
      \HomotopyQuotient
        {\ast}
        {\CyclicGroup{\ShiftedLevel}}
    \Big).
    }
    \end{tikzcd}
  \end{equation}
  \vspace{-2mm}
 \item {\bf (ii)} The  fractional level \eqref{FractionalLevel} is redundant
  when $\Denominator$ and $\ShiftedLevel$ contain a common factor,
  in that the same twisted cohomology group that it labels in \eqref{TEdKTheoryOfPuncturedPlaneAsDirectSum} then appears already for smaller values of $\ShiftedLevel$ and $\Denominator$. When this is not the case, hence when
 \vspace{-1mm}
  $$
    \mathrm{gcd}(\ShiftedLevel, \Denominator)
    \;=\;
    1
    \,,
  $$

  \vspace{-1mm}
\noindent  then precisely the formula
  \eqref{FractionalLevel} with precisely our condition $\ShiftedLevel \geq 2$ from \eqref{FractionLevelInCastOfCharacters}
  characterizes
  {\it admissible} fractional levels in conformal field theory.
  We come back to this in
  Rem. \ref{FractionalLevelWZWModels}
  and
  Rem. \ref{FurtherContentSeenInTheTEdKTheory} below.
\end{remark}

\begin{proposition}[Hypergeometric forms span twisted de Rham cohomology of punctured plane]
\label{OSBasisForTwistedCohomologyOfPuncturedPlane}
For $\vec \weight \,\in\, \{0, \cdots, \ShiftedLevel - 1\}^{\NumberOfPunctures}$ such that $\frac{\IncomingWeight}{\ShiftedLevel} \,\not\in\, \NaturalNumbers_+$,
the $[\FlatConnectionForm(\vec \weight, \ShiftedLevel)]$-twisted holomorphic de Rham cohomology of $\ComplexPlane \setminus \{\vec z\}$  in degree 1 is spanned by the ``hypergeometric'' 1-forms $\Differential \log(z - z_{{}_I})$ subject to the single relation
$
\big[
  \FlatConnectionForm(\vec \weight, \ShiftedLevel)
\big]
\,=\, 0
$.
That is, there is a natural isomorphism:
\begin{equation}
  \label{GeneratorsForTwistedDeRhamCohomologyOfPuncturedPlane}
  H^1
  \Big(
    \HolomorphicDeRhamComplex{\big}
      {
        \ComplexPlane
        \setminus
        \{\vec z\}
      }
      ,
      \,
      \DeRhamDifferential
        +
      \FlatConnectionForm(\vec \weight, \ShiftedLevel)
  \Big)
  \;\;
  \simeq
  \;\;
  \frac{
  \big\langle
    \Differential
    \log(z - z_{{}_1})
    ,
    \cdots,
    \Differential
    \log(z - z_{{}_{\NumberOfPunctures}})
  \big\rangle
  }
  {
    \Big\langle
    \underset{
     \scalebox{0.5}{$ 1 \leq I \leq \NumberOfPunctures $}
    }{\sum}
    \;
    \frac{\weight_I}{\ShiftedLevel}
    \Differential
    \log(z - z_{{}_I})
    \Big\rangle
  }.
\end{equation}
\end{proposition}
\begin{proof}
This is essentially a special case of \cite{ESV92}.\footnote{
For vanishing twist ($\vec \weight = 0$) the result
of Prop. \ref{OSBasisForTwistedCohomologyOfPuncturedPlane}
is due to  \cite{OrlikSolomon80}, whence the subalgebra of the de Rham algebra generated by the $\Differential \log(z - z_{{}_{I}})$ is also known as an {\it Orlik-Solomon algebra}.}
In extracting this statement from \cite{ESV92}, beware that their expression for $\omega$ is of the form \eqref{FlatConnectionWithExplicitOutgoingTerm}, cf. \cite[\S 4]{SchechtmanTeraoVarchenko95}: With this understood, the sufficient condition ``{\bf (Mon)}'' in \cite[p. 2]{ESV92} is here equivalently the condition $\frac{\IncomingWeight}{\ShiftedLevel} \,\not\in\, \NaturalNumbers_+$.
\end{proof}

 \medskip

These ``weights'' $\weight_{{}_I}$ of holonomies over $\Sigma^2$ are now going to be identified with weights in the sense of Lie theory:

\medskip

\noindent
{\bf The $\slTwo$-Modules.} (see e.g. \cite[\S 13]{DMS97})

\vspace{-.2cm}
\begin{itemize}[leftmargin=*]
\setlength\itemsep{-1pt}

\item
$
  \slTwo
  \,:=\,
  \mathfrak{sl}(2,\ComplexNumbers)
$
denotes the complex Lie algebra of traceless complex $2 \times 2$ matrices.

The usual Chevalley generators of this Lie algebra are denoted:
\begin{equation}
  \label{CanonicalGeneratorsOfslTwo}
  \RaisingOperator
  \,:=\,
\begin{bmatrix}
    0 & 1
    \\
    0 & 0
  \end{bmatrix}
  \,,
  \;\;\;\;\;\;\;\;
  \WeightOperator
  \,:=\,
\begin{bmatrix}
    1 & 0
    \\
    0 & -1
\end{bmatrix}
  ,\,
  \;\;\;\;\;\;\;\;
  \LoweringOperator
  \,:=\,
\begin{bmatrix}
    0 & 0
    \\
    1 & 0
\end{bmatrix}
  \;\;\;\;\;\;
  \in
  \;\;
  \slTwo
\end{equation}
subject to these Lie bracket relations:
\vspace{-2mm}
$$
  [
    \WeightOperator
    ,
    \RaisingOperator
  ]
  \,=\,
  +
  2 \RaisingOperator
  \,,
  \;\;\;\;\;\;\;
  [
    \WeightOperator
    ,
    \LoweringOperator
  ]
  \,=\,
  - 2 \LoweringOperator
  \,,
  \;\;\;\;\;\;\;
  [
    \RaisingOperator
    ,
    \LoweringOperator
  ]
  \,=\,
  \WeightOperator
  \,.
$$
We may and will understand $\slTwo$ as the complexification of the special unitary Lie algebra $\suTwo$, hence as a convenience for speaking about the unitary Lie representations of $\suTwo$ in terms of complex-linear Lie representations of $\slTwo$:

\vspace{-.3cm}
\begin{equation}
  \label{ComplexificationOfsuTwo}
  \slTwo
  \;\simeq\;
  (\suTwo)_{{}_{\ComplexNumbers}}
  \,,
  \hspace{1.6cm}
  \overset{
    \mathclap{
    \raisebox{3pt}{
      \tiny
      \color{darkblue}
      \bf
      \def\arraystretch{.9}
      \begin{tabular}{c}
        unitary Lie representations
        \\
        of real Lie algebra
      \end{tabular}
    }
    }
  }{
  \big\{
    \suTwo
    \xrightarrow[\RealNumbers]{\;}
    \UnitaryLieAlgebra{d}
  \big\}
  }
  \;\;
  \leftrightarrow
  \;\;
  \big\{
    \!\!\!
    \begin{tikzcd}[column sep=6pt]
      \slTwo
      \ar[r, shorten=-2pt]
      \ar
        [
          rr,
          rounded corners,
          to path={
            ([yshift=-0pt]\tikztostart.south)
            --
            ([yshift=-4pt]\tikztostart.south)
            --
            node[below]{
              \scalebox{.7}{$
                \ComplexNumbers
              $}
            }
            ([yshift=-3pt]\tikztotarget.south)
            --
            ([yshift=+3pt]\tikztotarget.south)
          },
          "{\ComplexNumbers}{swap}"
        ]
      \ar[
        rr,
        phantom,
        shift left=13pt,
        "{
         \mathclap{
          \raisebox{3pt}{
            \tiny
            \color{darkblue}
            \bf
            \def\arraystretch{.9}
            \begin{tabular}{c}
              complex-linear Lie representation
              \\
              of complexified Lie algebra
            \end{tabular}
          }
          }
        }"
      ]
      &
      \UnitaryLieAlgebra{d}
      \ar[r, hook, shorten=-2pt]
      &
      \mathfrak{gl}(\mathbb{C}^d)
    \end{tikzcd}
    \!\!\!
  \big\}
\end{equation}
\vspace{-.4cm}

In particular, we never regard $\slTwo$ as a real Lie algebra and are {\it not} concerned with what is called $\SpecialLinearGroup(2,\ComplexNumbers)$-Chern-Simons theory. Instead, the complex-linear representations of $\slTwo$ serve as a convenient way of speaking about the complex representations of the real Lie algebra $\suTwo$ which appear in the discussion of the usual $\SUTwo$ CS/WZW theory (eg. \cite{Witten82}).
%
%
%
%

\vspace{-1mm}
\item
$
  \HighestWeightIrrep{\weight}
  \,\in\,
  \slTwo \Modules
$ denotes the highest weight $\slTwo$-irrep
of highest weight $\weight$, meaning that
(it is irreducible and) there is a highest eigenvalue of $\WeightOperator$ \eqref{CanonicalGeneratorsOfslTwo}
equal to $\weight$.

\item
$\HighestWeightVector \,\in\, \HighestWeightIrrep{\weight}$ denotes the corresponding highest weight eigenvector,
characterized by
$\WeightOperator
  \cdot \HighestWeightVector
  =
  \weight \cdot \HighestWeightVector$ and
$\RaisingOperator \cdot \HighestWeightVector = 0$.

\item
$
\IncomingHighestWeightReps
\;:=\;
\HighestWeightIrrep{\weight_1}
  \otimes
  \cdots
  \otimes
  \HighestWeightIrrep{\weight_{\NumberOfPunctures}}
$
denotes the tensor product of ``incoming'' highest weight representations'' of $\slTwo$.

For $v = v_1 \otimes \cdots \otimes v_{\NumberOfPunctures} \,\in\, \IncomingHighestWeightReps$
and $1 \leq i \leq \NumberOfPunctures$,
we write
$$
  \LoweringOperator_i
    \cdot
  v
  \;:=\;
  v_1
    \otimes
    \cdots
    \otimes
  v_{i-1}
    \otimes
  (\LoweringOperator\cdot v_i)
    \otimes
  v_{i + 1}
    \otimes
    \cdots
    \otimes
  v_{\NumberOfPunctures}
  \;\;\;
  \in
  \;
  \IncomingHighestWeightReps
  \,.
$$
We define the action of
$\WeightOperator_i$ and $\RaisingOperator_i$ similarly.

\item
$
  \IncomingHighestWeightReps(\OutgoingWeight)
  \;\xhookrightarrow{\;\;}\;
  \IncomingHighestWeightReps
$
denotes the linear subspace of incoming representation vectors of total weight $\OutgoingWeight$:
\begin{equation}
  \label{IncomingRepresentationVectorsOfOutgoingWeight}
  \IncomingHighestWeightReps(\OutgoingWeight)
  \;:=\;
  \Big\{
    v \in
  \HighestWeightIrrep{\weight_1}
  \otimes
  \cdots
  \otimes
  \HighestWeightIrrep{\weight_{\NumberOfPunctures}}
    \;\big\vert\;
    \sum_{i} \WeightOperator_i
    \cdot
    v
    \;=\;
    \OutgoingWeight \cdot v
  \Big\}.
\end{equation}

\vspace{-2mm}
\item $V_{\mathfrak{g}} \,:=\, V/(\mathfrak{g}\cdot V)$ (for any Lie algebra $\mathfrak{g}$ and $\mathfrak{g}$-module $V$) denotes the quotient space of {\it co-invariants}, i.e., the space such that linear maps out of it are the $\mathfrak{g}$-{\it invariant} linear maps out of $V$.

Specifically:

\item
$
\big(
  \HighestWeightIrrep
    {\weight_1}
  \otimes
    \cdots
  \otimes
  \HighestWeightIrrep
    {\weight_{\NumberOfPunctures + 1}}
\big)_{\slTwo}
$
denotes the space of co-invariants of the diagonal $\slTwo$-action on the tensor products of its highest weight
irreps for the given weights $( \vec{\weight}, \OutgoingWeight)$.

For example, if $v_i \in \HighestWeightIrrep{\weight_i}$ denote the highest weight vectors, and $f \in \slTwo$
according to \eqref{CanonicalGeneratorsOfslTwo}, then the element
$$
  f
  \cdot \big(
    v_1 \otimes v_2 \otimes \cdots \otimes
    v_{\NumberOfPunctures + 1}
  \big)
  \;:=\;
  \big(
    (f \cdot v_1)
      \otimes
    v_2
      \otimes
      \cdots
      \otimes
    v_{\NumberOfPunctures + 1}
  \big)
  \,+\,
  \big(
    v_1
      \otimes
    (f \cdot v_2)
      \otimes
      \cdots
      \otimes
    v_{\NumberOfPunctures + 1}
  \big)
  \,+\,
  \cdots
  \,+\,
  \big(
    v_1
      \otimes
    v_2
      \otimes
      \cdots
      \otimes
    f \cdot v_{\NumberOfPunctures + 1}
  \big)
$$
becomes null in the space of co-invariants.

\end{itemize}

\begin{lemma}[$\slTwo$-Coinvariants in terms of incoming representation vectors {\cite[Lem. 2.3.3]{FeiginSchechtmanVarchenko94}}]
  \label{slTwoCoinvariantsInTermsOfIncomingRepresentationVectors}
  Adjoining the highest weight(=$\OutgoingWeight$)-vector
  $\HighestWeightVector_{\NumberOfPunctures + 1}$
  of the outgoing $\slTwo$-irrep
  $\HighestWeightIrrep{\OutgoingWeight}$
  to a tensor product of incoming $\slTwo$-representation vectors
  constitutes a linear identification
  of the space of $\slTwo$-coinvariants
  in the tensor product of incoming and outgoing
  irreps
  with the quotient of
  the weight=$\OutgoingWeight$-subspace of
  just the tensor product of incoming irreps
  \eqref{IncomingRepresentationVectorsOfOutgoingWeight}
  by the image of the diagonal lowering operator:
  \vspace{-2mm}
  \begin{equation}
    \label{IsomorphismFromIncomingRepVectorsToslTwoCoinvariants}
    \begin{tikzcd}[row sep=-3pt]
      \IncomingHighestWeightReps(\OutgoingWeight)
      \big/
      \big(
        \mathrm{im}
        (\sum_i f_i)
      \big)
      \ar[
        rr,
        "{ \sim }"{swap}
      ]
      &&
      \big(
        \HighestWeightIrrep{\weight_1}
        \otimes
        \cdots
        \otimes
        \HighestWeightIrrep{\weight_{\NumberOfPunctures}}
        \otimes
        \HighestWeightIrrep{\weight_{\NumberOfPunctures + 1}}
      \big)_{\slTwo}
      \\
   \scalebox{0.8}{$   \big[
        v_1
          \otimes
          \cdots
          \otimes
        v_{\NumberOfPunctures}
      \big]
      $}
      &\longmapsto&
   \scalebox{0.8}{$   \big[
        v_1
          \otimes
          \cdots
          \otimes
        v_{\NumberOfPunctures}
          \otimes
        \HighestWeightVector_{\NumberOfPunctures + 1}
      \big]
      $}
    \end{tikzcd}
  \end{equation}
\end{lemma}

\medskip

\noindent
{\bf The $\slTwoAffine{\Level}$-Modules.} (see e.g. \cite[\S 14]{DMS97})

\vspace{-.2cm}
\begin{itemize}[leftmargin=*]
\setlength\itemsep{-1pt}

\item $\slTwoAffine{\Level}$ denotes its ``affine'' version, i.e., the universal central extension
at level $k$ of the infinite-dimensional algebra of loops in $\slTwo$ (\cite{Kac83}).

\item
$\AffineHighestWeightIrrep{\weight}{\Level} \,\in\, \slTwoAffine{\Level}\Modules$ denotes the highest weight $\slTwoAffine{\Level}$-irrep.

\end{itemize}

\begin{remark}[\bf Integrable and admissible weights]
\label{IntegrableAndAdmissibleWeights}
$\,$

\noindent
{\bf (i)} One says that a weight $\weight \in \NaturalNumbers$ satisfies the
{\it integrability condition} (cf. \cite[\S 2.3.2, 3.2]{SchechtmanVarchenko90}) at a given level $\Level$
if and only if
\begin{equation}
  \label{IntegrabilityConditionOnWeights}
  \begin{array}{ll}
    &
    0
    \;\leq\;
    \weight
    \;\leq\;
    \level\,,
    \\
    \mbox{i.e.,}
    &
    0
    \;\leq\;
    \weight
    \;\leq\;
    \ShiftedLevel - 2
    \,.
  \end{array}
\end{equation}
The highest weight irreps
$\AffineHighestWeightIrrep{\weight}{\Level}$ whose highest weight satisfies the integrability condition
\eqref{IntegrabilityConditionOnWeights}:

\vspace{-.3cm}
\begin{itemize}[leftmargin=*]
\setlength\itemsep{-1pt}

\item  are indeed
{\it integrable}
(\cite[p. 176]{SchechtmanVarchenko90}\footnote{Beware that in \cite{FeiginSchechtmanVarchenko94} these integrable irreps are denoted $\overline{L}(\weight)$, while $\widehat L(\weight)$ there denotes a larger reducible rep, in the integrable case.}) in the sense of Lie theory,
in that their $\big(\slTwo \xhookrightarrow{\;} \slTwoAffine{\Level}\,\big)$-action extends  to a representation of the group $\SLTwo$, \cite[Prop. 3.6]{Kac83};

\item  label the {\it primary fields} of the corresponding WZW conformal field theory (e.g. \cite[(2.69)]{Walton00}).

\end{itemize}

\vspace{-2mm}
\noindent
{\bf (ii)}
More generally, for
admissible fractional levels $k = - 2 + \ShiftedLevel/\Denominator$
(Rem. \ref{RationalLevelsInKTheory}),
a weight $\weight$ is called {\it admissible} (e.g. \cite[(3.3)]{Rasmussen20}, following \cite{KacWakimoto88ModularInvariantRepresentations}\cite{KacWakimoto89})
if and only if
\begin{equation}
  \label{AdmissibilityCondition}
  \weight
  \;\in\;
  \left\{
  \!\!\!
  \def\arraystretch{1.3}
  \begin{array}{r}
    (1+a)
    \cdot
    \frac{\ShiftedLevel}{\Denominator}
    -
    b
    -
    1
    \mathrlap{\,,}
    \\
    -
    a
    \cdot
    \frac{\ShiftedLevel}{\Denominator}
    +
    b
    -
    1
  \end{array}
  \;
  \middle\vert
  \mbox{
    $
      \arraycolsep=4pt
      \def\arraystretch{1.3}
      \begin{array}{lll}
        \mbox{for} &
        a & \in \{0, \cdots, \Denominator - 1\}
        \\
        \mbox{and}
        &
        b & \in \{1, \cdots, \ShiftedLevel - 1\}
      \end{array}
    $
  }
  \!\!\!
  \right\}
  \;\;\;
  \supset
  \;\;\;
  \big\{
    0 , 1, \cdots, \ShiftedLevel - 2
  \big\}
  \,.
\end{equation}
On  the right we have highlighted that this set of admissible weights always contains
(using  $a = \Denominator-1$ in the first line or $a = 0$ in the second)
the subset of integrable weights \eqref{IntegrabilityConditionOnWeights},
which is always the largest subset of admissible weights that are also integers.
\end{remark}

\begin{remark}[\bf $\SLTwoZ$-Action on span of characters of $\slTwoAffine{\Level}$-Modules]
\label{SLTwoZActionOnAffineCharacters}
  Associated with an integrable highest weight
  $\slTwoAffine{\Level}$-irrep $\AffineHighestWeightIrrep{\weight}{\Level}$
  (Rem. \ref{IntegrableAndAdmissibleWeights})
  is its {\it affine character}
  $\AffineCharacter{\weight}{\Level}$ (e.g., \cite[\S]{DMS97}\cite[(3.43)]{Walton00}),
which may be thought of as the partition function of the chiral $\slTwoAffine{\Level}$-WZW-model over the complex torus
$
\Quotient
  { \ComplexNumbers }
  { (\Integers \times \tau\cdot \Integers) }
$.
Under the modular transformation action of $\SLTwoZ$
on the characters
(e.g. \cite[(4.6)]{Walton00})
via the canonical action on the modulus
$\tau \in \ComplexNumbers \vert_{ \mathrm{Re} > 0}$, given by
\begin{equation}
  \label{MoebiusTransformations}
  \begin{pmatrix}
       a & b
    \\
    c & d
  \end{pmatrix}
  \cdot \tau
  \;:=\;
  \frac{
    a \tau + b
  }{
    c \tau + d
  }
  \,,
  \;\;\;\;\;\;\;\;
  \mbox{e.g.}
  \;\;\;\;\;\;\;\;
  \underset{
    =: \TOperator
  }{
  \underbrace{\small
  \begin{pmatrix}
    1 & 1
    \\
    0 & 1
\end{pmatrix}
  }
  }
  \cdot \tau
  \;:=\;
  \tau + 1
  \,,
  \;\;\;\;\;\;\;
  \underset{
    =: \SOperator
  }{
  \underbrace{ \small
 \begin{pmatrix}
    0 & -1
    \\
    1 & 0
\end{pmatrix}
  }
  }
  \cdot \tau
  \;:=\;
  -1/\tau
  \,,
\end{equation}

\vspace{-1mm}
\noindent
these characters span a $\Level+1$-dimensional unitary representation
  $
    \big\langle
      \AffineCharacter{\weight}{\Level}
    \big\rangle_{ 0 \leq \weight \leq \Level }
  $
  of $\SLTwoZ$,
(\cite[p. 159]{KacWakimoto88}, based on \cite{KacPeterson84})  according to the following concrete formulas\footnote{The linear combination on the right in the second line of \eqref{SLTwoZTransformationOfAffineCharacters} is overdetermined, hence reduces
to a linear combination of just the first $\Level+1$ characters.}
  (\cite[(2.19) with (2.13)]{COZ87a}\cite[(26) with (24)]{COZ87b}):
  \vspace{-2mm}
\begin{equation}
    \label{SLTwoZTransformationOfAffineCharacters}
    \def\arraystretch{2.5}
  \begin{array}{rl}
    (T
    \cdot
    \AffineCharacter{\weight}{\Level}
    )(\tau)
    \;:=\;
    \AffineCharacter{\weight}{\Level}
    (\tau + 1)
    &
    =\;
    \exp
    \bigg(
      2\pi \ImaginaryUnit
      \Big(
        \frac{
          (\weight+1)^2
        }{
          4 \cdot
          (\Level + 2)
        }
        -
        \frac{1}{8}
      \Big)
    \!\!\bigg)
    \cdot
    \AffineCharacter{\weight}{\Level}(\tau)
    \,,
    \\
    (S
    \cdot
    \AffineCharacter{\weight}{\Level}
    )(\tau)
    \;:=\;
    \AffineCharacter{\weight}{\Level}
    (-1/\tau)
    &
    =\;
    \frac{-\ImaginaryUnit}{
      \sqrt{2(\Level+2)}
    }
    \;
   {\displaystyle \sum
      _{\weight' = 1}
      ^{ 2\Level + 2 }
      }
    \exp
    \bigg(
      2\pi\ImaginaryUnit
      \Big(
        \frac{
          (\weight + 1)
          (\weight' + 1)
        }{
          2(\Level + 2)
        }
      \Big)
    \!\! \bigg)
    \cdot
    \AffineCharacter{\weight}{\Level}(\tau)
    \,.
    \end{array}
  \end{equation}
\end{remark}

\begin{example}[The 2-dimensional $\SLTwoZ$-Representation on $\slTwoAffine{1}$-Characters]
\label{TheTwoDimensionalSLTwoZRepOnAffineCharacters}
For unit level $\Level = 1$, the
$\SLTwoZ$-representation \eqref{SLTwoZTransformationOfAffineCharacters} is 2-dimensional, spanned by the affine characters $\AffineCharacter{0}{1}$ and $\AffineCharacter{1}{1}$ at weight 0 and weight 1, respectively.
Since
$$
  \tfrac{1}{12}
  -
  \tfrac{1}{8}
  \;=\;
  -
  \tfrac{1}{24}
  \,,
  \;\;\;\;\;\;\;\;\;\;
  \tfrac{4}{12}
  -
  \tfrac{1}{8}
  \;=\;
  +
  \tfrac{5}{24}\,,
$$
the action of $\TOperator$ from \eqref{MoebiusTransformations}
on the characters
is given in this case by
\begin{equation}
  \TOperator
  \cdot
  \AffineCharacter{0}{1}
  \;=\;
  \exp
  (
    - 2 \pi \ImaginaryUnit
    /
    24
  )
  \cdot
  \AffineCharacter{0}{1}
  \,,
  \;\;\;\;\;\;\;\;\;\;\;\;\;
  \TOperator
  \cdot
  \AffineCharacter{1}{1}
  \;=\;
  \exp
  (
     2 \pi \ImaginaryUnit
    \cdot 5/
    24
  )
  \cdot
  \AffineCharacter{1}{1}
  \,.
\end{equation}
In particular, the action of $\TOperator$
on tensor powers of this
representation becomes trivial
exactly for
$$
  \TOperator
  \cdot
  \big(
    \AffineCharacter{\weight}{\Level}
  \big)^{\!\otimes^{24}}
  \;\;
  =
  \;\;
  1
  \cdot
  \big(
    \AffineCharacter{\weight}{\Level}
  \big)^{\!\otimes^{24}}
  \,.
$$

\end{example}
\begin{definition}[Truncation condition]
\label{TruncationCondition}
Following \cite[(14)]{FeiginSchechtmanVarchenko94},
we say that an $\NumberOfPunctures$-tuple of weights $ \vec{\weight}$
satisfies the {\it truncation condition} for a given number $\NumberOfProbeBranes$ of insertions
if and only if
\vspace{-2mm}
  \begin{align}
    \label{TruncationConditionFormula}
        \IncomingWeight
    \;+\;
    \OutgoingWeight
    &\;>\;
    2 \Level
    \nonumber
    \\
    \Leftrightarrow \quad
        \IncomingWeight
    \;+\;
    (\IncomingWeight
      -
    2 \cdot \NumberOfProbeBranes)
    &\;>\;
    2 (\ShiftedLevel - 2 )
    \\
    \nonumber
    \Leftrightarrow \hspace{2.55cm}
        \IncomingWeight
    & \;>\;
    \ShiftedLevel
      +
    ( \NumberOfProbeBranes - 2 )
    \,.
  \end{align}
\end{definition}
\begin{remark}[\bf Truncation condition combined with integrability]
\label{TruncationConditionCombinedWithIntegrability}
Assuming
the {\it integrability condition}
\eqref{IntegrabilityConditionOnWeights}
on the outgoing weight $\OutgoingWeight$,
which says that
\begin{align}
  \label{IntegrabilityConditionOnOutGoingWeight}
        \OutgoingWeight
    &\;\leq\;
    \Level
    \nonumber
    \\
    \Leftrightarrow
    \quad
        \IncomingWeight
      -
    2 \cdot \NumberOfProbeBranes
    &\;\leq\;
    \ShiftedLevel - 2
    \\
    \nonumber
    \Leftrightarrow
    \hspace{1.25cm}
    \IncomingWeight
    &\;\leq\;
    \ShiftedLevel
    \;+\;
    2( \NumberOfProbeBranes - 1 )
    \,,
\end{align}
the truncation condition is equivalent to
\begin{equation}
  \label{TheTruncationConditionCombinedWithIntegrability}
  \ShiftedLevel
    \,+\,
  (\NumberOfProbeBranes - 1)
  \;\;\leq\;\;
  \IncomingWeight
  \;\;\leq\;\;
  \ShiftedLevel
    \,+\,
  2(\NumberOfProbeBranes - 1)
  \,.
\end{equation}
\end{remark}

\begin{example}[Truncation condition for $\NumberOfProbeBranes = 1$]
\label{TruncationConditionForSingleProbeBrane}
The truncation condition for
$\NumberOfProbeBranes = 1$, and
using the integrability condition
(Rem. \ref{TruncationConditionCombinedWithIntegrability}),
is equivalently the condition
\begin{equation}
  \label{TruncationConditionForOneProbe}
  \IncomingWeight \;=\; \ShiftedLevel
  \,,
\end{equation}
while its failure is the condition $\IncomingWeight < \ShiftedLevel$.
\end{example}

\newpage

\noindent
{\bf The $\slTwoAffine{\Level}$-conformal blocks.}

\vspace{-.2cm}
\begin{itemize}[leftmargin=*]
\setlength\itemsep{-1pt}

\vspace{-0mm}
\item $\NumberOfProbeBranes \,\in\, \NaturalNumbers_+$ denotes a positive integer which, in the next bullet point, measures (half) the difference between incoming and outgoing weights, to be called the {\it degree} of the corresponding conformal block (see below \eqref{ActionOfAlgebraicFunctionsOverSigmaOnModulesAtThePunctures}), and eventually identified also with the number of probe branes, hence of points in a configuration of points in complex curve $\Sigma^2$.

\vspace{-.1cm}
\item
$
\OutgoingWeight
\;:=\;
\weight_{\NumberOfPunctures + 1}
\,:=\,
\IncomingWeight
\,-\, 2 \cdot \NumberOfProbeBranes$ denotes the
``outgoing weight at infinity'', constrained to be the sum of incoming weights minus twice the number of field
insertions (cf. \cite[above (12)]{SchechtmanVarchenko90}).

\item $\slTwo(\Sigma^2)$ denotes the ring of $\slTwo$-valued holomorphic functions on $\Sigma^2$
(hence meromorphic
functions on $\RiemannSphere$ with possible poles at $\{\vec z, \infty\}$), regarded as a Lie algebra under
the inherited pointwise Lie bracket.

 \item
 $\slTwo(\Sigma^2) \xrightarrow{ (-)\vert_{z_i} } \slTwoAffine{0} $ denotes the operation of forming the Laurent
 expansion of such an algebraic function around the puncture at $z_i$.
 The direct product of these around all punctures lifts uniquely to the diagonal central extension
 (\cite[(8)]{FeiginSchechtmanVarchenko94}):
 \vspace{-2mm}
$$
  \begin{tikzcd}[column sep=40pt]
    &&
    \;\;\;\;
    \widehat{
    \underset{
      \mathclap{
        \scalebox{.6}{$
          1 \!\leq\! i \!\leq\! \NumberOfPunctures \!+\! 1
        $}
      }
    }{\bigoplus}
    \;\;
    \slTwo
    }^{{}_{\Level}}
    \;\;\;
    \ar[d]
    \\
    \slTwo(\Sigma^2)
    \ar[
      rr,
      "{
        \underset{
          1
            \leq
          i
            \leq
          \NumberOfPunctures + 1
        }{\sum}
          (-)\vert_{z_i}
      }"{swap}
    ]
    \ar[urr, dashed, "{ \exists ! }"]
    && \;\;
    \underset{
      \mathclap{
        \scalebox{.6}{$
          1 \!\leq\! i \!\leq\! \NumberOfPunctures \!+\! 1
        $}
      }
    }{\bigoplus}
    \;\;
    \slTwoAffine{\,0}
  \end{tikzcd}
$$

\vspace{-2mm}
\noindent Through this Lie homomorphism, $\slTwo(\Sigma^2)$
acts on any $(\NumberOfPunctures + 1)$-fold tensor product
of $\slTwoAffine{\Level}$-modules, notably
on the tensor product of integrable highest weight irreps \eqref{IntegrabilityConditionOnWeights}
for the given weights $( \vec{\weight}, \OutgoingWeight)$:
\begin{equation}
  \label{ActionOfAlgebraicFunctionsOverSigmaOnModulesAtThePunctures}
  \slTwo(\Sigma^2)
  \;\acts\;
  \big(
  \AffineHighestWeightIrrep{\weight_1}{\Level}
  \otimes
  \cdots
  \otimes
  \AffineHighestWeightIrrep{\weight_{\NumberOfPunctures}}{\Level}
  \otimes
  \AffineHighestWeightIrrep{\OutgoingWeight}{\Level}
  \big)
  \,.
\end{equation}

\item
$
  \ConformalBlocks^{n}_{
    \slTwoAffine{\Level}
  }
  (
     \vec{\weight}
    ,
     \vec{z}
  )
  \;:=\;
  \big(
  \AffineHighestWeightIrrep{\weight_1}{\Level}
  \otimes
  \cdots
  \otimes
  \AffineHighestWeightIrrep{\weight_{\NumberOfPunctures}}{\Level}
  \otimes
  \AffineHighestWeightIrrep{\OutgoingWeight}{\Level}
  \big)_{\slTwo(\Sigma^2)}
  $
  denotes the co-invariants of the action
  \eqref{ActionOfAlgebraicFunctionsOverSigmaOnModulesAtThePunctures},
  called the {\it space of conformal blocks} (e.g. \cite{Beauville94}\cite[\S 1.4]{Kohno02}):

\vspace{-.2cm}
\begin{itemize}

\item
of the $\slTwoAffine{\Level}$-WZW-model

\item
on $\RiemannSphere \setminus \{z_1, \cdots, z_n, \infty\}$;

\item at level $\Level$,

\item with the $i$th puncture $z_i$ labeled by the integrable $\slTwoAffine{\Level}$-irrep of highest weight $\weight_i$, for all $i \,\in\, \{1, \cdots, N+1\}$,

\item with
$n
  \,=\,
\tfrac{1}{2}(\IncomingWeight - \OutgoingWeight)
  \;:=\;
\tfrac{1}{2} \Big(
  \Big(\;
    \underset{1 \leq i \leq \NumberOfPunctures}{\sum}
    \,
    \weight_i
  \Big)
  -
  \weight_{\NumberOfPunctures + 1}
\Big)
$
primary field insertions.

\end{itemize}

\end{itemize}

We are going to show that these spaces of conformal blocks are naturally transformed into the \TED-K-theory of $\HomotopyQuotient{\Sigma^2}{\CyclicGroup{\kappa}}$. The first step towards this reformulation is the fact that -- despite
the superficial appearance of the above definition -- conformal blocks may be expressed in terms of just the irreps of the
underlying finite-dimensional Lie algebra $\slTwo$:

\begin{proposition}[Conformal blocks as quotient of space of incoming $\slTwo$-reps  {\cite[Lem. 2.3.3, Thm. 2.3.6]{FeiginSchechtmanVarchenko94}}]
\label{ConformalBlocksAsQuotientsOfCoInvariantOfSLTwo}
$\,$

\noindent
  The above spaces of
  $\slTwoAffine{\Level}$-conformal blocks
  on $\RiemannSphere \setminus \{\vec z, z_{\NumberOfPunctures + 1}\}$
  are isomorphic to the
  plain $\slTwo$-coinvariants, hence to the
  quotient \eqref{IsomorphismFromIncomingRepVectorsToslTwoCoinvariants}
  of the incoming $\slTwo$-reps --
  except when the truncation condition holds (Def. \ref{TruncationCondition}),
  in which case they are even smaller:
  \begin{equation}
    \label{slTwoConformalBlocksAsQuotientOfSlTwoCoinvariants}
    \ConformalBlocks^{n}_{
      \slTwoAffine{\Level}
    }
    (
       \vec{\weight}
      ,
       \vec{z}
    )
    \;\;
    \simeq
    \;\;
    \left\{
    \!\!\!
    \def\arraystretch{2}
    \begin{array}{ll}
    \IncomingHighestWeightReps(\OutgoingWeight)
    \big/
    \mathrm{im}
    \big(
      \sum_i \LoweringOperator_i
    \big)
    \Big/
    \mathrm{im}
    \Big(
    \big(
      \sum_i (z_i \LoweringOperator_i)
    \big)
    ^{
      \ShiftedLevel  - 1 -  \OutgoingWeight
    }
    \Big)
    &
    \mbox{\rm
      if
      \eqref{TruncationConditionFormula}
      holds
    }
    \\
    \IncomingHighestWeightReps(\OutgoingWeight)
    \big/
    \mathrm{im}
    \big(
      \sum_i \LoweringOperator_i
    \big)
    \underset{
      \raisebox{-2pt}{
        \tiny
        \rm
        Lem. \ref{slTwoCoinvariantsInTermsOfIncomingRepresentationVectors}
      }
    }{
      \;\simeq\;
    }
    \big(
      \HighestWeightIrrep{\weight_1}
      \otimes
      \cdots
      \otimes
      \HighestWeightIrrep{\weight_{\NumberOfPunctures}}
      \otimes
      \HighestWeightIrrep{\OutgoingWeight}
    \big)_{\slTwo}
    &
    \mbox{\rm otherwise.}
    \end{array}
    \right.
  \end{equation}
\end{proposition}
\begin{remark}[Truncation]
  If the truncation condition does not hold then the
  image of the second operator
  in the first line of
  \eqref{slTwoConformalBlocksAsQuotientOfSlTwoCoinvariants}
  is null
  (\cite[(13)]{FeiginSchechtmanVarchenko94}).
  While, therefore, the first line
    already subsumes the second,
  we have split up the statement
  to highlight that, more often than not,
  the conformal blocks
  are just plain $\slTwo$-coinvariants
  (cf. item (ii) in Prop. \ref{Degree1ConformalBlocksInTwistedCohomology} below).
\end{remark}

\begin{example}[Conformal blocks in terms of incoming representation vectors.]
Recalling that
$
\NumberOfProbeBranes
=
\tfrac{1}{2}(\IncomingWeight - \OutgoingWeight)
$, we have:

\noindent
{\bf(i)} For
$
{\bf
\NumberOfProbeBranes
=
0}$,
the space
of incoming representation vectors
of total weight $\OutgoingWeight$
is spanned by the the single vector
\begin{equation}
  \label{AbsoluteVacuumVector}
  \big\vert
    \HighestWeightVector_1
    ,
    \cdots
    ,
    \HighestWeightVector_{\NumberOfPunctures}
  \big\rangle
  \;\;
  :=
  \;\;
    \HighestWeightVector_1
    \otimes
    \cdots
    \otimes
    \HighestWeightVector_{\NumberOfPunctures}
  \;\;\;
  \in
  \;
  \IncomingHighestWeightReps(\OutgoingWeight)
  \,.
\end{equation}

\noindent
{\bf(ii)} For
$
{\bf
\NumberOfProbeBranes
=
1}
$,
the space
of incoming representation vectors
of total weight $\OutgoingWeight$
is spanned by the $\NumberOfPunctures$ elements
\begin{equation}
  \label{fOperatorInithPosition}
  \LoweringOperator_{i}
  \big\vert
    \HighestWeightVector_1,
    \cdots,
    \HighestWeightVector_{\NumberOfPunctures}
  \big\rangle
  \;\;
  :=
  \;\;
    \HighestWeightVector_1
    \otimes
    \cdots
    \otimes
    \HighestWeightVector_{i-1}
    \otimes
    (\LoweringOperator \cdot \HighestWeightVector_{i})
    \otimes
    \HighestWeightVector_{i + 1}
    \otimes
    \cdots
    \otimes
    \HighestWeightVector_{\NumberOfPunctures}
  \;\;\;
  \in
  \;\;
  \IncomingHighestWeightReps(\OutgoingWeight)
  \,,
  \;\;\;\;\;
  i \,\in\, \{1, \cdots, \NumberOfPunctures\}
  \,,
\end{equation}
which, as representatives of $\slTwo$-coinvariants,
are subject to one linear relation
(Prop. \ref{slTwoCoinvariantsInTermsOfIncomingRepresentationVectors}):
\begin{equation}
  \label{fCoinvariantForSlTwoRe}
  \underset{
    1 \leq i \leq \NumberOfPunctures
  }{\sum}
  \;
  \LoweringOperator_i
  \cdot
  \vert
    \HighestWeightVector_1
    ,
    \cdots
    ,
    \HighestWeightVector_{\NumberOfPunctures}
  \rangle
  \;=\;
  0
  \;\;\;
  \in
  \;
  \big(
    \HighestWeightIrrep{\weight_1}
    \otimes
    \cdots
    \otimes
    \HighestWeightIrrep{\weight_{\NumberOfPunctures+1}}
  \big)_{\slTwo}
  \,.
\end{equation}

\vspace{-2mm}
\noindent By Prop. \ref{ConformalBlocksAsQuotientsOfCoInvariantOfSLTwo},
this already characterizes the conformal blocks,
unless $\IncomingWeight \,=\, \ShiftedLevel$
as in \eqref{TruncationConditionForOneProbe},
i.e., $\ShiftedLevel - 1 - \OutgoingWeight = 1$,
in which case there is one further linear relation:
\begin{equation}
  \label{TheTruncationConditionForOneProbe}
  \underset{
    1 \leq i \leq \NumberOfPunctures
  }{\sum}
  \,
  z_i
  \cdot
  \LoweringOperator_i
  \vert
    \HighestWeightVector_1
    , \cdots ,
    \HighestWeightVector_{\NumberOfPunctures}
  \rangle
  \;=\;
  0
  \;\;\;
  \in
  \;
  \ConformalBlocks^1_{\slTwoAffine{\Level}}
  (\vec \weight, \vec z)
  \,.
\end{equation}

\noindent
{\bf(iii)} For
$
{\bf
\NumberOfProbeBranes
=
2
}$,
the space
of incoming representation vectors
of total weight $\OutgoingWeight$
is spanned by the $\NumberOfPunctures^2$ elements
\begin{equation}
  \label{fOperatorInithAndjthPosition}
  \LoweringOperator_{i}
  \LoweringOperator_{j}
  \big\vert
    \HighestWeightVector_1
    ,
    \cdots
    ,
    \HighestWeightVector_{\NumberOfPunctures}
  \big\rangle
  \;\;\;
  \in
  \;\;
  \IncomingHighestWeightReps(\OutgoingWeight)
  \,,
  \;\;\;\;\;
  i, j \,\in\, \{1, \cdots, \NumberOfPunctures\}
  \,,
\end{equation}
which, as representatives of
$\slTwo$-coinvariants, are
subject to $\NumberOfPunctures$
relations:
\vspace{-2mm}
$$
  \sum_i \LoweringOperator_i
  \Big(
  \LoweringOperator_j
  \big\vert
    \HighestWeightVector_1
    ,
    \cdots
    ,
    \HighestWeightVector_{\NumberOfPunctures}
  \big\rangle
  \Big)
  \;\;
  =
  \;\;
  0
  \;\;\;\;
  \in
  \;
  \big(
    \HighestWeightIrrep{\weight_1}
    \otimes
    \cdots
    \otimes
    \HighestWeightIrrep{\weight_{\NumberOfPunctures + 1}}
  \big)_{\slTwo}
  \,,
  \;\;\;\;\;\;
  j \,\in\, \{1, \cdots, \NumberOfPunctures\}
  \,.
$$
\end{example}

The key observation now is that this Lie-algebraic definition of $\slTwoAffine{\Level}$-conformal blocks has a re-formulation purely in terms of the twisted cohomology of configurations of
$\NumberOfProbeBranes$
points in the punctured plane. For $\NumberOfProbeBranes = 1$, this is the following statement:

\begin{proposition}
[\phantom{.}$\slTwoAffine{\Level}$-Conformal blocks
in degree 1
as twisted holomorphic 1-cohomology
of punctured plane]
\label{Degree1ConformalBlocksInTwistedCohomology}
$\,$

\noindent
{\bf (i)} For integrable weights
$\vec \weight \,\in\, \{0,\cdots, \Level \}^\NumberOfPunctures$ \eqref{IntegrabilityConditionOnWeights},
the $[\FlatConnectionForm]$-twisted holomorphic de Rham cohomology \eqref{TwistedDeRhamIsomorphicToEquivariantDeRhamOnCoveringSpace}
of $\ComplexPlane \setminus \{z_1, \cdots, z_N\}$
with
$
  \FlatConnectionForm
  \,=\,
  \FlatConnectionForm(\vec w, \vec z)$
\eqref{TheMasterForm}
is concentrated in degree 1,
where it receives a natural transformation from the
$\widehat{\mathfrak{sl}(2)}_{\Level}$-conformal blocks
\eqref{slTwoConformalBlocksAsQuotientOfSlTwoCoinvariants}
in degree 1:
\vspace{-3mm}
\begin{equation}
  \label{IdentifyingDegreeOneConformalBlocksWithTwistedCohomology}
  \hspace{-1cm}
  \begin{tikzcd}[row sep=-2pt]
  \ConformalBlocks^1_{\slTwoAffine{\Level}}
    ( \vec{\weight},  \vec{z})
    \ar[rr]
    &&
    H^{1}
    \Big(
      \HolomorphicForms{\big}{\bullet}
      {
        \ComplexPlane
        \setminus
        \{ \vec z \}
      }
      ,\,
      \DolbeaultDifferential
      +
      \FlatConnectionForm
      ( \vec{\weight},  \vec{z})
      \wedge
    \Big)
    \\
    \hspace{37pt}
   \scalebox{0.86}{$
   f_i
   \vert
     \HighestWeightVector_{1}
     ,\cdots,
     \HighestWeightVector_{\NumberOfPunctures}
   \rangle
   \underset{
     \raisebox{-2pt}{
       \scalebox{.7}{
         \eqref{fOperatorInithPosition}
       }
     }
   }{
     \;\;
     =
     \;\;}
   \big[
      \HighestWeightVector_1
      ,\,
      \cdots
      ,\,
      ( f \cdot \HighestWeightVector_{{}_I} )
      ,\,
      \cdots
      ,\,
      \HighestWeightVector_{\NumberOfPunctures}
    \big]
    $}
    &
      \overset{
        \mathclap{
        \raisebox{2pt}{
          \tiny
          \color{greenii}
          \bf
          generators
        }
        }
      }{
        \longmapsto
      }
    &
    \scalebox{0.86}{$
      \big[
        -
        \frac{\weight_{{}_I}}{\ShiftedLevel}
        \frac{\Differential z}{(z - z_{{}_I})}
      \big]
      \hspace{60pt}
    $}
    \\
    \hspace{-56pt}
    \scalebox{0.86}{$
    { \displaystyle \sum_i}
      \,
      \LoweringOperator_{{}_I}
      \vert
        \HighestWeightVector_1
        ,\,
        \cdots
        ,\,
        \HighestWeightVector_{\NumberOfPunctures}
      \rangle
      \underset{
        \raisebox{1pt}{
          \scalebox{.7}{
            \eqref{fCoinvariantForSlTwoRe}
          }
        }
      }{
        \;=\;
      }
      0
    $}
    &
     \underset{
       \mathclap{
        \raisebox{-5pt}{
          \tiny
          \color{greenii}
          \bf
          relations
        }
        }
     }{
       \longmapsto
     }
    &
   \scalebox{0.86}{$
   {\displaystyle   \sum_I }
      \big[
        -
        \frac{\weight_{{}_I}}{\ShiftedLevel}
        \frac{\Differential z}{(z - z_{{}_I})}
      \big]
     \;=\;
     \big[
       (
         \DolbeaultDifferential
         +
         \FlatConnectionForm \wedge
       )
       1
     \big]
    $}
    \\
    \scalebox{.86}{$
     {\displaystyle \sum_i }
      z_{{}_I} \cdot \LoweringOperator_{{}_I}
      \vert
        \HighestWeightVector_1
        , \cdots ,
        \HighestWeightVector_{\NumberOfPunctures}
      \rangle
      \underset{
        \mathclap{
        \mbox{
          \tiny
          \rm
          \eqref{TheTruncationConditionForOneProbe}
        }
        }
      }{
        \;=\;
      }
      0
      \;\;
      \;
      \mbox{\rm (when $\IncomingWeight = \ShiftedLevel$)}
      \hspace{4pt}
    $}
    &\longmapsto&
    \scalebox{.86}{$
  { \displaystyle  \sum_I}
      \big[
        -
        z_{{}_I}
        \frac{\weight_{{}_I}}{\ShiftedLevel}
        \frac{\Differential z}{(z - z_{{}_I})}
      \big]
      \;=\;
      \big[
        (\DolbeaultDifferential
          +
        \FlatConnectionForm\wedge)
        z
      \big]
      \,.
    $}
  \end{tikzcd}
\end{equation}

\vspace{-3mm}
\noindent {\bf (ii)} This transformation is an {\it isomorphism},
at least when $\IncomingWeight < \ShiftedLevel$
(i.e., away from the truncation condition of Ex. \ref{TruncationConditionForSingleProbeBrane}).
\footnote{Even at the truncation condition $\IncomingWeight = \ShiftedLevel$ the map \eqref{IdentifyingDegreeOneConformalBlocksWithTwistedCohomology}
is claimed to be an injection in \cite[Rem. 3.4.3]{FeiginSchechtmanVarchenko94}, repeated in \cite[p. 107]{Kohno12};  but we are not aware of a proof.}
\end{proposition}
\begin{proof}
That we have a linear map
is the special case
$\NumberOfProbeBranes = 1$
of
{\cite[Cor. 3.4.2]{FeiginSchechtmanVarchenko94}}\footnote{
  Our $\NumberOfProbeBranes$ is their $q$.
  Our $\NumberOfPunctures$ is their $n$.
  Our $\weight_{{}_I}$ is their $m_i$.
  Our $\FlatConnectionForm$ is their $\alpha$.
}; we spell it out:
Noticing that all 1-forms here are twisted-closed already by degree reasons, to have a linear map as shown in \eqref{IdentifyingDegreeOneConformalBlocksWithTwistedCohomology} it suffices to see that the relations on the left are respected.
For the first relation in
\eqref{IdentifyingDegreeOneConformalBlocksWithTwistedCohomology}
this is immediate from \eqref{TheMasterForm}, as shown above.
To see that the more subtle second
relation is respected,
we may dramatically shortcut the
general argument in
\cite[\S 3.5]{FeiginSchechtmanVarchenko94}
as follows (cf. \cite[Ex. 12.4.6 (iii)]{Varchenko95}).
Using the assumption from \eqref{TruncationConditionForOneProbe} that
$\IncomingWeight = \ShiftedLevel$,
i.e., that
\begin{equation}
  \label{TruncationConditionForOneProbeAsTrivializationOfOutgoingPhase}
  \underset{
  \scalebox{0.6}{$  1 \leq I \leq \NumberOfPunctures $}
  }{\sum}
  \,
  \frac{\weight_{{}_I}}{\ShiftedLevel}
  \;=\;
  1
  \,,
\end{equation}

\vspace{-2mm}
\noindent we directly compute:
$$
  \def\arraystretch{1.5}
  \begin{array}{lll}
    (\DolbeaultDifferential
      +
     \FlatConnectionForm
     \wedge
    )z
    &
    \;=\;
    {\displaystyle   \sum_I }
    - z
    \frac{\weight_{{}_I}}{\ShiftedLevel}
    \frac{\Differential z}{(z-z_{{}_I})}
    \;+\;
    \Differential z
    &
    \proofstep{
      by \eqref{TheMasterForm}
    }
    \\
    & \;=\;
  {\displaystyle   \sum_I }
    - z
    \frac{\weight_{{}_I}}{\ShiftedLevel}
    \frac{\Differential z}{(z-z_{{}_I})}
    \;+\;
   {\displaystyle   \sum_i }
    \frac{\weight_{{}_I}}{\ShiftedLevel}
    \Differential z
    &
    \proofstep{
      by
      \eqref{TruncationConditionForOneProbeAsTrivializationOfOutgoingPhase}
    }
    \\
    & \;=\;
   {\displaystyle   \sum_I }
    - z
    \frac{\weight_{{}_I}}{\ShiftedLevel}
    \frac{\Differential z}{(z-z_{{}_I})}
    \;+\;
    {\displaystyle   \sum_I }
    (z - z_{{}_I})
    \frac{\weight_i}{\ShiftedLevel}
    \frac{\Differential z}
    { (z - z_{{}_I}) }
    \\
    & \;=\;
    {\displaystyle   \sum_I }
    - z_{{}_I}
    \frac{\weight_{{}_I}}{\ShiftedLevel}
    \frac{\Differential z}{(z-z_{{}_I})}
    \,.
  \end{array}
$$
This shows that
\eqref{IdentifyingDegreeOneConformalBlocksWithTwistedCohomology}
does define a natural linear transformation for all values of $\IncomingWeight$.
Finally, that this linear map is an isomorphism when $\IncomingWeight < \ShiftedLevel$ follows by Prop. \ref{OSBasisForTwistedCohomologyOfPuncturedPlane}.
\end{proof}

This result now has a natural re-formulation:

\begin{proposition}
[\phantom{.}$\slTwoAffine{\Level}$-Conformal blocks as secondary Chern classes in \TED-K-theory of
Riemann surface inside $\mathbb{A}$-type singularities]
\label{Degree1ConformalBlockInTEdKTheory}
For
integrable weights
$\vec \weight \,\in\, \{0,\cdots, \Level \}^\NumberOfPunctures$ \eqref{IntegrabilityConditionOnWeights},
the secondary
\TED-K-theory from Prop. \ref{TheBareBProfiniteIntegersTwistOfComplexRationalizedAEquivariantKTheory}
naturally receives
the space of
$\widehat{\mathfrak{sl}(2)}_{\Level}$-conformal blocks
at level $\Level = -2 + \ShiftedLevel$
in degree $n = 1$:
\begin{equation}
  \label{DegreeOneConformalBlocksInsideTEdKTheory}
  \begin{tikzcd}[column sep=7pt]
    \overset{
      \mathclap{
      \raisebox{2pt}{
        \tiny
        \color{darkblue}
        \bf
        \begin{tabular}{c}
          coformal blocks of the
          $\slTwo$-WZW model
          \\
          on the sphere,
          at level $\ShiftedLevel-2$,
          in degree 1
        \end{tabular}
      }
      }
    }{
    \ConformalBlocks^1
      _{\slTwoAffine{-2 + \ShiftedLevel}}
      (\vec \weight, \ShiftedLevel)
    }
    \ar[
      rr,
      "{
        \mbox{
          \tiny
          \color{greenii}
          \bf
          \def\arraystretch{.9}
          \begin{tabular}{c}
            natural
            \\
            transformation
          \end{tabular}
        }
      }"{swap, yshift=-2pt}
    ]
    &&
    H^{1}
    \Big(
      \HolomorphicForms{\big}{\bullet}
      {
        \ComplexPlane
        \setminus
        \{ \vec z \}
      }
      ,\,
      \DolbeaultDifferential
      +
      \FlatConnectionForm
      ( \vec{\weight},  \vec{z})
      \wedge
    \Big)
    \ar[
      rr,
      "{
        \mbox{
          \tiny
          \color{greenii}
          \bf
          \def\arraystretch{.9}
          \begin{tabular}{c}
            secondary
            \\
            Chern character
          \end{tabular}
        }
      }"{swap, yshift=-2pt}
    ]
    &&
    \overset{
      \mathclap{
      \raisebox{3pt}{
        \tiny
        \color{darkblue}
        \bf
        \begin{tabular}{c}
          secondary
          \TED-K-theory of
          punctured sphere
          inside $\mathbb{A}_{\ShiftedLevel}$-orbi-singularity
        \end{tabular}
      }
      }
    }{
    \mathrm{KU}
      ^{
        0
          +
        [
          \FlatConnectionForm
          (\vec \weight, \ShiftedLevel)
        ]
      }
      _{\diff}
    \Big(
      (
        \RiemannSphere
        \setminus
        \{\vec z, \infty\}
      )
      \times
      \HomotopyQuotient
        { \ast }
        { \CyclicGroup{\kappa} }
    \Big)
    }
    \,.
  \end{tikzcd}
\end{equation}
\end{proposition}
\begin{proof}
In view of Rem. \ref{TwistedCohomologyOverSteinManifolds},
this is the result of combining
Prop. \ref{Degree1ConformalBlocksInTwistedCohomology}
 with
 Prop. \ref{TheBareBProfiniteIntegersTwistOfComplexRationalizedAEquivariantKTheory}.
\end{proof}

We proceed to the generalization of this
situation away from the special case where conformal blocks are just in degree $1$.

\medskip

\noindent
{\bf The configuration spaces.}

\vspace{-3mm}
\begin{itemize}[leftmargin=*]
\setlength\itemsep{-1pt}

\item
The
complex manifold of configurations of $\NumberOfProbeBranes$ ordered points in some complex manifold $\ComplexManifold$
is denoted
(following \cite[\S 2.2]{SS19ConfigurationSpaces})
\vspace{-1mm}
\begin{equation}
  \label{ConfigurationSpaceOfPointsInThePlane}
  \ConfigurationSpace
    {\NumberOfProbeBranes}
  \big(
    \ComplexManifold
  \big)
  \;:=\;
  \Big\{
    z^1,
    \cdots,
    z^{\NumberOfProbeBranes}
    \,\in\,
    \ComplexManifold
    \;\big\vert\;
    \underset{i < j}{\forall} \;\;
    z^i \neq z^j
  \Big\}.
\end{equation}

\vspace{-3mm}
\noindent The punctured plane $\Sigma^2 \,=\, \ComplexPlane \setminus \{z_1, \cdots z_{\NumberOfPunctures}\}$
itself may be understood as one such configuration of $\NumberOfPunctures$ points in the complex plane.
We are interested in configurations of
(further)
$\NumberOfProbeBranes$ points inside this $\NumberOfPunctures$-punctured plane, which is usefully understood as the following fiber product of configuration spaces (the top right entry denotes the singleton set containing the given configuration of punctures):
\vspace{-4mm}
\begin{equation}
  \label{ConfigurationSpaceAsFiberProduct}
  \begin{tikzcd}[row sep=18pt]
    \overset{
      \mathclap{
      \raisebox{2pt}{
        \tiny
        \color{darkblue}
        \bf
        \def\arraystretch{.9}
        \begin{tabular}{c}
          configurations of
          $\NumberOfProbeBranes$
          points
          in the
          $\NumberOfPunctures$-punctured
          plane
        \end{tabular}
      }
      }
    }{
    \ConfigurationSpace
      {\NumberOfProbeBranes}
    \big(
      \ComplexPlane
      \setminus
      \{z_1, \cdots, z_{\NumberOfPunctures}\}
    \big)
    }
    \ar[d]
    \ar[rr]
    \ar[
      drr,
      phantom,
      "{\mbox{\tiny\rm(pb)}}"{pos=.35}
    ]
    &&
    \big\{
      \ComplexPlane
        \setminus
      \{z_1, \cdots, z_{\NumberOfPunctures}\}
    \big\}
    \ar[
      d,
      "{
        \mbox{
          \tiny
          \color{greenii}
          \bf
          \def\arraystretch{.9}
          \begin{tabular}{c}
           pick
           \\
           the configuration
           \\
           of puncturres
          \end{tabular}
        }
      }"{xshift=-8pt}
    ]
    \\
    \ConfigurationSpace
      {
        \NumberOfPunctures
        +
        \NumberOfProbeBranes
      }
    \big(
      \ComplexPlane
    \big)
    \ar[
      rr,
      "{
        \mbox{
          \tiny
          \color{greenii}
          \bf
          forget the last
          $\NumberOfProbeBranes$
          points
        }
      }"{swap}
    ]
    &&
    \ConfigurationSpace
      { \NumberOfPunctures }
    \big(
      \ComplexPlane
    \big)
  \end{tikzcd}
\end{equation}

\vspace{-3mm}
The previous discussion may be understood as concerning the degenerate case of configurations of a single point:
$$
  \underset{
    \;\,\scalebox{.7}{$\{1\}$}
  }{
    \mathrm{Conf}
  }
  \big(
    \ComplexPlane
    \setminus
    \{z_1, \cdots, z_{\NumberOfPunctures}\}
  \big)
  \;\;\;\simeq\;\;\;
  \ComplexPlane
  \setminus
  \{z_1, \cdots, z_{\NumberOfPunctures}\}
  \,.
$$

\vspace{-2mm}
In this vein, much of the previous discussion has evident generalizations, such as:

\item
The flat holomorphic connection 1-form
\eqref{TheMasterForm} is generalized to the
configuration space by setting:
\begin{equation}
  \label{TheMasterFormForSeveralProbes}
   \FlatConnectionForm
   (
     \vec \weight
     ,
     \kappa
   )
   \;:=\;
  \underset{
    \scalebox{0.6}{$
      \def\arraystretch{.9}
      \begin{array}{c}
      1
        \,\leq\,
      I
        \,\leq\,
      \NumberOfPunctures
      \\
      1
        \,\leq\,
      i
        \,\leq\,
      \NumberOfProbeBranes
      \end{array}
    $}
  }{\sum}
  -
  \frac
    {\weight_{{}_{I}}}
    {\kappa}
  \frac{
    \Differential z
  }{
    (
      z^i - z_{{}_I}
    )
  }
  \;+\;
  \underset{
    \scalebox{0.6}{$
      1
        \leq
      i
        <
      j
        \leq
      \NumberOfProbeBranes
    $}
  }{\sum}
  \frac{2}{\kappa}
  \frac
    { \Differential z }
    {
      (
        z^i
        -
        z^j
      )
    }
  \;\;\;\;\;
  \in
  \;
  \HolomorphicForms{\Big}{1}
  {
    \ConfigurationSpace
      {\NumberOfProbeBranes}
    \big(
      \ComplexPlane
      \setminus
      \{ z_1, \cdots, z_N \}
    \big)
  }
  \,.
\end{equation}

\vspace{-2mm}
\item Accordingly, the master function \eqref{TheMasterFunction} is generalized to the
following holomorphic function on the universal cover of $\ConfigurationSpace{\NumberOfProbeBranes}\big(\ComplexPlane \setminus \{1, \cdots, z_{\NumberOfPunctures}\}\big)$:
\vspace{-3mm}
\begin{equation}
  \label{TheMasterFunction2}
  \MasterFunction(
    \hat{z}^1,
    \cdots,
    \hat{z}^{\, \NumberOfProbeBranes}
    \,
  )
  \;:=\;
  \MasterFunction
  (\vec \weight, \kappa)
  (
    \hat{z}^1,
    \cdots,
    \hat{z}^{\, \NumberOfProbeBranes}
  )
  \;:=\quad
  \underset{
    \mathclap{
     \scalebox{0.6}{$
       \def\arraystretch{.9}
       \begin{array}{l}
       1
         \leq
       I
         \leq
       \NumberOfPunctures
       \\
       1
         \leq
       i
         \leq
       \NumberOfProbeBranes
       \end{array}
    $}
    }
  }{\prod}
  \;\;\;
  \big(
    \hat{z}^{\,i} - z_{{}_I}
  \big)^{ -\weight_{{}_I}/\kappa }
  \quad
  \underset{
    \mathclap{
     \scalebox{0.6}{$
       \def\arraystretch{.9}
       \begin{array}{l}
       1
         \leq
       i
         <
       j
         \leq
       \NumberOfPunctures
       \end{array}
    $}
    }
  }{\prod}
  \quad
  \big(
    \hat{z}^{\, i} - \hat{z}^j
  \big)^{ 2/\kappa }
  \,.
\end{equation}

\end{itemize}

\vspace{-.2cm}

Now the general form of Prop. \ref{Degree1ConformalBlocksInTwistedCohomology} is:
\begin{proposition}
[\phantom{.}$\slTwoAffine{\Level}$-Conformal blocks
as twisted holomorphic cohomology
of configuration space of punctured plane
{\cite[Cor. 3.4.2, Rem. 3.4.3]{FeiginSchechtmanVarchenko94}}]
\label{ConformalBlocksInTwistedCohomology}
For
$\Level \,\in\, \NaturalNumbers$
and
integrable weights
$\vec \weight \,\in\, \{0,\cdots, \Level \}^\NumberOfPunctures$
\eqref{IntegrabilityConditionOnWeights},
the $[\FlatConnectionForm(\vec w, \ShiftedLevel)]$-twisted
holomorphic cohomology
of the configuration space
$\ConfigurationSpace{\NumberOfProbeBranes}
\big(\ComplexPlane \setminus \{z_1, \cdots, z_N\}\big)$
\eqref{ConfigurationSpaceAsFiberProduct}
is concentrated in degree
$\NumberOfProbeBranes$,
where it naturally contains
the space of
degree=$\NumberOfProbeBranes$
$\widehat{\mathfrak{sl}(2)}_{\Level}$-conformal blocks
\eqref{slTwoConformalBlocksAsQuotientOfSlTwoCoinvariants}:
\vspace{-3mm}
\begin{equation}
  \label{IdentifyingConformalBlocksWithTwistedCohomologyForGeneralNumberOfProbes}
  \begin{tikzcd}[column sep=10pt, row sep=-2pt]
    \overset{
    }{
      \ConformalBlocks
        ^{\NumberOfProbeBranes}
        _{\slTwoAffine{\Level}}
      (\vec \weight, \vec z)
    }
    \;\simeq\;
    \IncomingHighestWeightReps
      (\IncomingWeight - 2\NumberOfProbeBranes)
    \big/(\cdots)
    \ar[rr, hook]
    &&
    H^{\NumberOfProbeBranes}
    \bigg(
      \HolomorphicForms{\Big}{\bullet}
      {
        \ConfigurationSpace
          {\NumberOfProbeBranes}
        \big(
          \ComplexPlane
          \setminus
          \{ \vec z \}
        \big)
      }
      ,\,
      \DolbeaultDifferential
      +
      \FlatConnectionForm
      (\vec \weight, \ShiftedLevel)
      \wedge
    \bigg)
    \,.
    \\
    \mathllap{
      \mbox{ e.g.}
      \hspace{10pt}
    }
    \hspace{10pt}
   \scalebox{0.86}{$
   f_{{}_I}
   f_{{}_J}
   \vert
     \HighestWeightVector_{1}
     \cdots,
     \HighestWeightVector_{\NumberOfPunctures}
   \rangle
   \underset{
     \raisebox{-2pt}{
       \scalebox{.7}{
         \eqref{fOperatorInithAndjthPosition}
       }
     }
   }{
     \;\;
     =
     \;\;}
   \big[
      \cdots
      ,\,
      ( f \cdot \HighestWeightVector_{{}_I} )
      ,\,
      \cdots
      ,\,
      ( f \cdot \HighestWeightVector_{{}_J} )
      ,\,
      \cdots
    \big]
    $}
    &
      \overset{
        \mathclap{
        \raisebox{2pt}{
          \tiny
          \color{greenii}
          \bf
          generators
        }
        }
      }{
        \longmapsto
      }
    &
    \scalebox{0.86}{$
      \Big[
        \frac{\weight_{{}_I}}{\ShiftedLevel}
        \frac{\Differential z}{(z^1 - z_{{}_I})}
        \wedge
        \frac{\weight_{{}_J}}{\ShiftedLevel}
        \frac{\Differential z}{(z^2 - z_{{}_J})}
      \Big]
      \hspace{30pt}
    $}
  \end{tikzcd}
\end{equation}
\end{proposition}

In the evident generalization of Prop. \ref{Degree1ConformalBlockInTEdKTheory},
using the above, we have the following:

 \newpage

\begin{theorem}
[\phantom{.}$\slTwoAffine{\Level}$-Conformal blocks as secondary Chern classes in
\TED-K-theory of
configurations inside
$\mathbb{A}$-type singularities]
\label{ConformalBlockInTEdKTheory}
For integrable weights
$\vec \weight \,\in\, \{0,\cdots, \Level \}^\NumberOfPunctures$
\eqref{IntegrabilityConditionOnWeights}
and for any $\NumberOfProbeBranes \,\in\, \NaturalNumbers_+$,
the secondary Chern classes in \TED-K-theory of Prop. \ref{TheBareBProfiniteIntegersTwistOfComplexRationalizedAEquivariantKTheory}
of the orbifold configuration space
$
\ConfigurationSpace{\NumberOfProbeBranes}
\big(\ComplexPlane \setminus \vec z\big) \times \HomotopyQuotient{\ast}{\CyclicGroup{\kappa}}
$
naturally receives
the space of $\slTwoAffine{\Level}$-conformal blocks
at level $\Level = -2 + \ShiftedLevel$
in degree $\NumberOfProbeBranes$:
\vspace{-4mm}
\begin{equation}
  \label{ConformalBlocksInsideTEdKTheory}
  \begin{tikzcd}[column sep=26pt]
  \overset{
    \mathclap{
      \raisebox{3pt}{
        \tiny
        \color{darkblue}
        \bf
        \begin{tabular}{c}
          conformal blocks of
          the $\suTwo$
          WZW model
        \end{tabular}
      }
    }
  }{
    \ConformalBlocks
      ^\NumberOfProbeBranes
      _{\slTwoAffine{-2 + \ShiftedLevel}}
    (\vec \weight, \vec z)
    }
  \ar[
    rr,
    "{
      \mbox{
        \tiny
        \color{greenii}
        \bf
        \begin{tabular}{c}
          natural
          \\
          transformation
        \end{tabular}
      }
    }"{swap}
  ]
  &&
    \overset{
      \mathclap{
      \raisebox{3pt}{
        \tiny
        \color{darkblue}
        \bf
        \begin{tabular}{c}
          \TED-K-theory
          of configurations of points
          in complex curve inside
          $\mathbb{A}$-type singularity
        \end{tabular}
      }
      }
    }{
    \mathrm{KU}
      ^{
        \NumberOfProbeBranes - 1
        +
        [
          \FlatConnectionForm
          (\vec \weight, \ShiftedLevel)
        ]
      }
      _{\diff}
    \Big(
      \ConfigurationSpace
        {\NumberOfProbeBranes}
      (
        \RiemannSphere
        \setminus
        \{\vec z, \infty\}
      )
      \times
      \HomotopyQuotient
        { \ast }
        { \CyclicGroup{\kappa} }
    \Big)
    }
    \,.
  \end{tikzcd}
\end{equation}
\end{theorem}

\noindent
{\bf $\slTwoAffine{\Level}$-CFT beyond admissible weights at integral level.}
It is clear that the \TED-K-theory group on the right of \eqref{ConformalBlocksInsideTEdKTheory} is larger than the space of standard conformal blocks included on the left, which suggests that the \TED-theory of configurations inside $\mathbb{A}$-type singularities
reflects yet further aspects of conformal field theory.
This is argued in Rem. \ref{FurtherContentSeenInTheTEdKTheory} below, for which we now briefly review some relevant background.

\begin{remark}[\bf Fractional-level WZW models]
\label{FractionalLevelWZWModels}
$\,$

\noindent
{\bf (i)}
The twisted cohomology on the right of \eqref{IdentifyingConformalBlocksWithTwistedCohomologyForGeneralNumberOfProbes} certainly makes sense for more general choices of parameters than assumed in Prop. \ref{ConformalBlocksInTwistedCohomology} (discussion in this generality goes back to \cite{SchechtmanVarchenko90}\cite{SchechtmanVarchenko91}, review is in \cite{Varchenko95}\cite[\S 7]{EtingofFrenkelKirillov98}) --  notably it makes sense at {\it admissible fractional levels}
$\Level = -2 + \ShiftedLevel/\Denominator$ (cf. Rem. \ref{RationalLevelsInKTheory}).

\noindent
{\bf (ii)}
Also the $\slTwoAffine{\Level}$-CFT model is known to make sense for
admissible fractional levels
(e.g. \cite[(1.4)]{FGPP93}\cite[(1)]{PetersenRasmussenYu96}\cite[\S 3.1]{Rasmussen20})
if also the weights are ``admissible'', which is the case \eqref{AdmissibilityCondition} for the integrable integer weights considered above. In fact, these admissible fractional-level models have been argued to be building blocks of logarithmic CFTs
(\cite[(3.1)]{CreutzigRidoutII13}\cite[(5.1)]{KawasetsuRidout19}\cite[(1.2)]{KawesetsuRidout22}, review in \cite{Ridout10}\cite{Ridout20}).

\noindent
{\bf (iii)} It is is expected \footnote{We thank
D. Ridout for discussion of this point.}
but remains unproven (cf. \cite[bottom of p. 35]{FGPP93})
that a suitable notion of conformal blocks for
these fractional-level WZW models
exists
and that Prop. \ref{ConformalBlocksInTwistedCohomology} holds true in this greater generality.
\end{remark}
We record this as:
\begin{conjecture}
[\phantom{.}$\slTwoAffine{\Level/r}$-Conformal blocks
as twisted holomorphic cohomology
of configuration space of punctured plane]
\label{FractionalLevelConformalBlocksInTwistedCohomology}
For
$\ShiftedLevel \,\in\, \NaturalNumbers_{\geq 2}$,
$r \,\in\, \{1, \cdots, \ShiftedLevel\}$
and
$\vec \weight \,\in\, \{0,\cdots, \Level \}^\NumberOfPunctures$,
the $[\FlatConnectionForm]$-twisted
holomorphic de Rham cohomology
of the configuration space
$\ConfigurationSpace{\NumberOfProbeBranes}
\big(\ComplexPlane \setminus \{z_1, \cdots, z_N\}\big)$
\eqref{ConfigurationSpaceAsFiberProduct}
with
$
  \FlatConnectionForm
  \,=\,
  \FlatConnectionForm(\vec w, \ShiftedLevel)$
\eqref{TheMasterFormForSeveralProbes}
is concentrated in degree
$\NumberOfProbeBranes$,
where it naturally contains
the space of
degree=$\NumberOfProbeBranes$
$\widehat{\mathfrak{sl}(2)}_{-2 + \ShiftedLevel/r}$-conformal blocks
\eqref{slTwoConformalBlocksAsQuotientOfSlTwoCoinvariants}:
\vspace{-4mm}
\begin{equation}
  \label{IdentifyingConformalBlocksWithTwistedCohomologyForFractionalLevel}
  \begin{tikzcd}[column sep=20pt, row sep=-2pt]
    \ConformalBlocks
      ^{\NumberOfProbeBranes}
      _{\slTwoAffine{-2+\ShiftedLevel/{\color{purple}r}}}
    (\vec \weight, \vec z)
    \; \ar[rr, hook]
    &&
    \;
    H^{\NumberOfProbeBranes}
    \bigg(
      \HolomorphicForms{\big}{\bullet}
      {
        \ConfigurationSpace
          {\NumberOfProbeBranes}
        \big(
          \ComplexPlane
          \setminus
          \{ \vec z \}
        \big)
      }
      ,\,
      \DolbeaultDifferential
      +
      \FlatConnectionForm
      (
        \vec \weight,
        \ShiftedLevel/{\color{purple}r}
      )
      \wedge
    \bigg)
    \,.
  \end{tikzcd}
\end{equation}
\end{conjecture}

\begin{remark}[\bf The logarithmic-admissible highest-weight irrep]
  \label{TheLogarithmicAdmissibleHighestWeightRep}
  Since the hypergeometric SV-construction of Prop. \ref{Degree1ConformalBlocksInTwistedCohomology} makes manifest that the conformal blocks depend on the weights only modulo the shifted level $\ShiftedLevel$ (because they enter the construction solely through the twisting holonomy $[\weight/\ShiftedLevel] \,\in\, \CyclicGroup{\ShiftedLevel}$ in \eqref{WeightsMatterModuloShiftedLevel}, \eqref{TheMasterForm}, see also \cite[(3.2)]{Rasmussen20}), it is curious that the admissible weights
  $\weight = 0, 1, \cdots, \Level$
  \eqref{AdmissibilityCondition}
  only
  {\it almost}
  exhaust the non-redundant available range:
  One single non-admissible class of weights
  remains:
  $
    [\Level + 1] \,=\,  [\ShiftedLevel - 1]
    \;\;
    \in
    \;
    \CyclicGroup{\ShiftedLevel}\;.
  $
  It has been argued
  in \cite{Rasmussen20}\footnote{In \cite{Rasmussen20} this non-admissible highest-weight irrep, or its character, corresponds to the case $r = p,\, s = 0$ in (4.17),  hence
  $\ell=0, s_0 = 0$ in (3.55), see also (4.2), hence to the boxes marked by a red upper corner in Figure 1 there. }
  that this single non-admissible weight does correspond to an irreducible highest-weight
  {\it affine Kac module} which is to be included in the context of logarithmic CFT (cf. Rem. \ref{FractionalLevelWZWModels}).
  In conclusion:
\end{remark}
\begin{remark}[\bf Fractional-level CFT and
logarithmic-admissible modules
possibly
seen
in \TED-K-cohomology]
\label{FurtherContentSeenInTheTEdKTheory}
Apart from the cokernel of the map \eqref{ConformalBlocksInsideTEdKTheory},
the differential \TED-K-cohomology group
in
Theorem \ref{ConformalBlockInTEdKTheory}
is larger than just one copy of the conformal blocks of the given integrable weights at the integral level $\Level = \ShiftedLevel - 2$:

\item {\bf (i) Fractional-levels.} By Rem. \ref{RationalLevelsInKTheory}, the secondary \TED-K-cohomology
groups are direct sums indexed by $r \,\in\, \{1, \cdots, \ShiftedLevel \}$
of the twisted cohomology, as on the right of \eqref{IdentifyingDegreeOneConformalBlocksWithTwistedCohomology},
but at all fractional levels $k = -2 + \ShiftedLevel/r$
\eqref{FractionalLevel}
which are {\it admissible} (in that $\ShiftedLevel \geq 2$, see Rem. \ref{RationalLevelsInKTheory}).
Hence if Conjecture \ref{FractionalLevelConformalBlocksInTwistedCohomology} holds,
then the statement of Prop. \ref{Degree1ConformalBlockInTEdKTheory} enhances to the stronger statement that secondary \TED-K-theory naturally contains  the conformal blocks of all admissible fractional-level WZW models at once:
\vspace{-.1cm}
\begin{equation}
  \label{FractionalLevekDegreeNConformalBlocksInsideTEdKTheory}
    \overset{
      \mathclap{
      \raisebox{6pt}{
        \tiny
        \color{darkblue}
        \bf
        \begin{tabular}{c}
          direct sum of
          conformal blocks
          \\
          of
          all
          $\ShiftedLevel$-fractional-level
          $\slTwo$-WZW models
        \end{tabular}
      }
      }
    }{
    \underset{
      1 \leq {\color{purple}r} \leq \ShiftedLevel
    }{\bigoplus}
    \ConformalBlocks
      ^\NumberOfProbeBranes
      _{\slTwoAffine{- 2 + \ShiftedLevel/{\color{purple}r}}}
    (\vec \weight, \vec z)
    }
    \;\;\;
    \xrightarrow
      {
        \phantom{--}
        \mbox{
          \tiny
          \color{greenii}
        \bf   natural transformation
        }
        \phantom{--}
      }
    \;\;\;
    \overset{
      \mathclap{
      \raisebox{3pt}{
        \tiny
        \color{darkblue}
        \bf
        \TED-K-theory of
        configurations of points in
        complex curve inside
        $\mathbb{A}$-type singularity
      }
      }
    }{
    \mathrm{KU}
      ^{
        \NumberOfProbeBranes
        -
        1
        +
        [
          \FlatConnectionForm
          (\vec \weight, \ShiftedLevel)
        ]
      }
      _{\diff}
    \Big(
      \ConfigurationSpace
        {\NumberOfProbeBranes}
      \big(
        \RiemannSphere
        \setminus
        \{\vec z, \infty\}
      \big)
      \times
      \HomotopyQuotient
        { \ast }
        { \CyclicGroup{\kappa} }
      ;\,
      \ComplexNumbers
    \Big)
    }
    \,.
\end{equation}

\item {\bf (ii) The logarithmic-admissible module.} In addition,
the
\TED-K-theory is defined also for twists that  correspond to  ``non-admissible'' weights
$\weight_i \,=\, \ShiftedLevel - 1$, which by Rem. \ref{TheLogarithmicAdmissibleHighestWeightRep} corresponds to a sector that one does want to include in the context of logarithmic CFT.
\end{remark}

\medskip

\section{TED-K-Theory of $\mathbb{A}$-type singularities}
\label{OneTwistedCohomologyAsTwistedEquivariantKTheory}

The aim of this section is to give an accurate but leisurely informal description of twisted equivariant differential (\TED) K-theory (full definitions and proofs are relegated to \cite{SS22TEC}\cite{SS22TED}),
with goal and focus on transparently bringing out its peculiar twisting, appearing on $\CyclicGroup{\ShiftedLevel}$-fixed loci,
by  $\CharacterGroup{\CyclicGroup{\ShiftedLevel}}$-cohomology in degree 1.

\medskip
In themselves, these twists are known, as ``inner local systems'' \cite{Ruan00},
but available descriptions
of how these act
(\cite{TuXu06}\cite{FreedHopkinsTeleman02ComplexCoefficients})
on equivariant K-theory have been somewhat indirect. To the end of making the phenomenon fully transparent, we use the higher geometric language laid out in \cite{FSS20CharacterMap}\cite{SS20OrbifoldCohomology}\cite{SS22EPB}.

\medskip

The upshot of the discussion here is to explain the simple but archetypical example of how \TED-K-theory of $\CyclicGroup{\ShiftedLevel}$-singularities with vanishing 3-twist
is, in its secondary Chern characters, equivalent to direct sums of ordinary (holomorphic) cohomology groups twisted by any tensor power of a given flat holomorphic connection -- which is the statement used as Prop. \ref{TheBareBProfiniteIntegersTwistOfComplexRationalizedAEquivariantKTheory} in \cref{ConformalBlocksAsTEdKTheory} in order to see the natural appearance of conformal blocks into \TED-K-theory.

\medskip

\noindent
{\bf Twisted $\mathrm{KU}$-Theory} (following \cite[Ex. 1.3.19]{SS22EPB}).

\vspace{-2mm}
\begin{itemize}[leftmargin=*]
\setlength\itemsep{-1pt}

\item
$\HilbertSpace$ denotes any countably-infinite dimensional complex Hilbert space, concretely that of square-summable sequences of complex numbers.


\item
$\FredholmOperators \,:=\, \FredholmOperators\big( \CyclicGroup{2} \times \HilbertSpace \big)$ denotes the space of Fredholm operators on this Hilbert space, meaning the set of odd-graded self-adjoint bounded operators on $\CyclicGroup{2} \times \HilbertSpace \,\simeq\, \HilbertSpace \oplus \HilbertSpace$
\begin{equation}
  \label{AFredholmOperator}
  \left(
  \!\!\!
  \begin{array}{cc}
    0 & \FredholmOperator
    \\
    \FredholmOperator^\dagger
    & 0
  \end{array}
  \!\!\!
  \right)
  \;\;\;
  \in
  \;\;
  \BoundedOperators
  \big(
    \CyclicGroup{2}
    \times
    \HilbertSpace
  \big)
  \qquad
  \mbox{hence}
  \qquad
  \begin{tikzcd}
    \HilbertSpace
    \ar[
      rr,
      shift left=2.6pt,
      "{\FredholmOperator}"
    ]
    \ar[
      from=rr,
      shift left=2.6pt,
      "{\FredholmOperator^\dagger}"
    ]
    &&
    \HilbertSpace
    \,,
  \end{tikzcd}
\end{equation}

\vspace{-2mm}
\noindent whose square differs from the identity by a compact operator:
$\FredholmOperator \circ \FredholmOperator^\dagger - 1,\,\; \FredholmOperator^\dagger \circ \FredholmOperator - 1 \;\;\;\in\; \CompactOperators(\HilbertSpace)$.

This implies that $\FredholmOperator$ and $\FredholmOperator^\dagger$ are isomorphisms up to a finite-dimensional kernel
and finite-dimensional co-kernel
$$
  \FredholmOperator
  \;\;
  :
  \;\;
  \HilbertSpace
  \;\simeq
  \!\!\!
  \begin{tikzcd}[row sep=-3pt, column sep=large]
    \HilbertSpace
    &&
  \;\;  \HilbertSpace
    \\
 \quad   \oplus \quad
    \ar[
      rr,
      "{
        \left(\!\!\!
        \hspace{-2pt}
        \begin{array}{cc}
          \mathrm{id}
          \\
          0
        \end{array}
        \hspace{-6pt}
        \begin{array}{cc}
          0
          \\
          0
        \end{array}
        \hspace{-2pt}
       \!\!\! \right)
      }"{description}
    ]
    &&
  \qquad   \oplus \qquad
    \\
  \;\;  \mathrm{ker}(\FredholmOperator)
    &&
    \mathrm{coker}(\FredholmOperator)
  \end{tikzcd}
  \!\!\!\!\!
  \simeq
  \;
  \HilbertSpace
  \,,
$$
which thus constitute a finite-dimensional  $\CyclicGroup{2}$-graded virtual vector space:
\begin{equation}
  \label{VirtualKernelOfFredholmOperator}
  \big(
    \mathrm{ker}(\FredholmOperator)
    \oplus
    \mathrm{ker}(\FredholmOperator^\dagger)
  \big)
  \;\ominus\;
  \big(
    \mathrm{coker}(\FredholmOperator)
    \oplus
    \mathrm{coker}(\FredholmOperator^\dagger)
  \big)
  \,.
\end{equation}

\item
$\GradedUH \,:=\, \UnitaryGroup\big( \CyclicGroup{2} \times \HilbertSpace \big)$ denotes the
corresponding unitary group,

\item $\GradedPUH \,:=\, \GradedUH/\CircleGroup$ denotes the corresponding projective unitary group, and

\item
$\PUH \,\subset\, \GradedPUH$ denotes its subgroup of 0-graded projective unitary operators.

\item
$\GradedUH \,\acts\, \FredholmOperators$
denotes the conjugation action $\FredholmOperator \,\mapsto\, \UnitaryOperator \circ \FredholmOperator \circ \UnitaryOperator^{\dagger} $,  and

\item
$\GradedPUH \acts \, \FredholmOperators$ denotes its passage to the projective quotient (since conjugation by $\ComplexNumber \in \CircleGroup \hookrightarrow \GradedUH$ is trivial).

\item
$\HomotopyQuotient{\FredholmOperators}{\GradedPUH}$ denotes the homotopy quotient of this action, formed in topological stacks.

Equipped with the canonical projection
\begin{equation}
  \label{UniversalFredFiberBundle}
  \begin{tikzcd}[column sep=large]
    \FredholmOperators
    \ar[r]
    &
    \HomotopyQuotient
      { \FredholmOperators }
      { \PUH }
    \ar[r]
    \ar[
      d,
      "{
        \HomotopyQuotient
          {\mathrm{pr}}
          {\PUH}
      }"{left}
    ]
    \ar[
      dr,
      phantom,
      "{\mbox{\tiny(pb)}}"
    ]
    &
    \HomotopyQuotient
      { \FredholmOperators }
      { \GradedPUH }
    \ar[
      d,
      "{
        \HomotopyQuotient
          {\mathrm{pr}}
          {\GradedPUH}
      }"
    ]
    \\
    &
    \mathbf{B} \PUH
    \ar[r]
    &
    \mathbf{B} \GradedPUH
  \end{tikzcd}
\end{equation}
this is the universal $\FredholmOperators$-fiber bundle over the topological {\it moduli stack}
$\mathbf{B}\GradedPUH \,:=\, \HomotopyQuotient{\ast}{\GradedPUH}$ of $\GradedPUH$-principal bundles.

\item
$\mathrm{KU}^{[\tau]}(\TopologicalSpace)
\,:=\,
\Truncation{0}
\,
\shape
\,
\SliceMaps{\big}{\mathbf{B}\PUH}
  {  (\TopologicalSpace, \tau) }
  {
    \HomotopyQuotient
      { \FredholmOperators }
      { \GradedPUH }
  }
$
denotes -- for a given twist
$
  \tau
  \,\in\,
  \Maps{}
    { \TopologicalSpace }
    { \mathbf{B} \GradedPUH }
$
-- the $[\tau]$-twisted complex K-cohomology of the topological space $\TopologicalSpace$, being the set homotopy classes of sections of the $\tau$-associated $\FredholmOperators$-bundle, see \eqref{TwistedEquivariantKTheoryAsHomotopyClassesOfSections} below.

\end{itemize}

Towards the equivariant generalization of twisted K-theory, we first consider the notion of ``stable'' projective unitary representations of the equivariance group:

\newpage

\noindent
{\bf Stable projective representations} (following \cite[Ex. 4.1.56]{SS22EPB}).

\vspace{-2mm}
\begin{itemize}[leftmargin=*]
\setlength\itemsep{-1pt}

\item
$\EquivarianceGroup$ denotes a finite group, to serve as the equivariance group and eventually to be specialized to a cyclic group $\EquivarianceGroup \,:=\, \CyclicGroup{\ShiftedLevel}$.

\item
$\Irreps{\EquivarianceGroup}$ denotes the set of isomorphism classes of its irreducible unitary representations.

\item
$
\ComplexRepresentationRing{\EquivarianceGroup}\;\simeq\;
  \Integers
  \big[
    \Irreps{\EquivarianceGroup}
  \big]
$ denotes the complex representation ring, identified with the ring of integer multiples of irreps under tensor product of representations:
\vspace{-2mm}
\begin{equation}
  \label{MultiplicationInTheRepresentationRing}
  [\rho_i],
  [\rho_j]
  \,\in\,
  \Irreps{\EquivarianceGroup}
  \,\subset\,
  \ComplexRepresentationRing{\EquivarianceGroup}
  \qquad
  \vdash
  \qquad
  [\rho_i] \cdot [\rho_j]
  \;:=\;
  \big[
    \rho_i
    \otimes
    \rho_j
  \big]
  \;=\;
  \bigg[\;\;
    \underset{
      \scalebox{.65}{$
      \begin{array}{c}
        { [\rho_k] \,\in\,}
        \\
        { \Irreps{\EquivarianceGroup} }
      \end{array}
      $}
    }{\bigoplus}
    \;
    \underset{
      \scalebox{.6}{$
        \big\{
          1, \cdots, \rho^k_{i j}
        \big\}
      $}
    }{\bigoplus}
    \rho_k
  \bigg]
  \;=\;
  \underset{
    \scalebox{.65}{$
      \begin{array}{c}
        { [\rho_k] \,\in\,}
        \\
        { \Irreps{\EquivarianceGroup} }
      \end{array}
    $}
  }{\sum}
    \rho^{k}_{i j}
  \cdot
  [\rho_k]\;.
\end{equation}

\vspace{-2mm}
\noindent For example, when $\EquivarianceGroup \,=\, \CyclicGroup{\ShiftedLevel}$ as in \eqref{IrrepsOfCyclicGroup}, then this product operation is addition of labels modulo $\ShiftedLevel$:
\vspace{-1mm}
\begin{equation}
  \label{MultiplicationOfIrrepsOfCyclicGroup}
  [\IrrepOfCyclicGroup{\DualDenominator}{\ShiftedLevel}], \;
  [\IrrepOfCyclicGroup{\DualDenominator'}{\ShiftedLevel}]
  \,\in\,
  \Irreps{\CyclicGroup{\ShiftedLevel}}
  \;\;\;\;\;\;\;\;\;\;\;
  \vdash
  \;\;\;\;\;\;\;\;\;\;\;
  [\IrrepOfCyclicGroup{\DualDenominator}{\ShiftedLevel}]
  \cdot
  [\IrrepOfCyclicGroup{\DualDenominator'}{\ShiftedLevel}]
  \;=\;
  [
    \IrrepOfCyclicGroup
    {(
      \DualDenominator
        +
      \DualDenominator'
      \,\mathrm{mod}\, \ShiftedLevel
    )}
    {\ShiftedLevel}
  ]
  \,.
\end{equation}

\vspace{-1mm}
\item
$
  \CharacterGroup{\EquivarianceGroup}
  \,:=\,
  \Homs{}
    { \EquivarianceGroup }
    { \CircleGroup }
  \;\subset\;
  \ComplexRepresentationRing{\EquivarianceGroup}
$
denotes the character group, regarded as the
subgroup of the group of units of the representation ring given by the $\EquivarianceGroup$-representations on $\ComplexNumbers$, and

\item
$
  1_{\rho}
    \,\in \!
  \begin{tikzcd}[column sep=small]
    \ComplexNumbers
    \ar[r, "{\sim}"]
    &
    {\rho} \!\!
  \end{tikzcd}
$ denotes the unit complex number as an element of a 1d irrep
$[\rho] \,\in\, \CharacterGroup{\EquivarianceGroup}$.

For example,  when $\EquivarianceGroup \,=\, \CyclicGroup{\ShiftedLevel}$ as in \eqref{IrrepsOfCyclicGroup}, then
$$
  \CharacterGroup{\CyclicGroup{\ShiftedLevel}}
  \;=\;
  \big(
    \big\{
      [\rho_1]
      ,\,
      \cdots
      ,\,
      [\rho_{\ShiftedLevel}]
    \big\}
    ,\,
   \cdot \;
  \big),
$$
with the group operation
\eqref{MultiplicationOfIrrepsOfCyclicGroup}.

\item
$
\CharacterGroup{\EquivarianceGroup}
\acts  \,
\ComplexRepresentationRing
  {\EquivarianceGroup}
$
denotes the action of the character group on (the abelian group underlying) the representation ring
by tensoring, hence by the formula \eqref{MultiplicationInTheRepresentationRing} for 1-dimensional $\rho_i$.

\item
$
\CharacterGroup{\EquivarianceGroup}
\,\acts\,
\big(
\ComplexNumbers
  \otimes_{{}_{\Integers}}
  \ComplexRepresentationRing{\EquivarianceGroup}
\big)
$
denotes the complexified complex representation ring equipped with the permutation representation of the above action of the character group.

For $\EquivarianceGroup \,=\, \CyclicGroup{\ShiftedLevel}$ this is the {\it regular representation} of
the character group
$\CharacterGroup{\CyclicGroup{\ShiftedLevel}}$,
isomorphic to the direct sum of
its 1d irreps:
\begin{equation}
  \label{RepresentationRingOfCyclicGroupAsRegRepOfCharacterGroup}
  \CharacterGroup{\CyclicGroup{\ShiftedLevel}}
  \,\acts\,
  \big(
    \ComplexNumbers
    \otimes_{{}_{\Integers}}
    \ComplexRepresentationRing{\CyclicGroup{\ShiftedLevel}}
  \big)
  \;\;
  \simeq
  \;\;
  \CharacterGroup{\CyclicGroup{\ShiftedLevel}}
  \,\acts\,
  \underset{
    \Irreps{\CyclicGroup{\ShiftedLevel}}
  }{\bigoplus}
  \,
  \ComplexNumbers
  \;\;
    \simeq
  \;\;
  \CharacterGroup{\CyclicGroup{\ShiftedLevel}}
  \,\acts
  \underset{
    1
      \leq \,
      \Denominator
     \, \leq
    \ShiftedLevel
  }{\bigoplus}
  \bigg(\;
    \underset{
      1
      \leq \,
      \DualDenominator
      \,
      \leq
      \ShiftedLevel
    }{\sum}
    e^{
      2 \pi \ImaginaryUnit
      \,
      \DualDenominator
      \Denominator
      /\ShiftedLevel
    }
    \cdot
    [\rho_{s}]
  \bigg)
  \;\;
    =:
  \;\;
  \underset{
    1
      \leq  \,
      \Denominator
     \, \leq
    \ShiftedLevel
  }{\bigoplus}
  \,
  \IrrepOfCyclicGroup
    {-\Denominator}
    {\ShiftedLevel}
  \,.
\end{equation}

\item
$
\Maps{}
  {\mathbf{B}\EquivarianceGroup}
  {\mathbf{B}\GradedPUH}
\,\simeq\,
\HomotopyQuotient
  {
    \Homs{}
      {\EquivarianceGroup}
      {\GradedPUH}
  }
  { \GradedPUH }
$
denotes the mapping stack, which is equivalently the homotopy quotient of the space of group homomorphisms
$\EquivarianceGroup \xrightarrow{\;} \GradedPUH$, hence of {\it projective unitary $\EquivarianceGroup$-representations}, by the $\GradedPUH$-conjugation action.

\item
$\Maps{}{\mathbf{B}\EquivarianceGroup}{\mathbf{B}\GradedPUH}^{\stable} \xhookrightarrow{\;\;} \Maps{}{\mathbf{B}\EquivarianceGroup}{\mathbf{B}\GradedPUH}$ denotes the sub-object of the ``stable'' projective unitary representations, which are indexed by a level of projectivity $[c_2] \,\in\,  H^2(\EquivarianceGroup;\, \CircleGroup)$ and give by the direct sums of a countably infinite number of copies of all $[c_2]$-projective irreps $\rho^{[c_2]}_i$ of $\EquivarianceGroup$.

\item
$
  \mathbf{B}
  \CharacterGroup{\EquivarianceGroup}
  \;\simeq\;
  \HomotopyQuotient
    { \big\{ \stable_0 \big\}  }
    { \CharacterGroup{\EquivarianceGroup} }
  \xrightarrow{\;\iota\;}
  \Maps{}
    { \mathbf{B}\EquivarianceGroup }
    { \mathbf{B} \PUH }
  ^\stable
$
denotes the further inclusion of the
stable projective representation at
$[c_2] = 0$
\begin{equation}
  \label{TheStableGRepresentation}
  \EquivarianceGroup
  \underset{
    \scalebox{.6}{$\stable_0$}
  }{\acts}
  \HilbertSpace
  \;\simeq\;
  \underset{
    \scalebox{.65}{$
      \begin{array}{c}
        { [\rho_i] \,\in\,  }
        \\
        { \Irreps{\EquivarianceGroup} }
      \end{array}
    $}
  }{\bigoplus}
  \rho_i
  \otimes
  \HilbertSpace
  \,,
\end{equation}

\vspace{-4mm}
\noindent together with its automorphisms
given by the following action of group characters
$[\rho] \,\in\, \CharacterGroup{\EquivarianceGroup}\xhookrightarrow{\; \Omega\iota \;}
\Omega_{\stable_0}
\Maps{}
  { \mathbf{B}\EquivarianceGroup }
  { \mathbf{B} \PUH }
$
by projective auto-intertwiners:
\vspace{-2mm}
\begin{equation}
  \label{CharacterGroupActionByAutoIntertwinersOnStableProjectiveRepresentation}
  \begin{tikzcd}[row sep=small, column sep=large]
    \mathbf{B}\EquivarianceGroup
    \ar[
      rrr,
      bend left=10,
      "{ \stable_0 }",
      "{}"{name=s, swap}
    ]
    \ar[
      rrr,
      bend right=10,
      "{
        \stable_0
      }"{swap},
      "{}"{name=t}
    ]
    &&&
    \mathbf{B}\PUH
    \ar[
      from=s,
      to=t,
      Rightarrow,
      "{
        \rho
        \,\in\,
        \CharacterGroup
          { \EquivarianceGroup }
      }"{xshift=1pt}
    ]
    \\
    {\phantom{A}}
    \\
\small
\bullet
    \ar[
      dd,
      "{g}"
    ]
    &&
    \underset{
      \scalebox{.65}{$
        \begin{array}{c}
          { [\rho_i] \,\in\,  }
          \\
          { \Irreps{\EquivarianceGroup} }
        \end{array}
      $}
    }{\bigoplus}
    \rho_i
    \otimes
    \HilbertSpace
    \ar[
      rr,
      "{ \scalebox{0.8}{$
        v
        \,\mapsto\,
        1_{\rho}
          \otimes
        v $}
       }",
       "{}"{pos=.9, swap, name=s}
    ]
    \ar[
      dd,
      start anchor={[yshift=14pt]},
      "{
     \scalebox{0.8}{$   \underset{[\rho_i]}{\oplus}
        \left(
          \rho_i(g)
          \otimes
          \mathrm{id}
        \right)
        $}
      }"{description}
    ]
    &&
    \underset{
      { [\rho_i] \,\in\,  }
      \atop
      { \Irreps{\EquivarianceGroup} }
    }{\bigoplus}
    \rho_i
    \otimes
    \HilbertSpace
    \ar[
      dd,
      start anchor={[yshift=14pt]},
      "{
    \scalebox{0.8}{$    \underset{[\rho_i]}{\oplus}
        \left(
          \rho_i(g)
          \otimes
          \mathrm{id}
        \right)
        $}
      }"{description}
    ]
    \\
    &
    \longmapsto
    &
    {}
    \\
    \bullet
    &&
    \underset{
      { [\rho_i] \,\in\,  }
      \atop
      { \Irreps{\EquivarianceGroup} }
    }{\bigoplus}
    \rho_i
    \otimes
    \HilbertSpace
    \ar[
      rr,
      "{
    \scalebox{0.8}{$    v
          \,\mapsto\,
        1_{\rho}
          \otimes
        v
        $}
      }"{swap},
       "{}"{pos=.1, name=t}
    ]
    &&
    \underset{
      { [\rho_i] \,\in\,  }
      \atop
      { \Irreps{\EquivarianceGroup} }
    }{\bigoplus}
    \rho_i
    \otimes
    \HilbertSpace
    \ar[
      from=t,
      to=s,
      Rightarrow,
      "{
     \scalebox{0.8}{$   \rho(g)(1_{\rho}) $}
      }"{description}
    ]
  \end{tikzcd}
\end{equation}
Here the diagram on the right shows that the operation of tensoring with the unit element $1_{\rho} \,\in\, \ComplexNumbers\,\simeq\, \rho$ intertwines the stable $\EquivarianceGroup$-representation with itself, up to (as befits a projective intertwiner) a coherent phase factor, here given by the values $\rho(g)(1_{\rho}) \,\in\, \CircleGroup \xhookrightarrow{\;} \ComplexNumbers$:
$$
  v
    \,\in\,
  \rho_j
    \,\subset\,
  \underset{i}{\oplus} \rho_i
  \;\;\;\;\;\;\;\;\;
  \vdash
  \;\;\;\;\;\;\;\;\;
  \big(\underset{i}{\oplus} \rho_i\big)(g)
  \big(
    1_{\rho}
      \otimes
    v
  \big)
  \;=\;
  \big(
    \rho(g)(1_{\rho})
  \big)
  \otimes
  \big(
    \rho_j(g)(v)
  \big)
  \;=\;
  \rho(g)(1_{\rho})
  \cdot
  \Big(
    1_\rho
    \otimes
    \big(\underset{i}{\oplus} \rho_i\big)(g)
    (
      v
    )
  \Big)
  \,.
$$
\end{itemize}

\medskip

\noindent
{\bf Twisted Equivariant $\mathrm{KU}$-Theory} (following \cite[Ex. 4.3.19]{SS22EPB}).

\vspace{-2mm}
\begin{itemize}[leftmargin=*]
\setlength\itemsep{-1pt}

\item
$\EquivarianceGroup \acts \,  \TopologicalSpace$ denotes a
topological space
$\TopologicalSpace$
equipped with a
continuous $G$-action  (a ``$\EquivarianceGroup$-space'').

\item
$
  \Maps{}{
  \HomotopyQuotient
    { \TopologicalSpace }
    { G }
  }
  { \mathbf{B}\GradedPUH }
  ^{\stable}
  \xhookrightarrow{\;}
  \Maps{}{
  \HomotopyQuotient
    { \TopologicalSpace }
    { G }
  }
  { \mathbf{B} \GradedPUH }
$
denotes the sub-object of the mapping stack on the maps that restrict to stable projective $\EquivarianceGroup_x$-representations (in the above sense) of the isotropy group $\EquivarianceGroup_x \,\subset\, \EquivarianceGroup$ of all $x \,\in\, \TopologicalSpace$.

\item
$
  \mathrm{KU}^{[\tau]}
  \big(
    \HomotopyQuotient
      { \TopologicalSpace }
      { \EquivarianceGroup }
  \big)
  \;:=\;
  \Truncation{0}
  \,\shape\,
  \SliceMaps{}{\mathbf{B}\GradedPUH}
    {
      \HomotopyQuotient
        { \TopologicalSpace }
        { \EquivarianceGroup }
    }
    {
      \HomotopyQuotient
        { \FredholmOperators }
        { \GradedPUH }
    }
$
denotes -- for a given stable twist
$
  \tau
  \,\in\,
  \Maps{}
    {
      \HomotopyQuotient
        {\TopologicalSpace}
        {\EquivarianceGroup}
    }
    { \mathbf{B} \GradedPUH }
  ^{\stable}
$
--
the
$[\tau]$-twisted
$\EquivarianceGroup$-equivariant
complex K-cohomology of $\TopologicalSpace$, being the set of homotopy classes of sections of the $\tau$-associated $\FredholmOperators$-fiber bundle over the orbifold $\HomotopyQuotient{\TopologicalSpace}{\EquivarianceGroup}$:
\vspace{-2mm}
\begin{equation}
\label{TwistedEquivariantKTheoryAsHomotopyClassesOfSections}
  \overset{
    \mathclap{
    \raisebox{2pt}{
      \tiny
      \color{orangeii}
      \bf
      \def\arraystretch{.9}
      \begin{tabular}{c}
        twisted equivariant
        \\
        complex K-cohomology
      \end{tabular}
    }
    }
  }{
  \mathrm{KU}^{[\tau]}
  \big(
    \HomotopyQuotient
      { \TopologicalSpace }
      { \EquivarianceGroup }
  \big)
  }
  \;\;
  :=
  \;\;
  \left\{
  \!\!\!
  \adjustbox{raise=2pt}{
  \begin{tikzcd}[column sep=large]
    &&
    \overset{
      \mathclap{
      \raisebox{2pt}{
        \tiny
        \color{darkblue}
        \bf
        \def\arraystretch{.9}
        \begin{tabular}{c}
          universal
          \\
          $\FredholmOperators$-bundle
        \end{tabular}
      }
      }
    }{
    \HomotopyQuotient
      { \FredholmOperators }
      { \GradedPUH }
    }
    \ar[d]
    \\
    \underset{
      \mathclap{
      \raisebox{-0pt}{
        \tiny
        \color{darkblue}
        \bf
        orbifold
      }
      }
    }{
    \HomotopyQuotient
      { \TopologicalSpace }
      { \EquivarianceGroup }
    }
    \ar[
      rr,
      "{
        \mbox{
          \tiny
          \color{greenii}
          \bf
          equivariant
          twist
        }
      }"{swap},
      "{ \tau }"{}
    ]
    \ar[
      urr,
      dashed,
      "{
        \mbox{
          \tiny
          \color{orangeii}
          \bf
          \def\arraystretch{.9}
          \begin{tabular}{c}
          twisted equivariant
          \\
          cocycle
          \end{tabular}
        }
      }"{sloped}
    ]
    &&
    \underset{
      \mathclap{
      \raisebox{-2pt}{
        \tiny
        \color{darkblue}
        \bf
        \def\arraystretch{.9}
        \begin{tabular}{c}
          moduli
          \\
          stack
        \end{tabular}
      }
      }
    }{
      \mathbf{B}\GradedPUH
    }
  \end{tikzcd}
  }
  \!\!\!
  \right\}_{
   \!\!\big/_{\homotopy}
  }
\end{equation}

\end{itemize}

We now specialize this general situation to the  of case of interest here, namely (1) to $\EquivarianceGroup$-fixed loci and (2) to $\EquivarianceGroup = \CyclicGroup{\ShiftedLevel}$:

\medskip
\noindent
{\bf Twisted equivariant K-theory of singularities.}

\vspace{-2mm}
\begin{itemize}[leftmargin=*]
\setlength\itemsep{-1pt}

\item
$\FredholmOperators^{\EquivarianceGroup} \xhookrightarrow{\;} \FredholmOperators$ denotes the subspace of Fredholm operators which are fixed under the action of the unique stable $\EquivarianceGroup$-representation \eqref{TheStableGRepresentation}, hence the following stacky homotopy fiber space:
\begin{equation}
  \label{EquivariantFredholmOperators}
  \begin{tikzcd}[column sep=huge]
    \FredholmOperators^{\EquivarianceGroup}
    \ar[r]
    \ar[d]
    \ar[
      dr,
      phantom,
      "{
        \mbox{
          \tiny
          \rm
          (pb)
        }
      }"
    ]
    &
    \Maps{\big}
      { \mathbf{B} \EquivarianceGroup }
      {
        \HomotopyQuotient
          { \FredholmOperators }
          { \PUH }
      }
    \ar[d]
    \\
    \ast
    \ar[
      r,
      "{
        \vdash \;
        \underset{i}{\oplus} \rho_i
        \otimes \HilbertSpace
      }"{swap}
    ]
    &
    \Maps{}
      { \mathbf{B}\EquivarianceGroup }
      { \mathbf{B} \PUH }\,.
  \end{tikzcd}
\end{equation}
Since the $\PUH$-action on $\FredholmOperators$ is by conjugation, these $\EquivarianceGroup$-fixed Fredholm operators are precisely the $\EquivarianceGroup$-equivariant Fredholm operators.
Therefore,
Schur's Lemma implies that $\EquivarianceGroup$-fixed Fredholm operators are direct sums of plain Fredholm operators $\FredholmOperator_i \,:\, \HilbertSpace \xrightarrow{\;} \HilbertSpace$ indexed by the $\EquivarianceGroup$-irreps:
\begin{equation}
  \label{EquivariantFredholmOperatorAsDirectSumOverIrreps}
  \FredholmOperator
  \,\in\,
  \FredholmOperators^{\EquivarianceGroup}
  \;\;\;\;\;\;\;\;\;\;\;\;\;
  \vdash
  \;\;\;\;\;\;\;\;\;\;\;\;\;
  \FredholmOperator
  \;\;:
  \begin{tikzcd}
    \underset{
      { [\rho_i] \in }
      \atop
      { \Irreps{\EquivarianceGroup} }
    }{\bigoplus}
    \rho_i \otimes \HilbertSpace
    \ar[
      rr,
      "{
        \bigoplus_i
        \,
        \mathrm{id}
        \otimes
        \FredholmOperator_i
      }"
    ]
    &&
    \underset{
      { [\rho_i] \in }
      \atop
      { \Irreps{\EquivarianceGroup} }
    }{\bigoplus}
    \rho_i \otimes \HilbertSpace
    \,.
  \end{tikzcd}
\end{equation}
Hence, in particular, the shape of the space
of $\EquivarianceGroup$-fixed Fredholm operators is the
$\Irreps{\EquivarianceGroup}$-indexed product
of copies of the shape of the space all Fredholm operators:
\begin{equation}
  \label{ShapeOfSpaceOfEquivariantFredholmOperators}
  \shape
  \big(
    \FredholmOperators^{\EquivarianceGroup}
  \big)
  \;\simeq\;
  \underset{
    \Irreps{\EquivarianceGroup}
  }{\prod}
  \big(
    \shape
    \,
    \FredholmOperators
  \big)
  \,.
\end{equation}
Consequently,
the
(co)kernels
\eqref{VirtualKernelOfFredholmOperator}
of $\EquivarianceGroup$-fixed Fredholm operators
are
 $\EquivarianceGroup$-subspaces
 of the stable
 $\EquivarianceGroup$-representation,
 hence are any finite-dimensional
 virtual $\EquivarianceGroup$-representations:
 \vspace{-2mm}
\begin{equation}
  \label{COKernelOfEquivariantFredholmOperator}
  \hspace{-2mm}
  {
    \FredholmOperator
    \,\in\,
    \FredholmOperators^{\EquivarianceGroup}
   \qquad
   \vdash
   \qquad
  }
  \FredholmOperator
  \;\;
  :
  \;\;
  \bigoplus_i
  \,
  \rho_i
  \otimes
  \HilbertSpace
  \;\simeq
  \!\!\!
  \begin{tikzcd}[row sep=-3pt, column sep=large]
    \bigoplus_i \rho_i
    \otimes
    \HilbertSpace
    &&
    \bigoplus_i \rho_i
    \otimes
    \HilbertSpace
    \\
   \qquad  \quad
   \oplus
    \qquad
    \ar[
      rr,
      "{
        \left(\!\!\!
        \hspace{-2pt}
        \begin{array}{cc}
          \mathrm{id}
          \\
          0
        \end{array}
        \hspace{-6pt}
        \begin{array}{cc}
          0
          \\
          0
        \end{array}
        \hspace{-2pt}
       \!\!\! \right)
      }"{description}
    ]
    &&
    \qquad
    \oplus \qquad \quad
    \\
    \quad \mathrm{ker}(\FredholmOperator)
    &&
    \mathrm{coker}(\FredholmOperator)
  \end{tikzcd}
  \!\!\!\!\!
  \simeq
  \;
  \bigoplus_i
  \rho_i
  \otimes
  \HilbertSpace
  \,.
\end{equation}

\vspace{-2mm}
\noindent Moreover, on the fixed locus
$\FredholmOperators^{\EquivarianceGroup}$
there is still an action
$
  \FredholmOperator
  \,\mapsto\,
  \color{purple}
  [\rho] \cdot \FredholmOperator
$
by group characters $[\rho] \,\in\, \CharacterGroup{\EquivarianceGroup}$,
induced by their projective intertwining
action
\eqref{CharacterGroupActionByAutoIntertwinersOnStableProjectiveRepresentation} on the stable $\EquivarianceGroup$-representation:
\vspace{-2mm}
\begin{equation}
  \label{ActionOfGroupCharacterOnEquivariantFRedholmOperator}
  \hspace{-4mm}
  \begin{tikzcd}[column sep=70pt]
    \overset{
      \mathclap{
      \raisebox{2pt}{
        \tiny
        \color{darkblue}
        \bf
        \def\arraystretch{.9}
        \begin{tabular}{c}
          stable
          $\EquivarianceGroup$-representation
        \end{tabular}
      }
      }
    }{
      \bigoplus_i
      \rho_i
      \otimes
      \HilbertSpace
    }
    \ar[
      rrr,
      phantom,
      shift right=18pt,
      "{
        \mbox{
          \tiny
          \color{orangeii}
          \bf
          \begin{tabular}{c}
            action of group character on
            equivariant Fredholm operator
          \end{tabular}
        }
      }"{pos=.45}
    ]
    \ar[
      dr,
      "{
        \FredholmOperator
      }",
      "{
        \mbox{
          \tiny
          \color{greenii}
          Fredholm operator
        }
      }"{swap, sloped, pos=.4}
    ]
    \ar[
      rr,
      "{
        v
        \,\mapsto\,
        1_{\rho} \otimes v
      }"
    ]
    \ar[
      dd,
      "{
        \bigoplus_i
        \rho_i
        \otimes
        \HilbertSpace
      }"{description}
    ]
    &[-10pt]
    &
    \bigoplus_i
    \rho_i
    \otimes
    \HilbertSpace
    \ar[
      dr,
      "{
        \color{purple}
        [\rho]
        \cdot
        \FredholmOperator
      }"
    ]
    &
    {}
    \\
    \ar[
      dr,
      phantom,
      "{
        \mbox{
          \tiny
          \color{orangeii}
          \def\arraystretch{.9}
          \begin{tabular}{c}
            equivariance of
            \\
            Fredholm operator
          \end{tabular}
        }
      }"{sloped}
    ]
    &
    \bigoplus_i
    \rho_i
    \otimes
    \HilbertSpace
    \ar[
      rr,
      "{
        v
        \,\mapsto\,
        1_{\rho} \otimes v
      }",
      "{
        \mbox{
          \tiny
          \color{greenii}
          \bf
          tensoring with unit of
          group character
        }
      }"{swap}
    ]
    \ar[
      dd,
      "{
        \bigoplus_s
        \rho_i(g)
        \otimes
        \mathrm{id}
      }"{description}
    ]
    &&
    \bigoplus_i
    \rho_i
    \otimes
    \HilbertSpace
    \ar[
      dd,
      "{
        \bigoplus_s
        \rho_i(g)
        \otimes
        \mathrm{id}
      }"{description}
    ]
    \\
    \bigoplus_i
    \rho_i
    \otimes
    \HilbertSpace
    \ar[
      dr,
      "{
        \FredholmOperator
      }"{swap}
    ]
    &
    {}
    \ar[
      rr,
      phantom,
      "{
        \mbox{
          \tiny
          \color{orangeii}
          \bf
          \def\arraystretch{.9}
          \begin{tabular}{c}
            projective intertwining action
            \\
            of group character
          \end{tabular}
        }
      }"
    ]
    &&
    {}
    \\
    &
    \bigoplus_i
    \rho_i
    \otimes
    \HilbertSpace
    \ar[
      rr,
      "{
        v
        \,\mapsto\,
        1_{\rho} \otimes v
      }"{swap}
    ]
    &&
    \bigoplus_i
    \rho_i
    \otimes
    \HilbertSpace
  \end{tikzcd}
\end{equation}

\vspace{0mm}
\noindent From the top square it is manifest that this $\CharacterGroup{\EquivarianceGroup}$-action on the
(co)kernels \eqref{COKernelOfEquivariantFredholmOperator}
is again by tensoring with unit elements in the group characters:
\vspace{-1mm}
$$
  \mathrm{ker}
  \big(
    [\rho]
    \cdot
    \FredholmOperator
  \big)
  \;=\;
  1_{\rho} \otimes
  \mathrm{ker}(\FredholmOperator)
  \,,
  \;\;
  \;\;\;\;\;
  \mathrm{coker}
  \big(
    [\rho]
    \cdot
    \FredholmOperator
  \big)
  \;\simeq\;
  1_{\rho}
    \otimes
  \mathrm{coker}(\FredholmOperator)
  \,.
$$
Therefore, the $\CharacterGroup{\EquivarianceGroup}$-action on the
geometric $\EquivarianceGroup$-fixed locus
\eqref{ShapeOfSpaceOfEquivariantFredholmOperators}
of the
classifying
$\EquivarianceGroup$-space
is by permuting the irrep labels as given by their (inverse) tensor product with the group characters:
\begin{equation}
  \label{CharacterGroupActionOnShapeOfGFixedFredholmOperators}
  \begin{tikzcd}[row sep=2pt]
    \CharacterGroup{\EquivarianceGroup}
    \times
    \shape
    \big(
      \FredholmOperators^{\EquivarianceGroup}
    \big)
    \ar[rr]
    \ar[d, phantom, "{\simeq}"{rotate=-90}]
    &&
    \shape
    \big(
      \FredholmOperators^{\EquivarianceGroup}
    \big)
    \ar[d, phantom, "{\simeq}"{rotate=-90}]
    \\[+4pt]
    \CharacterGroup{\EquivarianceGroup}
    \times
    \underset{
      \Irreps{\EquivarianceGroup}
    }{\prod}
    \shape\,\FredholmOperators
    &&
    \underset{
      \Irreps{\EquivarianceGroup}
    }{\prod}
    \shape\,\FredholmOperators
    \\[-3pt]
 \scalebox{0.8}{$
 \Big(
      {\color{purple}[\rho]},
      \,
      \big(
        f_{[\rho_i]}
      \big)_{[\rho_i] \in \Irreps{\EquivarianceGroup}}
    \Big)
    $}
    &\longmapsto&
\scalebox{0.8}{$
\big(
      f_{
        {\color{purple}[\rho]^{-1}}
        \cdot
        [\rho_i]
      }
    \big)
    _{
      [\rho_i] \in \Irreps{\EquivarianceGroup}
      \,.
    }
    $}
  \end{tikzcd}
\end{equation}

\vspace{-2mm}
\item
$
\HomotopyQuotient
  {\TopologicalSpace}
  {\EquivarianceGroup}
\;\simeq\;
\TopologicalSpace
 \times
\HomotopyQuotient
  {\ast}
  {\EquivarianceGroup}
$
denotes the special case of the domain orbifold
where $\TopologicalSpace$ is fixed by the $\EquivarianceGroup$-action, hence where $\TopologicalSpace$ is entirely ``inside a $\EquivarianceGroup$-orbi-singularity''.

In this case, the mapping stack adjunction
gives
the following natural equivalence of the moduli stack of stable twists:
\begin{equation}
  \label{ModuliOfStableTwistsOnGOrbiSingularity}
  \begin{tikzcd}[column sep=large]
  \Maps{}
    {
      \HomotopyQuotient
        { \TopologicalSpace }
        { \EquivarianceGroup }
    }
    { \mathbf{B} \PUH }
  ^{\stable}
  \ar[
    r,
    "{
      \underset{
        \mathclap{
        \raisebox{-1pt}{
          \tiny
          assmpt.
        }
        }
      }{
        \sim
      }
    }"{swap}
  ]
  &
  \Maps{}
    {
      \TopologicalSpace
      \times
      \HomotopyQuotient
        { \ast }
        { \EquivarianceGroup }
    }
    { \mathbf{B} \PUH }
  ^{\stable}
    \ar[
      r,
      "{ \widetilde{(-)} }",
      "{
        \underset{
          \raisebox{-1pt}{
            \tiny
            adj.
          }
        }{
          \sim
        }
      } "{swap}
    ]
    &
  \Maps{\big}
    { \TopologicalSpace }
    {
      \Maps{}
        {
          \mathbf{B} \EquivarianceGroup
        }
        { \mathbf{B} \PUH }
      ^\stable
    }.
  \end{tikzcd}
  \end{equation}

\vspace{-2mm}
If here $\SmoothManifold$ is a smooth manifold, then the maps on the right are classified by (\cite[Ex. 4.3.19]{SS22EPB}):
$$
  \shape
  \,
  \Maps{}
    { \mathbf{B}\EquivarianceGroup }
    { \mathbf{B} \PUH }
  ^\stable
  \;\;
  \simeq
  \;\;
  \Maps{}
    { B G }
    { B^3 \Integers }
  \,.
$$
Furthermore, when  $\EquivarianceGroup \,\simeq\, \CyclicGroup{\ShiftedLevel}$, where
the cohomology groups are given by
$$
  \begin{aligned}
    H^3
    \big(
      B \CyclicGroup{\ShiftedLevel}
      ;\,
      \Integers
    \big)
    &
    \;\simeq\;
    H^3_{\mathrm{Grp}}
    \big(
      \CyclicGroup{\ShiftedLevel}
      ;\,
      \Integers
    \big)
    \\
    & \;\simeq\;
    0
    \\
    {\phantom{A}}
  \end{aligned}
  \qquad  \quad
  \mbox{and}
  \qquad  \quad
  \begin{aligned}
    H^2
    \big(
      B \CyclicGroup{\ShiftedLevel};
      \, \Integers
    \big)
    &
    \;\simeq\;
    H^1
    \big(
      B \CyclicGroup{\ShiftedLevel}
      ;\,
      B \CircleGroup
    \big)
    \\
    &
    \;\simeq\;
    \Homs{}
      { \CyclicGroup{\ShiftedLevel} }
      { \CircleGroup }
    \\
    &
    \;=\;
    \CharacterGroup{\CyclicGroup{\ShiftedLevel}}
    \,,
  \end{aligned}
$$
this classifying space of twists becomes
$
  \Maps{}
    { B \CyclicGroup{\ShiftedLevel} }
    { B^3 \Integers }
  \;\simeq\;
  B \CharacterGroup{\CyclicGroup{\ShiftedLevel}}
  \,\times\,
  B^3 \Integers
  \,,
$
so that the twists \eqref{ModuliOfStableTwistsOnGOrbiSingularity} on an $\mathbb{A}_{\ShiftedLevel-1}$-singularity
$\SmoothManifold$ are classified by
\vspace{-2mm}
\begin{equation}
  \label{ClassesOfTwistsOnATypeSingularity}
  \Truncation{0}
  \,
  \shape
  \,
  \Maps{\big}
    { \SmoothManifold }
    {
      \Maps{}
        { \mathbf{B} \CyclicGroup{\ShiftedLevel} }
        { \mathbf{B} \PUH }
    }
  \;\;
  \simeq
  \;\;
  \overset{
    \mathclap{
    \raisebox{2pt}{
      \tiny
      \color{darkblue}
      \bf
      \begin{tabular}{c}
        special
        \\
        deg=1 twists
      \end{tabular}
    }
    }
  }{
  H^1\big(
    \SmoothManifold
    ;\,
    \CharacterGroup{\CyclicGroup{\ShiftedLevel}}
  \big)
  }
  \,\times\,
  \overset{
    \raisebox{2pt}{
      \tiny
      \color{darkblue}
      \bf
      \def\arraystretch{.9}
      \begin{tabular}{c}
        ordinary
        \\
        deg=3 twists
      \end{tabular}
    }
  }{
  H^3\big(
    \SmoothManifold
    ;\,
    \Integers
  \big)
  }
  \,.
\end{equation}
The second factor is the familiar degree=3 twist of complex K-theory, while the first factor is the twist
by flat complex line bundles with holonomy in the character group of the singularity, that appear in Prop. \ref{TheBareBProfiniteIntegersTwistOfComplexRationalizedAEquivariantKTheory}.

\item
Denote by
\vspace{-3mm}
$$
  \widetilde{\tau}
  \;\in\!
  \begin{tikzcd}[column sep=large]
  \Maps{}
    { \TopologicalSpace }
    {
      \mathbf{B} \CharacterGroup{\EquivarianceGroup}
    }
  \ar[
    r,
    "{
      \Maps{}
        { \TopologicalSpace }
        { \iota }
    }",
    "{
      \mbox{
        \tiny
        \eqref{TheStableGRepresentation}
     }
    }"{swap}
  ]
  &[+9pt]
  \Maps{\Big}
    { \TopologicalSpace }
    {
      \Maps{}
        { \mathbf{B}\EquivarianceGroup }
        { \mathbf{B} \PUH }
      ^\stable
    }
  \underset{
    \mathclap{
    \raisebox{-0pt}{
      \tiny
      \eqref{ModuliOfStableTwistsOnGOrbiSingularity}
    }
    }
  }{
    \;\simeq\;
  }
  \Maps{\big}
    {
      \TopologicalSpace
      \times
      \HomotopyQuotient{\ast}{\EquivarianceGroup}
    }
    { \mathbf{B} \PUH }
  ^{\stable}
  \end{tikzcd}
$$

\vspace{-3mm}
\noindent the assumption that the given stable twist
$\tau$ on $\TopologicalSpace \times \HomotopyQuotient{\ast}{\EquivarianceGroup}$ is on $\TopologicalSpace$
a cocycle with coefficients in $B \CharacterGroup{\EquivarianceGroup}$
\eqref{CharacterGroupActionByAutoIntertwinersOnStableProjectiveRepresentation}
(i.e., in the first factor of \eqref{ClassesOfTwistsOnATypeSingularity}).
This means that, under the equivalence
\eqref{ModuliOfStableTwistsOnGOrbiSingularity},
the
TE-K-theory cocycles \eqref{TwistedEquivariantKTheoryAsHomotopyClassesOfSections} now take values
in $\CharacterGroup{\EquivarianceGroup}$-associated
\eqref{ActionOfGroupCharacterOnEquivariantFRedholmOperator}
$\FredholmOperators^{\EquivarianceGroup}$-fiber
\eqref{EquivariantFredholmOperators}
bundles over $\TopologicalSpace$:
\vspace{-2mm}
\begin{equation}
  \label{TheAdjunctKCocycle}
  \begin{tikzcd}[column sep=large]
    &
    \HomotopyQuotient
      { \FredholmOperators }
      { \PUH }
    \ar[d]
    \\
    \ComplexManifold
    \times
    \HomotopyQuotient
      { \ast }
      { \EquivarianceGroup }
    \ar[
      r,
      "{
        \tau
      }",
      "{
        \mbox{
          \tiny
          \color{greenii}
          \bf
          \def\arraystretch{.9}
          \begin{tabular}{c}
            equivariant
            \\
            twist
          \end{tabular}
        }
      }"{swap}
    ]
    \ar[
      ur,
      dashed,
      "{
        \mbox{
          \tiny \bf
          \color{greenii}
          K-cocycle
        }
      }"{sloped}
    ]
    &
    \mathbf{B}\PUH
  \end{tikzcd}
  \hspace{.7cm}
    \leftrightarrow
  \hspace{.7cm}
  \begin{tikzcd}[column sep=large]
    &
    \HomotopyQuotient
      {
        \FredholmOperators^{
          \scalebox{.7}{$\EquivarianceGroup$}
        }
      }
      { \CharacterGroup{\EquivarianceGroup} }
    \ar[d]
    \ar[
      r
    ]
    \ar[dr, phantom, "{\mbox{\tiny(pb)}}"]
    &
    \Maps{}
      { \mathbf{B}\EquivarianceGroup }
      {
        \HomotopyQuotient
          { \FredholmOperators }
          { \PUH }
      }
      ^{\stable}
    \ar[d]
    \\
    \ComplexManifold
    \ar[
      r,
      "{
        \mbox{
          \tiny
          \color{orangeii}
          \bf
          \def\arraystretch{.9}
          \begin{tabular}{c}
            1-twist on
            \\
            singularity
          \end{tabular}
        }
      }"{swap}
    ]
    \ar[
      ur,
      dashed,
      "{
        \mbox{
          \tiny
          \color{greenii}
          \bf
          \def\arraystretch{.9}
          \begin{tabular}{c}
            adjunct
            \\
            K-cocycle
          \end{tabular}
        }
      }"{sloped, pos=.40}
    ]
    \ar[
      rr,
      rounded corners,
      to path={
        (\tikztostart.south)
        -- ([yshift=-12.8pt]\tikztostart.south)
        -- node[above]{\scalebox{.78}{$\widetilde{\tau}$}}
           ([yshift=-10pt]\tikztotarget.south)
        -- ([yshift=-00pt]\tikztotarget.south)
      }    ]
    &
    \mathbf{B}
    \CharacterGroup{\EquivarianceGroup}
    \ar[
      r,
      "{\iota}",
      "{
        \mbox{
          \tiny
          \eqref{TheStableGRepresentation}
        }
      }"{swap}
    ]
    &
    \Maps{}
      { \mathbf{B}\EquivarianceGroup }
      { \mathbf{B}\PUH }
    ^{\stable}
  \end{tikzcd}
\end{equation}

\end{itemize}

 This brings out the special 1-twist that appears on singularities. Next we see how
this becomes the twist  a flat complex line bundle, by entering the differential setting,
i.e., into twisted equivariant differential K-theory.

\medskip

\noindent
{\bf \TED-K-Theory of $\mathbb{A}$-type singularities} (following \cite[Def. 4.38]{FSS20CharacterMap}).

\vspace{-1mm}
\begin{itemize}[leftmargin=*]
\setlength\itemsep{-1pt}

\vspace{-.1cm}
\item
$\SmoothInfinityGroupoids := \Localization{\LocalWeakEquivalences} \SimplicialPresheaves(\CartesianSpaces)$ denotes the homotopy theory of $\infty$-stacks over the site of smooth manifolds (e.g.  \cite[Ex. A.49]{FSS20CharacterMap}).

For the most part here we just need that, in addition to plain homotopy types, this contains the sheaves of De Rham complexes:

\vspace{-.0cm}
\item
$ B^{d} \DeRhamComplex{}{-; \ComplexNumbers} \,\in\, \SmoothInfinityGroupoids $ denotes the image under the Dold-Kan embedding $\DoldKanConstruction$ (cf \cite[Ex. A.66]{FSS20CharacterMap}) of the sheaf of de Rham complexes of $\ComplexNumbers$-valued differential forms, regarded as chain complexes in non-positive degrees, shifted up in degrees by $d$ and then cohomologically truncated in degree 0:
$$
  B^{d}
  \DeRhamComplex{}{-; \ComplexNumbers}
  \;:=\;
  \DoldKanConstruction
  \Big(
    \cdots
    \xrightarrow{0}
    0
    \xrightarrow{0}
    \DifferentialForms{}{0}
      {-; \ComplexNumbers}
    \xrightarrow{\DeRhamDifferential}
    \DifferentialForms{}{1}
      {-; \ComplexNumbers}
    \xrightarrow{\DeRhamDifferential}
    \cdots
    \xrightarrow{\DeRhamDifferential}
    \DifferentialForms{}{d}
      {-; \ComplexNumbers}
    \vert_{ \DeRhamDifferential = 0}
  \Big)
  \;\;\;
  \in
  \;
  \SmoothInfinityGroupoids
  \,.
$$
By the Poincar{\'e} Lemma and as an incarnation of the De Rham theorem, this  $\infty$-stack is equivalent
to the classifying space for ordinary cohomology with complex coefficients in degree $d$:
\begin{equation}
  \label{PoincareLemma}
  B^d \ComplexNumbers
  \;\simeq\;
  B^d
 \DeRhamComplex{}{-; \ComplexNumbers}
  \;\;\;
  \in
  \;
  \SmoothInfinityGroupoids
  \,.
\end{equation}

\vspace{-2mm}
\item
$
B^d
\DeRhamComplex{}
  {-; \IrrepOfCyclicGroup{\Denominator}{\ShiftedLevel}}
$
denotes the de Rham complex as above,
now regarded with an action \eqref{IrrepsOfCyclicGroup} of $\CharacterGroup{\CyclicGroup{\ShiftedLevel}}$ on its coefficients.

\item
$
  \HomotopyQuotient
    {
      B^d
      \DeRhamComplex{}
      {-;\IrrepOfCyclicGroup{\Denominator}{\ShiftedLevel}}
    }
    {
      \CharacterGroup
        {\CyclicGroup{\ShiftedLevel}}
    }
  \;\;\;
  \in
  \;
  \Slice
    { (\SmoothInfinityGroupoids) }
    {
      \mathbf{B}
      \CharacterGroup{\CyclicGroup{\ShiftedLevel}}
    }
$
denotes the corresponding homotopy quotients.

\vspace{-.0cm}
\item
$X \xrightarrow{\eta^{\ComplexNumbers}_{X}} \CRationalization X$ denotes the {\it $\ComplexNumbers$-rationalization} of a
simply-connected homotopy type $X$, hence its $\RationalNumbers$-rationalization followed by derived extension of scalars along $\RationalNumbers \xhookrightarrow{\;} \ComplexNumbers$ (cf. \cite[Rem. 3.64]{FSS20CharacterMap}).

In particular, the $\ComplexNumbers$-rational homotopy type of the space of Fredholm operators
\eqref{AFredholmOperator} is the map representing the usual Chern character (cf. \cite[Ex. 4.13]{FSS20CharacterMap}):
\begin{equation}
  \label{CRationalizationOfSpaceOfFredholmOperators}
  \begin{tikzcd}
  \shape \,
  \FredholmOperators
  \ar[r]
  &
  \CRationalization
  \big(
    \shape
    \,
    \FredholmOperators
  \big)
  \;\simeq\;
  \underset{
    d \in \NaturalNumbers
  }{\prod}
  \,
  B^{2d} \ComplexNumbers
  \end{tikzcd}
\end{equation}
and that of its $\EquivarianceGroup$-fixed locus
with its $\CharacterGroup{\CyclicGroup{\ShiftedLevel}}$-action \eqref{CharacterGroupActionOnShapeOfGFixedFredholmOperators} is as follows:
\begin{equation}
  \label{CRationalizationOfGFixedFRedholmOperators}
  \def\arraystretch{1.7}
  \begin{array}{lll}
  \CharacterGroup{\CyclicGroup{\ShiftedLevel}}
  \acts \;
  \Big(
    \Localization{\ComplexNumbers}
    \,
    \shape
    \big(
      \FredholmOperators^{\EquivarianceGroup}
    \big)
  \Big)
  &
  \;\simeq\;
  \CharacterGroup{\CyclicGroup{\ShiftedLevel}}
  \acts \;
  \Big(
    \Localization{\ComplexNumbers}
    \,
    \underset{
      \Irreps{\EquivarianceGroup}
    }{\prod}
    \,
    \shape
    \,
    \FredholmOperators
  \Big)
  &
  \proofstep{
    by \eqref{ShapeOfSpaceOfEquivariantFredholmOperators}
  }
  \\
  &
  \;\simeq\;
  \CharacterGroup{\CyclicGroup{\ShiftedLevel}}
  \acts \;
  \Big(\;
    \underset{
      \Irreps{\EquivarianceGroup}
    }{\prod}
    \Localization{\ComplexNumbers}
    \,
    \shape
    \,
    \FredholmOperators
  \Big)
  \\
  &
  \;\simeq\;
  \CharacterGroup{\CyclicGroup{\ShiftedLevel}}
  \acts \;
  \Big(\;
    \underset{
      \Irreps{\EquivarianceGroup}
    }{\prod}
    \;
    \underset{
      d \,\in\, \NaturalNumbers
    }{\prod}
    B^{2d} \ComplexNumbers
  \Big)
  &
  \proofstep{
    by \eqref{CRationalizationOfSpaceOfFredholmOperators}
  }
  \\
  &
  \;\simeq\;
  \CharacterGroup{\CyclicGroup{\ShiftedLevel}}
  \acts \;
  \Big(\;
  \underset{
    \Irreps{\EquivarianceGroup}
  }{\bigoplus}
  \;
  \underset{
    d \in \NaturalNumbers
  }{\bigoplus}
   B^{2d}
   \DeRhamComplex{}
     { -;\,\ComplexNumbers }
  \Big)
  &
  \proofstep{
    by
    \eqref{PoincareLemma}
  }
  \\
  &
  \;\simeq\;
  \underset{
    1 \leq \, \Denominator \, \leq \ShiftedLevel
  }{\bigoplus}
  \;\,
  \underset{
    d \in \NaturalNumbers
  }{\bigoplus}
\,   B^{2d}
   \DeRhamComplex{}
     {
     -
     ;\,
     \IrrepOfCyclicGroup{\Denominator}{\ShiftedLevel}
     }
   &
   \proofstep{
     by
     \eqref{CharacterGroupActionOnShapeOfGFixedFredholmOperators}
     \&
     \eqref{RepresentationRingOfCyclicGroupAsRegRepOfCharacterGroup}
   }.
  \end{array}
\end{equation}
(In the last step, notice that the sign on the far right of \eqref{RepresentationRingOfCyclicGroupAsRegRepOfCharacterGroup} cancels against the inversion in the action \eqref{CharacterGroupActionOnShapeOfGFixedFredholmOperators}.)

\item
$
\DifferentialKU
\,\in\,
\SmoothInfinityGroupoids$
denotes the coefficient stack for {\it differential complex K-theory}, namely the following
homotopy fiber product (as in the general setting of \cite{HopkinsSinger05}\cite{BNV13}, we follow \cite[Def. 4.38]{FSS20CharacterMap}):
\begin{equation}
  \label{FiberProductDefiningHolomorphicDifferentialKTheory}
  \begin{tikzcd}[column sep=65pt]
    \overset{
    }{
      \DifferentialKU
    }
    \ar[d]
    \ar[
      rr,
      "{
      }"
    ]
    \ar[
      drr,
      phantom,
      "{\mbox{\tiny (pb)}}"{pos=.58}
    ]
    &
    &
    \underset{
      d \in \NaturalNumbers
    }{\bigoplus}
    \DifferentialForms{}{2d}
      {-;\, \ComplexNumbers}
      \vert_{ \DeRhamDifferential = 0}
    \ar[
      d,
      start anchor={[yshift=5pt]},
      "{
        \mbox{
          \tiny
        }
      }"{xshift=-2pt}
    ]
    \\
    \shape \, \FredholmOperators
    \ar[
      r,
      "{
        \mbox{
          \tiny
          \eqref{CRationalizationOfSpaceOfFredholmOperators}
        }
      }"{swap}
    ]
    \ar[
      rr,
      rounded corners,
      to path={
           ([yshift=-00pt]\tikztostart.south)
        -- ([yshift=-11pt]\tikztostart.south)
        -- node[above]{\scalebox{.7}{$
             \mathbf{ch}
           $}}
           ([yshift=-2pt]\tikztotarget.south)
        -- ([yshift=+07pt]\tikztotarget.south)
      }
    ]
    &
    \CRationalization \,
    \shape \, \FredholmOperators
    \ar[
      r,
      "{\sim}",
      "{
        \mbox{
        \tiny
          \eqref{CRationalizationOfSpaceOfFredholmOperators}
          \&
          \eqref{PoincareLemma}
                  }
      }"{swap}
    ]
    &
    \underset{
      d \in \NaturalNumbers
    }{\bigoplus}
    \,
    B^{2d}
    \DeRhamComplex{}
      {-;\ComplexNumbers}
  \end{tikzcd}
\end{equation}

The full twisted equivariant version of this holomorphic-differential K-theory is defined in the analogous way for each fixed locus $\FredholmOperators^{H}$, as  $\EquivarianceGroup/H$ for $H \subset \EquivarianceGroup$ ranges through the category of orbits of the group $\EquivarianceGroup$. Here we only need and
hence only display the
adjunct situation \eqref{TheAdjunctKCocycle} over the $\EquivarianceGroup$-singularity:

\item
$
\HomotopyQuotient
{
\DifferentialKU
  ^{\CyclicGroup{\ShiftedLevel}}
}
{ \CharacterGroup{\CyclicGroup{\ShiftedLevel}} }
\;\in\;
\ComplexInfinityGroupoids
$
denotes the homotopy quotient of the
$\CharacterGroup{\CyclicGroup{\ShiftedLevel}}$-action
induced by the respective actions
\eqref{CRationalizationOfGFixedFRedholmOperators}
on the $\CyclicGroup{\ShiftedLevel}$-fixed loci
in the factors of the fiber product,
hence the following homotopy pullback:
\begin{equation}
  \label{FiberProductDefiningTEdKThepryOnSingularity}
  \begin{tikzcd}[column sep=huge]
    \overset{
    }{
      \HomotopyQuotient
        { \DifferentialKU^{\CyclicGroup{\ShiftedLevel}} }
        { \CharacterGroup{\CyclicGroup{\ShiftedLevel}} }
    }
    \ar[d]
    \ar[
      rr,
      "{
      }"
    ]
    \ar[
      drr,
      phantom,
      "{\mbox{\tiny (pb)}}"{pos=.58}
    ]
    &
    &
    \HomotopyQuotient
    {
    \underset{
      1 \leq \, \Denominator \, \leq \NumberOfPunctures
    }{\bigoplus}
    \;
    \underset{
      d \in \NaturalNumbers
    }{\bigoplus}
    \DifferentialForms{}{2d}
      {
        -
        \,;
        \IrrepOfCyclicGroup{\Denominator}{\ShiftedLevel}
      }
      \vert_{\DeRhamDifferential = 0}
    }
    { \CharacterGroup{\CyclicGroup{\ShiftedLevel}} }
    \ar[
      d,
      start anchor={[yshift=5pt]},
      "{
        \mbox{
          \tiny
        }
      }"{xshift=-2pt}
    ]
    \\
    \HomotopyQuotient
    {
      \shape
      \big(
        \FredholmOperators^{\CyclicGroup{\ShiftedLevel}}
      \big)
    }
    { \CharacterGroup{\CyclicGroup{\ShiftedLevel}} }
    \ar[
      r,
      "{
        \mbox{
          \tiny
          \eqref{CRationalizationOfSpaceOfFredholmOperators}
        }
      }"{swap}
    ]
    \ar[
      rr,
      rounded corners,
      to path={
           ([yshift=-00pt]\tikztostart.south)
        -- ([yshift=-14pt]\tikztostart.south)
        -- node[above]{\scalebox{.7}{$
             \HomotopyQuotient
               { \mathbf{ch}^{\CyclicGroup{\ShiftedLevel}} }
               { \CharacterGroup{\CyclicGroup{\ShiftedLevel}} }
           $}}
           node[below]{
             \tiny
             {\color{greenii} \bf
             twisted equivariant Chern character}
             on $\mathbb{A}$-type singularity
           }
           ([yshift=-7pt]\tikztotarget.south)
        -- ([yshift=+07pt]\tikztotarget.south)
      }
    ]
    &
    \HomotopyQuotient
    {
      \Big(
        \CRationalization \,
        \shape
        \big(
          \FredholmOperators^{\CyclicGroup{\ShiftedLevel}}
        \big)
      \Big)
    }
    { \CharacterGroup{\CyclicGroup{\ShiftedLevel}} }
    \ar[
      r,
      "{\sim}",
      "{
        \mbox{
          \tiny
          \eqref{CRationalizationOfGFixedFRedholmOperators}
        }
      }"{swap}
    ]
    &
    \HomotopyQuotient
    {
      \underset{
        1 \leq \, \Denominator \, \leq \NumberOfPunctures
      }{\bigoplus}
      \;
      \underset{
        d \in \NaturalNumbers
      }{\bigoplus}
      \,
      B^{2d}
      \DeRhamComplex{}
        {
          -
          \,;
          \IrrepOfCyclicGroup{\Denominator}{\ShiftedLevel}
        }
    }
    { \CharacterGroup{\CyclicGroup{\ShiftedLevel}} }
  \end{tikzcd}
\end{equation}

This is the moduli stack for
\TED-K-theory of $\CyclicGroup{\ShiftedLevel}$-singularities
(as appearing in Prop. \ref{TheBareBProfiniteIntegersTwistOfComplexRationalizedAEquivariantKTheory}), which is now defined,
analogously to \eqref{TwistedEquivariantKTheoryAsHomotopyClassesOfSections}, as:
$$
  \mathrm{KU}
    ^{1 + \NumberOfProbeBranes + [\FlatConnectionForm]}
    _{\diff}
  \big(
    \SmoothManifold
    \times
    \HomotopyQuotient
      {\ast}
      { \CyclicGroup{\ShiftedLevel} }
  \big)
  \;\;
  =
  \;\;
  \left\{\!\!\!\!
  \adjustbox{raise=4pt}{
  \begin{tikzcd}
    &&
    \HomotopyQuotient
      {
        \DifferentialKU
          ^{\CyclicGroup{\ShiftedLevel}}
      }
      { \CharacterGroup{\CyclicGroup{\ShiftedLevel}} }
    \ar[d]
    \\
    \ComplexManifold
    \ar[
      rr,
      "{ [\FlatConnectionForm] }"{swap}
    ]
    \ar[
      urr,
      dashed
    ]
    &&
    \mathbf{B}
    \CharacterGroup{\CyclicGroup{\ShiftedLevel}}
  \end{tikzcd}
  }
  \!\!\!\! \right\}_{\big/\homotopy}
  .
$$

Forming the homotopy fiber sequence of the twisted equivariant Chern character map \eqref{FiberProductDefiningTEdKThepryOnSingularity}
we obtain first the
moduli stack for the flat version
and then the secondary Chern character in TED-K-theory of $\mathbb{A}$-type singularities, yielding the stacky avatar of a
twisted and equivariant version of the differential cohomology hexagon \eqref{TheDifferentialCohomologyHexagonForKTheory}:
$$
\hspace{-2mm}
  \begin{tikzcd}[column sep=8pt]
    \HomotopyQuotient
    {
      \underset{
        1 \leq \, \Denominator \, \leq \NumberOfPunctures
      }{\bigoplus}
      \;
      \underset{
        d \in \NaturalNumbers
      }{\bigoplus}
      \,
      \mathbf{B}^{1 + 2d}
      \DeRhamComplex{}
        {
          -
          \,;
          \IrrepOfCyclicGroup{\Denominator}{\ShiftedLevel}
        }
    }
    { \CharacterGroup{\CyclicGroup{\ShiftedLevel}} }
    \ar[dr]
    \ar[rr]
    &
    &
    \HomotopyQuotient
    {
      \DifferentialKU^{\CyclicGroup{\ShiftedLevel}}
    }{ \CharacterGroup{\CyclicGroup{\ShiftedLevel}} }
    \ar[
      rr,
      " { \scalebox{0.8}{$
        \HomotopyQuotient
        {
        \mathbf{ch}
          ^{\CyclicGroup{\ShiftedLevel}}
          _{\diff}
        \;}
        {\; \CharacterGroup{\CyclicGroup{\ShiftedLevel}} }
      $}}"
    ]
    \ar[dr]
    &
    &
    \HomotopyQuotient
    {
      \underset{
        1 \leq \, \Denominator \, \leq \NumberOfPunctures
      }{\bigoplus}
      \;
      \underset{
        d \in \NaturalNumbers
      }{\bigoplus}
      \,
      \mathbf{B}^{2d}
      \DeRhamComplex{}
        {
          -
          \,;
          \IrrepOfCyclicGroup{\Denominator}{\ShiftedLevel}
        }
    }
    { \CharacterGroup{\CyclicGroup{\ShiftedLevel}} }\;.
    \\
    &
    \HomotopyQuotient{
      \mathrm{KU}
        ^{ \CyclicGroup{\ShiftedLevel} }
        _{\flat}
    }{ \CharacterGroup{\CyclicGroup{\ShiftedLevel}} }
    \ar[rr]
    \ar[ur]
    &&
    \HomotopyQuotient
    {
      \shape
      \big(
        \FredholmOperators
          ^{\CyclicGroup{\ShiftedLevel}}
      \big)
    }{ \CyclicGroup{\ShiftedLevel} }
    \ar[
      ur,
      " { \scalebox{0.8}{$
        \HomotopyQuotient
        {
          \mathbf{ch}
            ^{\CyclicGroup{\ShiftedLevel}}
        }
        { \CharacterGroup{\ShiftedLevel} }
        $}
      }"{sloped, swap}
    ]
  \end{tikzcd}
$$
Finally,
the equivalence classes of maps into these smooth
stacks, sliced over $\mathbf{B}\CyclicGroup{\ShiftedLevel}$,
gives the cohomology groups:

\vspace{-.4cm}
\begin{equation}
  \label{ObtainingTEdKTheoryAndItsChernCharacter}
  \hspace{-4mm}
\begin{array}{c}
  \raisebox{-11pt}{$
  \underset{
    {d \in \NaturalNumbers}
    \atop
    {1 \leq \, \Denominator \, \leq \ShiftedLevel}
  }{\bigoplus}
  H^{1 + 2d }
  \big(
    \DeRhamComplex{}
      {\SmoothManifold ;\, \ComplexNumbers}
    ,
    \DeRhamDifferential
      +
    \Denominator
      \cdot
    \FlatConnectionForm
  \big)
  $}
  \\
    \rotatebox{-90}{
    \hspace{-12pt}
    $\simeq$
  }
  \\
  \\
  \left\{
  \;\;
  \adjustbox{raise=4pt}{
  \begin{tikzcd}[column sep=6pt]
    &&
    \HomotopyQuotient
      {
        \mathllap{
        \underset
          {
            \mathclap{
            { d \in \NaturalNumbers }
            \atop
            { {1 \leq \, \Denominator \, \leq \ShiftedLevel}  }
            }
          }
          {\bigoplus}
          \,
          \mathbf{B}^{1+2d}
          }
        \;
        \DeRhamComplex{}
          {
            -;
            \IrrepOfCyclicGroup{\Denominator}{\ShiftedLevel}
          }
      }
      { \CharacterGroup{\CyclicGroup{\ShiftedLevel}} }
    \ar[
      d,
      shorten <=-11pt
    ]
    \\
    \ComplexManifold
    \ar[
      rr,
      "{ [\FlatConnectionForm] }"{swap}
    ]
    \ar[
      urr,
      dashed,
      shorten >=-14pt
    ]
    &&
    \mathbf{B}
    \CharacterGroup{\CyclicGroup{\ShiftedLevel}}
  \end{tikzcd}
  }
  \!\!\!\!\!\!
  \right\}_{\!\!\!\!\!\!\big/\homotopy}
\end{array}
\hspace{-0.4cm}
\begin{tikzcd}[column sep=14pt]
  {}
  \ar[
    rr,
    shift left=42pt,
    "{
    }"
  ]
  &&
  {}
\end{tikzcd}
\hspace{-.6cm}
  \begin{array}{c}
  \DifferentialKU^{ 0 + [\FlatConnectionForm] }
  \big(
    \ComplexManifold
    \times
    \HomotopyQuotient
      {\ast}
      { \CyclicGroup{\ShiftedLevel} }
  \big)
  \\
  \rotatebox{-90}{
    $:=$
  }
  \\
  \\
  \left\{
  \!\!\!\!\!
  \adjustbox{raise=4pt}{
  \begin{tikzcd}[column sep=4pt]
    &&
    \HomotopyQuotient
      {
        \DifferentialKU^{\CyclicGroup{\ShiftedLevel}}
      }
      { \CharacterGroup{\CyclicGroup{\ShiftedLevel}} }
    \ar[d]
    \\
    \ComplexManifold
    \ar[
      rr,
      "{ [\FlatConnectionForm] }"{swap}
    ]
    \ar[
      urr,
      dashed
    ]
    &&
    \mathbf{B}
    \CharacterGroup{\CyclicGroup{\ShiftedLevel}}
  \end{tikzcd}
  }
  \!\!\!\!\!
  \right\}_{\!\!\!\!\!\big/\homotopy}
\end{array}
\hspace{-0.6cm}
\begin{tikzcd}[column sep=19pt]
  {}
  \ar[
    rr,
    shift left=42pt,
    "{ \scalebox{0.8}{$
      \big(
        \HomotopyQuotient
        {
        \mathbf{ch}
          ^{\CyclicGroup{\ShiftedLevel}}
          _{\diff}
        \;}
        {\; \CharacterGroup{\CyclicGroup{\ShiftedLevel}} }
      \big)_{\!\ast}
      $}\;\;
    }"
  ]
  &&
  {}
\end{tikzcd}
\hspace{-.9cm}
\begin{array}{c}
  \raisebox{-11pt}{$
  \underset{
    {d \in \NaturalNumbers}
    \atop
    {1 \leq \, \Denominator \, \leq \ShiftedLevel}
  }{\bigoplus}
  H^{2d }
  \big(
    \DeRhamComplex{}
      {\SmoothManifold}
    ,
    \DeRhamDifferential
      +
    \Denominator
      \cdot
    \FlatConnectionForm
  \big)
  $}
  \\
  \rotatebox{-90}{
    \hspace{-12pt}
    $\simeq$
  }
  \\
  \\
  \left\{
  \!\!\!\!
  \adjustbox{raise=4pt}{
  \begin{tikzcd}[column sep=6pt]
    &&
    \HomotopyQuotient
      {
        \mathllap{
        \underset
          {
            \mathclap{
            { d \in \NaturalNumbers }
            \atop
            { {1 \leq \, \Denominator \, \leq \ShiftedLevel}  }
            }
          }
          {\bigoplus}
          \,
          \mathbf{B}^{2d}
          }
        \;
        \DeRhamComplex{}
          {
            -;
            \IrrepOfCyclicGroup{\Denominator}{\ShiftedLevel}
          }
      }
      { \CharacterGroup{\CyclicGroup{\ShiftedLevel}} }
    \ar[
      d,
      shorten <=-11pt
    ]
    \\
    \ComplexManifold
    \ar[
      rr,
      "{ [\FlatConnectionForm] }"{swap}
    ]
    \ar[
      urr,
      dashed,
      shorten >=-14pt
    ]
    &&
    \mathbf{B}
    \CharacterGroup{\CyclicGroup{\ShiftedLevel}}
  \end{tikzcd}
  }
  \!\!\!\!\!\!
  \right\}_{\!\!\!\!\!\!\big/\homotopy}
  \end{array}
\end{equation}
(Here the equivalences on the left and right follow by
\eqref{HolomorphicDeRhamComplexOverSteinManifolds}
and
\eqref{TwistedDeRhamIsomorphicToEquivariantDeRhamOnCoveringSpace},
after
using that the simplicial presheaf $\mathbf{B}^{2d}\DeRhamComplex{}{-}$ satisfies descent over smooth manifolds, so that the sliced hom-space above is equivalently that of $\pi_1(\SmoothManifold) \xrightarrow{\;} \CharacterGroup{\CyclicGroup{\ShiftedLevel}}$-equivariant de Rham cohomology on the universal cover of $\SmoothManifold$.)

This presentation
\eqref{ObtainingTEdKTheoryAndItsChernCharacter}
now makes manifest
the claim of Prop. \ref{TheBareBProfiniteIntegersTwistOfComplexRationalizedAEquivariantKTheory}.

\end{itemize}

\newpage

\section{Anyonic defect branes in TED-K-theory}
\label{QuantumStatesAndObservables}

In this section we match
the mathematical results in \cref{ConformalBlocksAsTEdKTheory}
to
computations and expectations in the string theory literature, concerning the nature of codimension=2 defect branes.
The result  supports the notion that the expected exotic defect brane charges are well reflected in K-theory -- if this is understood in its refined incarnation as \TED-K-theory (\cref{OneTwistedCohomologyAsTwistedEquivariantKTheory}). Of course, it is generally expected that all of  equivariant \& twisted \& differential K-theory is necessary for measuring D-brane charge
(see Rem. \ref{HypothesesAboutBraneChargeQuantization} below), but the relevant sector according to Prop. \ref{TheBareBProfiniteIntegersTwistOfComplexRationalizedAEquivariantKTheory} has not found attention before.

\medskip
We conclude
below by highlighting that a full match requires considering
\TED-K-theory not just of spacetime manifolds transversal to defect branes, but also of the configuration spaces of points inside these transversal spaces; and we close,
around Rem. \ref{MThreeBraneModuliViaHypothesisH},
by explaining how this squares with our previously discussed
{\HypothesisH}
about the nature of brane charges in M-theory. First to recall this and related hypotheses:

\noindent
\begin{remark}
[{\bf Hypotheses about brane charge quantization.}]
\label{HypothesesAboutBraneChargeQuantization}
Despite all existing discussion of brane physics in non-perturbative string theory, an complete theory of branes has been missing, and the rules followed by stable brane charges have remained hypothetical. We briefly indicate three such hypotheses and some of their interrelations (the first two of these are well-known but do not have established names, so
we introduce the following terms for ease of reference):

\vspace{-2mm}
\begin{itemize}[leftmargin=*]
\setlength\itemsep{-1pt}

\item
\hypertarget{HypothesisD}{}
{\bf Hypothesis D} -- the hypothesis
(\cite{AspinwallDouglas}\cite{DouglasFiolRomelsberger}, review includes \cite{Aspinwall+}\cite{Pietromonaco17}\cite[p. 3-4]{CollinucciSavelli14})
that stable D-brane charge, at least on complex-analytic spaces such as Calabi-Yau manifolds, is reflected in stability conditions \cite{Bridgeland}
on the homotopy category of chain complexes (``derived category'') of coherent sheaves on these complex spaces, hence essentially of (possibly degenerate) holomorphic vector bundles (\cite{Block05}).

The $\Integers$-grading $V_\bullet$ on these chain complexes has been motivated, somewhat vaguely, from a ghost degree seen in topological strings, while the differential $V_\bullet \xrightarrow{\;} V_{\bullet + 1}$ models the
tachyon field between brane/anti-brane Chan-Paton bundles. This latter aspect is the compelling one, but what it really motivates is not so much a differential on a chain complex, but a Fredholm operator on a $\Integers/2$-graded vector bundle -- as in \eqref{AFredholmOperator} -- whose (co)kernel is the stable (anti-)brane charge that remains after tachyon condensation.

\smallskip
Since Fredholm operators represent topological K-theory, this motivates  instead (see \cite[\S 3]{Witten98}):

\item
\hypertarget{HypothesisK}{}
{\bf Hypothesis K} -- the hypothesis that D-brane charge is measured in some (twisted, equivariant, differential) version of topological K-theory of spacetime (\cite{Witten98}, see \cite{SzaboValentino07}\cite[\S 1]{BMSS19}\cite{GS-RR} for further pointers and details).

But plain topological K-theory -- as the name suggests -- is insensitive to the complex-analytic structure that is so crucial in motivating {\it Hypothesis D} above. This striking disconnect between the two proposals for D-brane charge quantization
seems to have received little systematic attention, though it clearly suggests to consider some form of K-theory of holomorphic vector bundles (a point made in \cite{Sharpe99}, see also \cite[\S 5.3.3]{Scheidegger01}).

We observe that the holomorphic vector bundle structure must be thought of as part of the {\it differential} enhancement of K-theory:

\vspace{-2mm}
\begin{itemize}[leftmargin=*]
\setlength\itemsep{-1pt}

\item[{\bf 1\rlap{.}}]
{\bf Holomorphic structure.}
The Koszul-Malgrange Theorem
(e.g. \cite[Thm. 2.1.53]{DonaldsonKronheimer97})
says that holomorphic vector bundles
are equivalently complex vector bundles equipped with a flat anti-holomorphic covariant derivative.
This already implies a purely differential-geometric incarnation of
the derived category of coherent sheaves over a complex manifold
(\cite[\S 4]{Block05}),
and hence of the categories of D-branes according to Hypothesis D (as highlighted in \cite{Bergman08}).

\item[{\bf 2\rlap{.}}]
{\bf Holomorphic connections.}
Completing this to a
{\it holomorphic connection} on a holomorphic vector bundle implies that the Chern character
{\it vanishes} and, under mild conditions, that the connection is in fact {\it flat}
(e.g. \cite{Biswas98}). Therefore, Chern character invariants  of holomorphic connections on
holomorphic vector bundles are necessarily {\it secondary} invariants
(e.g. \cite{DHZ92},
i.e., of Cheeger-Simons-type, see \cite[\S 4.3]{FSS20CharacterMap}),
and thus
(cf. Rem. \ref{TwistedCohomologyOverSteinManifolds})
``holomorphic K-theory'' over complex manifolds is largely a topic in {\it flat differential} K-theory
\eqref{FiberProductDefiningHolomorphicDifferentialKTheory}.

\end{itemize}

\item
\hypertarget{HypothesisH}{}
{\bf Hypothesis H} is the hypothesis
(\cite{FSS19b}\cite{FSS19HopfWZ}\cite{SS20M5Anomaly}\cite{FSS20TwistedString}\cite{SS19TadpoleCancellation}\cite{SS19ConfigurationSpaces}\cite{SS21MF}, initiated in \cite[\S 2.5]{Sati13})
that: \footnote{Following \cite{FSS20CharacterMap}, we say ``abelian cohomology'' for Whitehead-generalized cohomology theories represented by spectra, and ``non-abelian cohomology'' for cohomology theories represented by unstable spaces, such as those classifying principal bundles, Giraud gerbes or unstable Cobordism/Cohomotopy.}

\vspace{-2mm}
\begin{itemize}[leftmargin=*]
\setlength\itemsep{-1pt}

\item M-brane charge is in (tangentially twisted,  equivariant, differential) non-abelian {\it Cohomotopy} theory;

\item the full Cohomotopy {\it cocycle spaces} (of which Cohomotopy classes are just the connected components) constitute the moduli spaces of M-brane configurations;

\item the abelian cohomology {\it of} these Cohomotopy cocycle spaces/moduli spaces reflects the quantum states of M-branes.

\end{itemize}
\vspace{-.1cm}

For example (see \cite{FSS16}\cite{BMSS19}): the Cohomotopy cocycle space of the M-theory circle (hence
the moduli of M-branes in their ``double dimensional reduction'' to type IIA string theory) is rationally equivalent to essentially the classifying space of twisted complex K-theory, and hence itself carries a
canonical class in twisted K-theory, the pullback of which to 10d spacetime yields the usual K-theory
classes of D-branes (all at least in rational approximation).

Similarly (see \cite{SS19ConfigurationSpaces}\cite{CSS21}), the non-abelian Cohomotopy cocycle space of the
space transverse to intersecting D6/D8-branes, properly formulated, is the configuration space of points
in the 3-space transverse to the D6s inside the D8s (as befits the space of moduli  of $\mathrm{D6}\!\!\perp\!\!\mathrm{D9}$ -intersections) and the abelian cohomology of this configuration
space reflects much of the quantum structure expected of such brane intersections.

In short, according to {\it Hypothesis H}, the brane charges and their moduli are fundamentally
in non-abelian Cohomotopy theory,
while their quantum states are reflected in the abelian cohomology {\it of}, in turn,
the Cohomotopy cocycle spaces (cf. \cite{GS-Postnikov}), which tend to be related to configuration spaces;
and here twisted K-theory is singled out as being the abelian cohomology theory which canonically observes,
on the Cohomotopy cocycle spaces, the double dimensional reduction from 11d to type II string theory,
thereby connecting back to
\hyperlink{HypothesisK}{\it Hypothesis K} above.

\end{itemize}
\end{remark}

\noindent
We now explain what the discussion in
\cref{ConformalBlocksAsTEdKTheory} says about brane charges,
under these hypotheses:

\noindent
\begin{remark}
[{\bf Exotic defect branes and Anyon statistics.}]
\label{ExoticDefectBranesAndAnyonStatistics}
It has been argued
\cite[p. 12]{deBoerShigemori12}
that all ``exotic branes'' in string theory (e.g. \cite{BermanMusaevOtsuki19}) are equivalently codim=2 defect branes with U-duality group monodromy around the point where they puncture their transverse 2-dimensional space;
 and it has been speculated \cite[p. 65]{deBoerShigemori12} that this may be understood as realizing anyon statistics in string theory, though any details have remained open.

But
under \hyperlink{HypothesisK}{\it Hypothesis K}, the statement of
Prop. \ref{Degree1ConformalBlockInTEdKTheory} and
Thm. \ref{ConformalBlockInTEdKTheory}
is that defect brane charges subsume conformal blocks -- these, in their dependence on the punctured transverse space, solve the Knizhnik-Zamolodchikov (KZ) equation and thus constitute {\it monodromy braid representations} (e.g. \cite{Kohno87}\cite{GuHaghighatLiu21}, review in \cite{TodorovHadjiivanov01} \cite[\S 1]{GuHaghighatLiu21}).
This is exactly what reflects anyon statistics (e.g., \cite{Lerda92}\cite{Rao16}), here now of defect branes as they are braided around each other in their transverse space (\hyperlink{ThePuncturedPlane}{\it Figure 1}).

Indeed, it is well-understood that anyons in 2 spatial dimensions may be identified with primary fields (Rem. \ref{IntegrableAndAdmissibleWeights}) of a (rational chiral 2d) CFT, such that their wavefunctions are identified with the corresponding conformal blocks  (this is due to \cite{MooreRead91}, further developments in \cite{GuHaghighatLiu21}\cite{ZWXT21}, review in \cite[\S 9]{Lerda92}\cite[\S 8.3]{Wang10}\cite{Su18}).
In fact, realistic anyon species are described by the $\suTwoAffine{k}$-WZW model CFT, see p. \pageref{SUTwoANyonSpecies} below.

Following this identification, it has become customary to take the unitary modular tensor categories (UMTC) formed by primary CFT fields and their fusion rules
to be the very definition of the given ``topological order'' exhibited by a topological (i.e. gapped) phase of matter (e.g. \cite[Def. 1.3]{Delaney19}).

\end{remark}
We come back to this anyonic aspect of defect branes in \cref{Outlook}. Here we continue with discussing the evidence that the \TED-K-theory of the punctured plane reflects further expected properties of defect branes in F/M-theory:

\medskip

\noindent
\begin{remark}
[{\bf D7-brane charges expected in type IIB/F-theory.}]
\label{DSevenBraneChargesExpeectedInFTheory}
The most evident defect branes in string theory are D7-branes in type IIB string theory \cite{Bershoeffetal96}.
Their naive charge lattice, under \hyperlink{HypothesisK}{\it Hypothesis K}, is the (reduced) K-theory of their compact(ified) transverse space (see Rem. \ref{DBraneCHargesAndRRFields} below for background), which in the basic situation of a transverse plane means:
\begin{equation}
  \label{NaiveD7BraneCharge}
  \left\{
  \!\!\!\!\!\!\!\!
  \mbox{
    \small
    \def\arraystretch{1}
    \begin{tabular}{c}
      D7-brane charges as
      \\
      seen in plain K-theory
    \end{tabular}
  }
  \!\!\!\!\!\!
  \right\}
  \;\;\;=\;\;\;
  {\mathrm{KU}}^{0}
  \big(
    \mathbb{R}^{7,1}_+
    \wedge
    \mathbb{R}^2_{\mathrm{cpt}}
  \big)
  \;\;\simeq\;\;
  {\mathrm{KU}}^{0}
  \big(
    S^2
  \big)
  \;\;\simeq\;\;
  \Integers
  \,.
\end{equation}
This says (or would say) that there is a single species of D7-branes which comes in integer multiples (of ``coincident branes''). While routinely stated
(e.g. \cite[\S 2.1]{BergmanGimonHorava99}\cite[Tab. 3]{OlsenSzabo99}\cite[\S 4.2]{Schwarz01}, following \cite[\S 4.1]{Witten98}\cite[\S 2]{Horava98}), this is not actually expected to be even close to the full answer:
In the hypothetical non-perturbative completion of type IIB string theory that goes by the working title {\it F-theory} \cite{Vafa96}\cite{Sen96},
one instead expects all of the following
(see also analogous earlier discussion of
codim=2
``stringy cosmic strings'' in \cite[\S 3]{GSVY90}), compare \hyperlink{FigureD7BraneConfiguration}{\it Figure 2} below:

\vspace{-1mm}
\begin{itemize}
[leftmargin=15pt]
\setlength\itemsep{-1pt}

\item[\bf (i\rlap{)}]
{\bf Complex structure.}
The space transverse to D7-branes carries the structure
of a complex curve $\Sigma^2$
(e.g. \cite[(2.4)]{Sen96} \cite[(9.190)]{BeckerBeckerSchwarz06}, nowadays traditionally denoted ``$B$'', e.g. \cite[(2.40)]{Weigand18}), being the base of an elliptic fibration \eqref{MFEllipticFibration}
whose fibers degenerate at the locus of each $\mathrm{D7}_{I}$, which hence appears as a puncture $z_{{}_{I}}$ in $\Sigma^2$ (e.g. \cite[(2.5)]{Sen96}\cite[(9.206)]{BeckerBeckerSchwarz06}\cite[\S 3.2]{Weigand18}).

\item[\bf (ii\rlap{)}]
{\bf Axio-Dilaton field.}
With respect to this complex structure,
 the RR-field $F_1$, which measures the flux through 1-spheres around the D7-branes,
 in its S-duality covariant incarnation where its potential is the {\it axio-dilaton field}
 (cf. \cite[(2.4)]{Sen96}\cite[\S 3.1]{BayntonBurgessNierop09}\cite[(2.5)]{Weigand18})

\vspace{-.6cm}
\begin{equation}
   \label{TheAxioDilaton}
   \tau
   \;:=\;
   \overset{
     \mathclap{
     \raisebox{3.6pt}{
       \tiny
       \color{darkblue}
       \bf
       axion
     }
     }
   }{
     C_0
   }
   +
   \ImaginaryUnit
   \overset{
     \mathclap{
     \raisebox{1pt}{
       \tiny
       \color{darkblue}
       \bf
       dilaton
     }
     }
   }{
     e^{\phi}
   }
   \;\;\;
   \in
   \;
   \Maps{\big}
     {\, \widehat{\Sigma^2} }
     {
       \ComplexNumbers \vert_{ \mathrm{Re} > 0}
     }
   \,,
 \end{equation}
 is constrained (by supersymmetry)
 in the vicinity of the $I$-th D7 puncturing at $z_{{}_I}$,
 to have the following
 ``profile''
 (\cite[p. 2]{Vafa96}\cite[(2.7)]{Sen96}\cite[(2.7)]{Weigand18}):
 \begin{equation}
   \label{AxioDilatonProfile}
   F_1
   \;:=\;
   \DeRhamDifferential \tau
   \;:=\;
   \DeRhamDifferential
   \big(
     C_0
       +
     \ImaginaryUnit e^{\phi}
   \big)
   \;\;
   \overset{!}{=}
   \;\;
   \tfrac{1}{2\pi\ImaginaryUnit}
   \DeRhamDifferential
   \log(z - z_{{}_I})
   \,+\,
   \Differential
   \big(
   \mbox{{\small terms regular in} $(z - z_{{}_I})$}
   \big)
   \,.
 \end{equation}

\item[\bf (iii\rlap{)}]
{\bf $\mathbf\SLTwoZ$-Multiplets.}
The S-duality group $\SLTwoZ$ is meant to act
on the possible charges carried by each $\mathrm{D7}_{I}$,
making them transform in $\SLTwoZ$ representations,
often assumed to be 2-dimensional (``doublets'' or ``$(p,q)$''-branes, e.g. \cite[\S 2.1]{Weigand18}), but possibly also higher dimensional (``multiplets'', such as triplets \cite{MeessenOrtin98}).

In particular, the shift $\tau \mapsto \tau + 1$ which is picked up by the axio-dilaton \eqref{TheAxioDilaton} as a probe encircles one of the D7-branes, according to \eqref{AxioDilatonProfile}, is meant to correspond to the action of the corresponding element $\TOperator \,\in\, \SLTwoZ$,
from \eqref{MoebiusTransformations},
on the D7-brane charges
 (e.g. \cite[(2.12)]{Weigand18}).

\item[{\bf (iv\rlap{)}}]
{\bf 24 D7-branes to cancel tadpoles.}
Consistency of the {\it total} D7-brane charge under the shift action $\tau \mapsto \tau + 1$ is meant to require the presence of exactly 24 D7-branes (\cite[(2.5)]{Sen96}).

\item[\bf (v\rlap{)}]
{\bf $\SpecialUnitaryGroup(\ShiftedLevel)$-gauge fields
on D3-branes
at $\mathbb{A}_{\ShiftedLevel - 1}$.}
Type IIB string theory exists over ADE-orbifolds, notably over $\mathbb{A}_{\ShiftedLevel-1}$-type orbifolds  $\HomotopyQuotient{\Quaternions}{\CyclicGroup{\ShiftedLevel}}$
(\cite{DouglasMoore96}\cite{JohnsonMyers96})
and D7-branes extending over such orbifolds
contain bound D3-branes, transverse to the singular locus, which carry $\mathcal{N}=2$, $D=4$ super Yang-Mills theory with gauge group $\SpecialUnitaryGroup(\ShiftedLevel)$
(\cite{BanksDouglasSeiberg96}\cite[\S 3]{BlumIntriligator97}\cite[p. 7]{BDVFLM02}
%
%
%
%
%
): This is {\it Seiberg-Witten theory} in its incarnation via {\it geometric engineering} (\cite{KKV96}) in type IIB string theory (review in \cite[\S 4]{Lerche96}).

\item[\bf (vi\rlap{)}]
{\bf Transversal $\suTwoAffine{\Level}$ D-geometry.}
The conformal field theory of open superstrings
on a worldsheet $\Sigma^2$
stretching between the D7-branes
and probing the ambient $\mathbb{A}_{\ShiftedLevel-1}$-singularity
(as in \cite{Kutasov96})
is controlled by the $\suTwoAffine{\Level}$-WZW model at
just the shifted level
$k = \ShiftedLevel - 2$
\eqref{ShiftedLevel}:

\vspace{-.4cm}
\begin{equation}
  \label{ShiftOfLevelViaStringsOnALE}
  \mbox{
    LG-CFT with target
    $
    \HomotopyQuotient{\Quaternions}{\CyclicGroup{\ShiftedLevel}}
    $
  }
    \;=\;
    \suTwoAffine{\ShiftedLevel - 2}
    \;\oplus\;
    \mbox{4 free fermions}
    \;\oplus\;
    \mbox{1 free scalar}
    \,.
\end{equation}
\vspace{-.6cm}

\noindent
This crucial statement is due to \cite[p. 10-12]{OoguriVafa96}\cite[p. 4]{LLS00}, but may not to have found due attention; see also Rem. \ref{LevelAndRankInTheAGTCorrespondence} below.

In the sense of non-commutative D-brane geometry (as in \cite[\S 7.4]{FGR97}\cite[\S 2]{Douglas99}),
this $\suTwoAffine{\ShiftedLevel-2}$-algebra may be understood
(\cite[bottom of p. 11]{OoguriVafa96})
as encoding the quantum geometry of the complement 3-sphere $S^3 \simeq \SpecialUnitaryGroup(2)$
around \eqref{ComplementOfBraneInSpacetime}
the $\mathbb{A}_{\ShiftedLevel-1}$-singularity.

\item[\bf (vii\rlap{)}]
{\bf $\SpecialUnitaryGroup(\ShiftedLevel)$-BPS states on D3.}
The boundary states of this CFT
\eqref{ShiftOfLevelViaStringsOnALE}
are labeled
(\cite[p. 5]{Lerche00}\cite[\S 2]{LLS00})
by the primary fields  of the $\suTwoAffine{\Level}$-WZW model,
hence
(Rem. \ref{IntegrableAndAdmissibleWeights})
by integrable weights $\weight_{{}_I}$
and encode,
in a stringy realization of the MacKay correspondence,
the BPS states of the $\mathcal{N} =2$ $\SpecialUnitaryGroup(\ShiftedLevel)$ SYM-theory on the worldvolume of the D3-branes
(\cite{Lerche00}\cite{LLS00}).

\vspace{-.2cm}

\end{itemize}

\vspace{-.6cm}
\begin{center}
\hypertarget{FigureD7BraneConfiguration}{}
\hspace{4cm}
\begin{tikzpicture}

  \draw (0,0) node
    {$
      \overset{
      }{
        \mathbb{R}^{3,1}
      }
    $};

  \draw (.7,.56) node
    {$
      \mathrlap{
        \raisebox{8pt}{
        \hspace{-10pt}
        \rotatebox{36}{
        \hspace{-5pt}
        \tiny
        \color{orangeii}
        \bf
        \def\arraystretch{.9}
        \begin{tabular}{l}
          transverse
          \\
          complex curve
        \end{tabular}
        }
      }
      }
      {
        \Sigma^2
      }
    $};

  \draw (1.6+.1,+.05) node
    {
      $
        \HomotopyQuotient
          { \Quaternions }
          { \CyclicGroup{\ShiftedLevel} }
        \mathrlap{
          \hspace{-17pt}
          \raisebox{6pt}{
          \rotatebox{36}{$
            \mathrlap{
            \mbox{
            \tiny
            \color{darkblue}
            \bf
            \def\arraystretch{.9}
            \begin{tabular}{c}
              $\mathbb{A}_{\ShiftedLevel-1}$-type
              \\
              singularity
            \end{tabular}
            }
            }
          $}
          }
        }
      $
    };

  \begin{scope}[shift={(0,-.5)}]

  \draw
    (-1, 0) node
    {$
      \color{orangeii}
      \mathbb{A}_{\ShiftedLevel-2}
    $};

  \draw[
    line width=8pt,
    orangeii
  ]
    (-.4,0) to (.8+.4,0);
 \draw
   (-.4,0) to (.8+.4,0);

 \draw (1.6,0) node
 {$
   \{0\}
 $};

 \end{scope}

  \begin{scope}[shift={(0,-1)}]

  \draw
    (-1, 0) node
    {
      $\mathllap{\mathrm{SL}(2,\Integers)\mbox{-\normalfont{multiplet}}}$
      \color{darkblue}
      $\mathrm{D7}_{{}_I}$
    };

  \draw[
    line width=8pt,
    darkblue
  ]
    (-.4,0) to (0+.4,0);
 \draw
   (-.4,0) to (0+.4,0);

 \draw (.8,0) node
   {$
     \{z_{{}_I}\}
   $};

  \draw[
    line width=8pt,
    darkblue
  ]
    (1.6-.4,0) to (1.6+.4,0);
 \draw
   (1.6-.4,0) to (1.6+.4,0);

 \end{scope}

  \begin{scope}[shift={(0,-1.5)}]

  \draw
    (-1, 0) node
    {
      \color{gray}
      $\mathllap{\mbox{
        $\SpecialUnitaryGroup(\ShiftedLevel)$-SYM
        on
       }}
      \;
      \mathrm{D3}^i$
    };

  \draw[
    line width=8pt,
    gray
  ]
    (-.4,0) to (0+.4,0);
 \draw
   (-.4,0) to (0+.4,0);

 \draw (.8,0) node
   {$
     \{z^i\}
   $};

 \draw (1.6,0) node
   {$
     \{0\}
   $};

 \end{scope}

 \begin{scope}[shift={(6,-.6)}]

  \draw
    (0,-.5+.2)
    .. controls (1,-.5+.2) and (1,+.5+.15) ..
    (2,+.5+.15);
  \draw[gray]
    (0,-.5+.2+.1)
    .. controls (1,-.5+.2+.1) and (1,+.5+.15-.1) ..
    (2,+.5+.15-.1);
  \draw[gray]
    (0,-.5+.2-.1)
    .. controls (1,-.5+.2-.1) and (1,+.5+.15+.1) ..
    (2,+.5+.15+.1);

  \draw[line width=2pt]
    (0,-.7) to node[below] {$\Sigma^2$} (2,-.7);

  \draw[darkblue, line width=3]
    (0,1) to (0,-1);
  \draw[dashed, darkblue, line width=3]
    (0,1.4) to (0,1);

  \draw[darkblue, line width=3]
  (2,1) to
    node[right]
    {$
      \HomotopyQuotient
        { \color{darkblue} \mathrm{D7} }
        { \color{black} \CyclicGroup{\ShiftedLevel} }
    $}
  (2,-1);
  \draw[dashed, darkblue, line width=3]
    (2,1.4) to (2,1);

  \draw (0, 1.56) node
   {\scalebox{.8}{$
     \weight_{{}_I}
   $}};
  \draw (2, 1.56) node
   {\scalebox{.8}{$
     \weight_{{}_J}
   $}};

  \draw (.8,.5)
    node
    {\scalebox{.8}{$
      \suTwoAffine{\ShiftedLevel-2}
    $}};

  \end{scope}

\end{tikzpicture}

\vspace{0mm}
$\,$
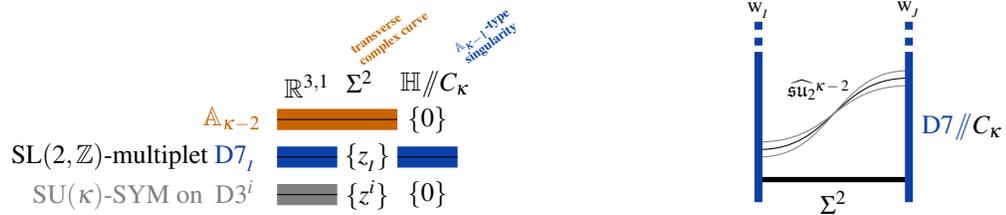
\hspace{1cm} {\footnotesize {\bf Figure 2.} Defect brane configurations in F-theory.
Compare
\hyperlink{ThePuncturedPlane}{\it Figure 1}
and
compare
the M-theoretic configurations in \hyperlink{MBraneConfigurations}{\it Figure 3}.}
\end{center}
\end{remark}

In conclusion, it used to be an open question (maybe never explicitly stated as such) how \hyperlink{HypothesisK}{\it Hypothesis K}, which naively seems to predict the simple answer \eqref{NaiveD7BraneCharge}, can be compatible with this rich charge structure (Rem. \ref{DSevenBraneChargesExpeectedInFTheory}, \hyperlink{FigureD7BraneConfiguration}{\it Figure 2}) actually expected for D7-branes.\footnote{The suggestion to include at least the complex structure on the transverse space $\Sigma^2$ in the charge quantization law for D7-branes is implicit
in discussions of D7-brane charges
under \hyperlink{HypothesisD}{\it Hypothesis D},
e.g. in  \cite[p. 4]{CollinucciSavelli14}\cite[\S 4.1.1]{Schwieger19},
but this alone still does not yield $\mathrm{SL}(2)$-charges.}
We now point out,
culminating in \eqref{ConcludingD7BraneCharges} below,
that when K-theory is understood in the full beauty of  \TED-K-theory according to
\cref{OneTwistedCohomologyAsTwistedEquivariantKTheory}, then the discussion in \cref{ConformalBlocksAsTEdKTheory} provides at least part of the answer to this question:

\begin{remark}
[{\bf D7-Brane Charges seen in \TED-K-Theory.}]
\label{ListOfExoticD7BraneChargeAspects}
According to \cref{ConformalBlocksAsTEdKTheory},
\TED-K-theory of punctured planes transverse to D7-branes exhibits all of the following effects:

\vspace{-.2cm}
\begin{itemize}[leftmargin=*]
\setlength\itemsep{-1pt}

\item
A further field $\FlatConnectionForm(\vec \weight; \ShiftedLevel)$ appears
around D7-branes at $\mathbb{A}$-type singularities,
of which the discussion in \cref{OneTwistedCohomologyAsTwistedEquivariantKTheory} shows that it is a peculiar ``twisted sector'' of the Kalb-Ramond B-field, which may not have found due attention yet.

\item
Its  effect is to attach new charges
$\weight_I \,\in\, \{0, \cdots \ShiftedLevel-1\}$
from \eqref{TheMasterForm}
to the D7-brane punctures $z_{{}_I}$,
naturally identified with
integrable/admissible highest weight representations
of  $\slTwoAffine{\Level}$
(Rem. \ref{IntegrableAndAdmissibleWeights})
and hence with affine characters $\AffineCharacter{\weight_{{}_I}}{\Level}$, from Rem. \ref{SLTwoZActionOnAffineCharacters}.

\item
As such, these new charges canonically transform under an action of $\SLTwoZ$ (Rem. \ref{SLTwoZActionOnAffineCharacters}).

This applies even when the twisted sector B-field $\FlatConnectionForm(\vec v, \ShiftedLevel)$ actually {\it vanishes} -- which is, implicitly, the case considered in previous discussions of F-theory --, because this is just the case where all weights take the (admissible) value $\weight_{{}_I} = 0$.

\item
The affine character that is naturally associated with these weights in generality transforms under the monodromy shift $\tau \mapsto \tau + 1$ \eqref{MoebiusTransformations} by picking up a phase that is a primitive 24th root of unity \eqref{SLTwoZTransformationOfAffineCharacters}. Therefore,  it requires the tensor product of 24 of these charges to get an invariant under the shift operation $\tau \mapsto \tau + 1$
(Ex. \ref{TheTwoDimensionalSLTwoZRepOnAffineCharacters}).

\item
This discrete set of charges parameterizes the full space of secondary charges constituting the degree-1 conformal blocks on $\Sigma^2$ of the $\slTwoAffine{\Level}$-WZW model at level $\Level = \ShiftedLevel -2$ (Prop. \ref{Degree1ConformalBlockInTEdKTheory}).

\end{itemize}
\end{remark}

This \TED-K-theory data (Rem. \ref{ListOfExoticD7BraneChargeAspects}) favorably compares with the expected list of exotic brane charges in Rem. \ref{DSevenBraneChargesExpeectedInFTheory} (\hyperlink{FigureD7BraneConfiguration}{\it Figure 2}).
To fully appreciate this match, it may be worthwhile to step back and discuss general aspects of D-brane charge seen in K-theory, some subtleties of which may not have received due appreciation:

\begin{remark}
[{\bf D-Brane charges and RR-fields in K-Theory}]
\label{DBraneCHargesAndRRFields}
While the topic has a somewhat long history, the following general aspect is crucial but a little subtle and
may need more amplification (cf. \cite[\S 2]{MooreWitten00}).
Given a ``flat brane'', i.e., one  with worldvolume $\mathbb{R}^{p,1}$ linearly embedded into a Minkowski spacetime $\mathbb{R}^{d,1}$, there are {\it two different} spherical domains on which to measure the brane charge:

\begin{itemize}[leftmargin=*]
\setlength\itemsep{-1pt}

\item[{\bf (\rlap{i)}}]
\hspace{2pt}
The {\it complement} of the worldvolume
\vspace{-1mm}
\begin{equation}
  \label{ComplementOfBraneInSpacetime}
  \mathbb{R}^{d,1}
  \setminus
  \mathbb{R}^{p,1}
  \;\;
  \underset{\tiny
    \mathclap{
      \homotopy
    }
  }{
    \simeq
  }
  \;\;
  \mathbb{R}^{d-p} \setminus \{0\}
  \;\;
  \underset{ \tiny
    \mathclap{ \homotopy}
  }{
    \;\simeq\;
  }
  S^{d-p-1}
  \,.
\end{equation}

\vspace{-2mm}
\noindent This models the case of  gravitationally back-reacted {\it black} branes whose actual locus is,
or would be, a singularity. Just like the singularity of a charged black hole -- which is the special
case $d= 3$, $p = 1$, to which the original and archetypical electromagnetic charge quantization argument
due to Dirac applies; this singular brane worldvolume is not actually part of spacetime.

Still, the charge hidden in the singularity is reflected in the
integrated  {\it field line flux} through any sphere enclosing it -- a concept going back all the way to Faraday, which in terms of de Rham cohomology says:
$$
  \begin{tikzcd}[row sep=-2pt, column sep=2pt]
    H^{d-p-1}_{\mathrm{dR}}
    \big(
      \mathbb{R}^{d,1}
      \setminus
      \mathbb{R}^{p,1}
    \big)
    &
    \simeq
    &
    H^{d-p-1}_{\mathrm{dR}}
    \big(
      S^{d-p-1}
    \big)
  \ar[rr, "\sim"]
  &&
  \RealNumbers \;.
  \\
  &&
  \underset{
    \raisebox{-2pt}{
      \tiny
      \color{darkblue}
      \bf
      \def\arraystretch{.9}
      \begin{tabular}{c}
        field flux density
        \\
        on spacetime
      \end{tabular}
    }
  }{
  \scalebox{0.85}{$
    F_{d-p-1}
  $}
  }
  &\longmapsto&
  \underset{
    \raisebox{-2pt}{
      \tiny
      \color{darkblue}
      \bf
      \begin{tabular}{c}
        total flux through sphere
        \\
        around singular brane
      \end{tabular}
    }
  }{
  \scalebox{0.85}{$
  \underset{
    {S^{d - p - 1}}
  }
  {
    \int
  }
  F_{d-p-1}
  $}
  }
  \end{tikzcd}
$$

\vspace{-3mm}
\item[{\bf (\rlap{ii)}}]
\hspace{4pt}
The (Aleksandrov-){\it compactification} of the space to the worldvolume, by adjoining a single ``point at infinity'':
\begin{equation}
  \label{CompactifiedTransverseSpace}
  \mathbb{R}^{p,1}_+
  \wedge
  \big(
    \mathbb{R}^{d,1}
    \!/\,
    \mathbb{R}^{p,1}
  \big)_{\mathrm{cpt}}
  \;\;
  \simeq
  \;\;
  \mathbb{R}^{p,1}_+
  \wedge
  \mathbb{R}^{d-p}_{\mathrm{cpt}}
  \;\;
  \underset{\tiny
    \mathclap{ \homotopy}
  }{
    \simeq
  }
  \;\;
  \mathbb{R}^{d - p }_{\mathrm{cpt}}
  \;\simeq\;
  S^{d-p}
  \,.
\end{equation}

\vspace{-2mm}
\noindent This corresponds to the case of {\it solitonic} branes with respect to non-gravitational
fields: The one-point compactification (together with the use of reduced cohomology) expresses that the topologically non-trivial field configurations which support/constitute the brane vanish (trivialize)
far away from the brane locus, and the fact that the brane locus is not removed from spacetime expresses
that the field is well-behaved everywhere in between, in that its charge density $J$ is finite (i.e.,
well-defined) everywhere:

\vspace{-4mm}
$$
  \begin{tikzcd}[row sep=-4pt, column sep=2pt]
    \mathllap
    H^{d-p}_{\mathrm{dR}}
    \Big(
      \mathbb{R}^{p,1}
      \wedge
      \big(
        \mathbb{R}^{d,1}
        /
        \mathbb{R}^{p,1}
      \big)
    \Big)
    &
    \simeq
    &
    H^{d-p}_{\mathrm{dR}}
    \big(
      S^{d-p}
    \big)
    \ar[rr]
    &&
    \mathbb{R}\;.
    \\
    &&
    \underset{
      \raisebox{-2pt}{
        \tiny
        \color{darkblue}
        \bf
        \def\arraystretch{.9}
        \begin{tabular}{c}
          charge density
          \\
          on spacetime
        \end{tabular}
      }
    }{
        \scalebox{0.85}{$  J_{d-p} $}
    }
    &\longmapsto&
    \underset{
      \raisebox{-2pt}{
        \tiny
        \color{darkblue}
        \bf
        \def\arraystretch{.9}
        \begin{tabular}{c}
          total charge of
          \\
          solitonic brane
        \end{tabular}
      }
    }{
      \scalebox{0.85}{$
    \underset{
      {S^{d-p}}
    }
    {
      \int
    }
    J_{d-p}
    $}
    }
  \end{tikzcd}
$$

\end{itemize}

\vspace{-3mm}
In both cases, the effective space on which to measure flat brane charge is a (higher-dimensional) sphere, but the dimension of that sphere for the case of singular branes \eqref{ComplementOfBraneInSpacetime} differs from that for non-gravitational solitons \eqref{CompactifiedTransverseSpace} by one. Accordingly, the degree of the cohomology theory measuring charges in the corresponding situations differs by one.

\noindent
This reasoning is equivalent to the one that leads to the traditional form of {\HypothesisK},  according to which (\cite[\S 2]{MooreWitten00}):

\vspace{1mm}
-- Type IIB D-brane charge is measured in $\mathrm{KU}^1$.

-- Type IIB RR-field flux is measured in $\mathrm{KU}^0$.
\end{remark}

\noindent
{\bf How to measure D7-brane charge.}
However, the D7-branes of F-theory are manifestly of the singular form \eqref{ComplementOfBraneInSpacetime},
since the axio-dilaton field \eqref{AxioDilatonProfile} diverges on the would-be brane locus, the removal of
which is the pre-requisite for non-trivial monodromy around the brane.
Hence the spacetime manifold on which to measure D7-brane charge is the complement
\eqref{ComplementOfBraneInSpacetime}
of their singular worldvolumes,
On these, if one demands with {\HypothesisK} that D7-brane charge is still measured by complex K-theory in
degree 0, then,
in contrast to \eqref{NaiveD7BraneCharge},
the naive underlying brane charge actually vanishes
\begin{equation}
  \label{KUZeroChargeOfCircle}
  \mathrm{KU}^0
  \big(
    \mathbb{R}^{9,1}
    \setminus
    \mathbb{R}^{7,1}
  \big)
  \;=\;
  \mathrm{KU}^0
  \big(
    \mathbb{C} \setminus \{z_0\}
  \big)
  \;=\;
  \mathrm{KU}^0
  \big(
    S^1
  \big)
  \;=\;
  0
  \,,
\end{equation}
while non-trivial RR-field fluxes show up:
\begin{equation}
  \label{KUOneChargeOfCircle}
  \mathrm{KU}^1
  \big(
    \mathbb{R}^{9,1}
    \setminus
    \mathbb{R}^{7,1}
  \big)
  \;=\;
  \mathrm{KU}^1
  \big(
    \mathbb{C} \setminus \{z_0\}
  \big)
  \;=\;
  \mathrm{KU}^1
  \big(
    S^1
  \big)
  \;=\;
  \Integers
  \,.
\end{equation}

This is not a contradiction, just a subtlety when applying {\HypothesisK} to F-theory; in fact it reflects the hallmark effect in F-theory, where the presence of D7-branes is all witnessed by the non-triviality of the axio-dilaton RR-field $F_1$ \eqref{AxioDilatonProfile} around them.
Mathematically, the resolution is that the cohomology theory to use is not plain but, in particular, {\it differential} K-theory \eqref{TheDifferentialCohomologyHexagonForKTheory}, which here is {\it flat} differential K-theory (since also the even-degree Chern character forms necessarily vanish on $S^1$):
Its long exact sequence \eqref{TheDifferentialCohomologyHexagonForKTheory} collapses to a short exact sequence,
since the group of RR-field configurations $\mathrm{KU}^1(S^1) \simeq \Integers$,
from \eqref{KUOneChargeOfCircle},
has no torsion and since the group of underlying D7-brane charges vanishes $\mathrm{KU}^0(S^1) \,\simeq\, 0$, from \eqref{KUZeroChargeOfCircle}:
\vspace{-2mm}
$$
  \begin{tikzcd}[row sep=0pt]
    0
    \,\simeq\,
    \mathrm{ker}(\mathrm{ch}^1)
    \ar[r]
    &
    \overset{
      \mathclap{
      \raisebox{2pt}{
        \tiny
        \color{darkblue}
        \bf
        \def\arraystretch{.9}
        \begin{tabular}{c}
          axio-dilaton
          \\
          RR-fields near
          \\
          $I$-th $\mathrm{D7}$-brane
        \end{tabular}
      }
      }
    }{
      \mathrm{KU}^1
      (S^1)
    }
    \ar[
      rr,
      "{
        \mathrm{ch}^1(S^1)
      }"{swap},
      "{
        \mbox{
          \tiny
          \color{greenii}
          \bf
          Chern character
        }
      }"{yshift=1pt}
    ]
    &&
    H^1_{\mathrm{dR}}
    \big(
      \ComplexPlane
      \setminus \{z_{{}_I}\}
      ;\, \ComplexNumbers
    \big)
    \ar[
      rr,
      "{
        \mbox{
          \tiny
          \color{greenii}
          \bf
          \def\arraystretch{.9}
          \begin{tabular}{c}
            secondary
            \\
            Chern character
          \end{tabular}
        }
      }"
    ]
    &&
    \overset{
      \raisebox{2pt}{
        \tiny
        \color{darkblue}
        \bf
        \begin{tabular}{c}
          charge lattice of
          \\
          $I$-th D7-brane
        \end{tabular}
      }
    }{
      \mathrm{KU}^{0}_{\flat}
      (S^1)
    }
    \ar[r]
    &
    \mathrm{KU}^0
    (S^1)
    \,\simeq\,
    0\;.
    \\
    &
  \scalebox{0.8}{$  1 $}
    &
    \underset{
      \raisebox{+1pt}{
        \tiny
        \eqref{HolomorphicDeRhamComplexOverSteinManifolds}
        \&
        \eqref{GeneratorsForTwistedDeRhamCohomologyOfPuncturedPlane}
      }
    }{
      \longmapsto
    }
    &
    \underset{
      \raisebox{2pt}{
        \tiny
        \color{orangeii}
        \bf
        \begin{tabular}{c}
          \color{darkblue}
          RR-field strength of
          \\
          axio-dilaton field
        \end{tabular}
      }
    }{
    \scalebox{0.86}{$
      \big[
      \DeRhamDifferential
      \log(z - z_1)
      \big]
    $}
    }
  \end{tikzcd}
$$

\vspace{-3mm}
\noindent As indicated in the bottom line, this sequence witnesses, under the natural identifications of Rem. \ref{TwistedCohomologyOverSteinManifolds} and Prop. \ref{OSBasisForTwistedCohomologyOfPuncturedPlane}, the precise form of the axio-dilaton field  \eqref{AxioDilatonProfile} as expected in F-theory.
Moreover,
this state of affairs generalizes to any number $\NumberOfPunctures$ of parallel D7-branes. Indeed,
since
$$
\Sigma^2
\;:=\;
\ComplexPlane \setminus
\{z_1, \cdots, z_{\NumberOfPunctures}\}
\;
\underset{\tiny \homotopy}{\,\simeq\,}
\;
\underset{\NumberOfPunctures}{\bigvee}
\,
S^1
\,,
$$
we have
$$
  \mathrm{KU}^0(\Sigma^2)
  \;\simeq\;
  \underset{
  \scalebox{0.6}{$  1 \leq I \leq \NumberOfPunctures $}
  }{\bigoplus}
  \,
  \mathrm{KU}^0(S^1)
  \;\simeq\;
  0
  \,,
  \;\;\;\;\;\;\;\;\;\;\;\;\;\;
  \mathrm{KU}^1
  \big(
    \Sigma^2
  \big)
  \;\simeq\;
  \underset{
  \scalebox{0.6}{$  1 \leq I \leq \NumberOfPunctures $}
  }{\bigoplus}
  \,
  \mathrm{KU}^{1}
  (S^1)
  \;\simeq\;
  \underset{
 \scalebox{0.6}{$   1 \leq I \leq
    \NumberOfPunctures $}
  }{\bigoplus}
  \;
  \Integers
  \;\;\simeq\;\;
  \big\langle
    1_{I}
  \big\rangle
    _{I = 1}
    ^{\NumberOfPunctures}
  \,.
$$
and hence the long exact sequence \eqref{TheDifferentialCohomologyHexagonForKTheory} now truncates to:
\vspace{-2mm}
\begin{equation}
  \label{SESForD7BraneChargeOnSigmaTwo}
  \begin{tikzcd}[row sep=3pt]
    0
    \ar[r]
    &
    \overset{
      \raisebox{2pt}{
        \tiny
        \color{darkblue}
        \bf
        \begin{tabular}{c}
          axio-dilaton fields
          \\
          around D7-branes
        \end{tabular}
      }
    }{
    {\mathrm{KU}}^1
    \big(
      \Sigma^2
    \big)
    }
    \ar[
      rr,
      "{\mathrm{ch}^1}"{swap},
      "{
        \mbox{
          \tiny
          \color{greenii}
          \bf
          Chern character
        }
      }",
      hook
    ]
    &&
    \overset{
      \raisebox{2pt}{
        \tiny
        \color{orangeii}
        \bf
        \begin{tabular}{c}
          axio-dilaton fields with
          \\
          secondary D7-brane charges
        \end{tabular}
      }
    }{
      H^1_{\mathrm{dR}}
      \big(
        \Sigma^2
        ;\,
        \ComplexNumbers
      \big)
    }
    \ar[
      rr,
      ->>,
      "{
        \mbox{
          \tiny
          \color{greenii}
          \bf
          \def\arraystretch{.9}
          \begin{tabular}{c}
          secondary
          \\
          Chern character
          \end{tabular}
        }
      }"
    ]
    &&
    \overset{
      \mathclap{
      \raisebox{2pt}{
        \tiny
        \color{darkblue}
        \bf
        \begin{tabular}{c}
          secondary charges of
          \\
          singular D7-branes
        \end{tabular}
      }
      }
    }{
      {\mathrm{KU}}^0_{\flat}
      \big(
        \Sigma^2
      \big)
    }
    \ar[r]
    &
    0\;.
    \\[-5pt]
    &
    \scalebox{0.8}{$
      1_{{}_I}
    $}
    &
    \underset{
      \raisebox{+1pt}{
        \tiny
        \eqref{HolomorphicDeRhamComplexOverSteinManifolds}
        \&
        \eqref{GeneratorsForTwistedDeRhamCohomologyOfPuncturedPlane}
      }
    }{
      \longmapsto
    }
    &
    \underset{
    }{  \scalebox{0.8}{$
    \big[
      \DeRhamDifferential
      \log(z-z_{{}_I})
    \big]
    $}
    }
    &&
  \end{tikzcd}
\end{equation}

\vspace{-3mm}
\noindent The left part of this sequence formalizes the expectation that the integral charge of the singular D7-branes of F-theory is reflected not in their charge densities, as in the naive equation \eqref{NaiveD7BraneCharge}, but in the monodromy of the axio-dilaton field around them (as it should be). On the other hand, the right part of this sequence says that, in addition to this traditional expectation, a {\it secondary} charge density is carried by these singular D7-branes after all, which combines with the integral charges to form the complex cohomology of the transverse space.

\medskip

\noindent
{\bf Measuring D7-brane charge at $\mathbb{A}$-type singularities.}
This analysis now reveals a remarkably rich charge structure as we consider these D7-branes located on an $\mathbb{A}$-type singularity
(cf. \hyperlink{FigureD7BraneConfiguration}{\it Figure 2})
and measure their charge in the full \TED-K-theory of \cref{OneTwistedCohomologyAsTwistedEquivariantKTheory}, as then the plain cohomology group of axio-dilaton/D7-charges in \eqref{SESForD7BraneChargeOnSigmaTwo} is generalized to the twisted holomorphic cohomology groups discussed in \cref{ConformalBlocksAsTEdKTheory}:
\vspace{-2mm}
$$
\hspace{-1mm}
  \begin{tikzcd}[column sep=30pt]
    \overset{
      \mathclap{
      \raisebox{+2pt}{
        \tiny
        \color{darkblue}
        \bf
        \begin{tabular}{c}
          axio-dilaton fields
          around D7-branes
          \\
          at $\mathbb{A}$-type singularity
        \end{tabular}
      }
      }
    }{
      \mathrm{KU}
        ^{1 + [\FlatConnectionForm(\vec \weight, \ShiftedLevel)]}
       _{\flat}
      \big(
        \Sigma^2
        \times
        \HomotopyQuotient
          { \ast }
          { \CyclicGroup{\ShiftedLevel} }
      \big)
    }
    \ar[
      rr,
      "{
        \mbox{
          \tiny
          \color{greenii}
          \bf
          \TED-Chern character
        }
      }"
    ]
    &&
    \overset{
      \mathclap{
      \raisebox{2pt}{
        \tiny
        \color{orangeii}
        \bf
        \begin{tabular}{c}
          axio-dilaton fields with
          secondary D7-brane charges
          \\
          at $\mathbb{A}$-type singularity
        \end{tabular}
      }
      }
    }{
    H^1
    \Big(
      \HolomorphicDeRhamComplex{\big}
        { \Sigma^2 }
      ,\,
      \DeRhamDifferential
      +
      \FlatConnectionForm(\vec \weight, \ShiftedLevel)
    \Big)
    }
    \ar[
      rr,
      "{
        \mbox{
          \tiny
          \color{greenii}
          \bf
          \def\arraystretch{.9}
          \begin{tabular}{c}
            secondary
            \\
            \TED-Chern character
          \end{tabular}
        }
      }"
    ]
    &&
    \overset{
      \raisebox{2pt}{
        \tiny
        \color{darkblue}
        \bf
        \def\arraystretch{.9}
        \begin{tabular}{c}
          secondary charges of
          singular D7-branes
          \\
          at $\mathbb{A}$-type singularity
        \end{tabular}
      }
    }{
    \mathrm{KU}
      ^{0 + [\FlatConnectionForm(\vec \weight, \ShiftedLevel)]}
      _{\flat}
    \big(
      \Sigma^2
      \times
      \HomotopyQuotient
        { \ast }
        { \CyclicGroup{\ShiftedLevel} }
    \big).
    }
  \end{tikzcd}
$$

\noindent
By Prop. \ref{Degree1ConformalBlockInTEdKTheory}, this implies the emergence of the exotic charge structure listed on p.
\pageref{ExoticDefectBranesAndAnyonStatistics}-\pageref{ListOfExoticD7BraneChargeAspects}
(where we the direct sum over the
background orbifold B-field which gives the weight labels, according to \eqref{TheAdjunctKCocycle} and \eqref{TheMasterForm}):
\begin{equation}
  \label{ConcludingD7BraneCharges}
  \hspace{-4mm}
  \left\{
  \!\!\!\!\!\!\!\!\!
  \mbox{ \small
    \def\arraystretch{1}
    \begin{tabular}{c}
      Axio-dilaton fields with
      \\
      secondary D7-brane charges
      \\
      as seen in \TED-K-theory
    \end{tabular}
  }
  \!\!\!\!\!\!
  \right\}
  \;\;
  \simeq
  \;\;
  \underset{
    \mathclap{
    \raisebox{-2pt}{
      \tiny
      \color{darkblue}
      \bf
      \def\arraystretch{1}
      \begin{tabular}{c}
      \end{tabular}
    }
    }
  }{
    \underset{
      (\vec \weight, \ShiftedLevel)
    }{\bigoplus}
  }
  \,
  H^1
  \Big(
    \HolomorphicDeRhamComplex{\big}
      { \Sigma^2 }
    ,\,
    \DeRhamDifferential
      +
    \FlatConnectionForm
      (\vec\weight, \ShiftedLevel)
  \Big)
  \begin{tikzcd}[column sep=30pt]
    {}
    \ar[
      from=r,
      "{
        \mbox{
          \tiny
          \eqref{DegreeOneConformalBlocksInsideTEdKTheory}
        }
      }"
    ]
    &
    {}
  \end{tikzcd}
  \;\;
  \underset{
    \mathclap{
    \raisebox{-2pt}{
      \tiny
      \color{darkblue}
      \bf
      \def\arraystretch{1}
      \begin{tabular}{c}
        $\SLTwoZ$-charges
        \\
        via KR B-field
      \end{tabular}
    }
    }
  }{
    \underset{
      (\vec \weight, \ShiftedLevel)
    }{\bigoplus}
  }
  \;\;
  \underset{
    \mathclap{
    \raisebox{-3pt}{
      \tiny
      \color{darkblue}
      \bf
      anyonic structure
    }
    }
  }{
  \ConformalBlocks
    ^{1}
    _{\slTwoAffine{\ShiftedLevel-2}}
  (\vec \weight, \vec z)
  }
\;.
\end{equation}

It remains to discuss how the conformal blocks in degrees larger than 1 may analogously be identified with brane charges. For this purpose, we turn to M-theory and then invoke {\HypothesisH}:

\medskip

\noindent
{\bf M3 = M5$\perp$M5-Defect Branes in M-theory.}
In M-theory, codim=2 defect branes appear in the guise of 3-branes inside M5-brane worldvolumes
(\cite{HughesLiuPolchinski86}\cite{RocekTseytlin99}\cite{HoweLambertWest97a}, see also \cite{BayntonBurgessNierop09}), which may be understood as the loci where a transversal M5-brane intersects the ambient M5-brane
(\cite{PapadopoulosTownsend96}\cite{Tseytlin96}\cite{KachruOzYin98}\cite{GMST98}).
These 3-branes do not have an established name, but following terminology such as {\it M-strings}
(\cite{HIKLV13})
for the analogous 1-branes inside M5-brane worldvolumes,
it makes sense to refer to them as {\it M3-branes}, for short, see \hyperlink{MBraneConfigurations}{\it Figure 3}:
\vspace{-3mm}
\begin{center}
\hypertarget{MBraneConfigurations}{}
\begin{tikzpicture}

\begin{scope}[shift={(0,0)}]

  \draw (0,0) node
    {$
      \overset{
      }{
        \RealNumbers^{3,1}_{\phantom{\mathclap{A}}}
      }
    $};

  \draw (.8, 0) node
    {$
      \Sigma^2_{\phantom{\mathclap{A}}}
    $};

  \draw (1.6,+.0) node
    {
      $
        \MFTorus
      $
    };

  \draw (2.4,+.0) node
    {$
      \ComplexNumbers^{\phantom{1}}_{\phantom{A}}
    $};

  \draw (3.2,+.0) node
    {$
      \RealNumbers^1_{\phantom{A}}
    $};

  \begin{scope}[shift={(0,-.5)}]

  \draw
    (-1, 0) node
    {
      \color{orangeii}
      $\mathrm{MK6}$
    };

  \draw[
    line width=8pt,
    orangeii
  ]
    (-.4,0) to (.8+.4,0);
 \draw
   (-.4,0) to (.8+.4,0);

 \draw (2.4, -0) node
  {$
  $};

  \draw[
    line width=8pt,
    orangeii
  ]
    (3.2-.4,0) to (3.2+.4,0);
 \draw
    (3.2-.4,0) to (3.2+.4,0);

 \end{scope}

 \begin{scope}[shift={(0,-1)}]

  \draw
    (-1, 0) node
    {
      \color{greenii}
      $\mathrm{M5}$
    };

  \draw[
    line width=8pt,
    greenii
  ]
    (-.4,0) to (.8+.4,0);
  \draw
   (-.4,0) to (.8+.4,0);

  \end{scope}

 \begin{scope}[shift={(0,-1.5)}]

  \draw
    (-1, 0) node
    {
      \color{greenii}
      $\mathrm{M5}^i$
    };

  \draw[
    line width=8pt,
    greenii
  ]
    (-.4,0) to (.0+.4,0);
  \draw
   (-.4,0) to (.0+.4,0);

  \draw[
    line width=8pt,
    greenii
  ]
    (1.6-.4,0) to (1.6+.4,0);
  \draw
   (1.6-.4,0) to (1.6+.4,0);

  \end{scope}

  \begin{scope}[shift={(0,-2)}]

  \draw
    (-1,0) node
    {
      \color{gray}
      $\mathrm{M3}^i$
    };

  \draw[
    line width=8pt,
    gray
  ]
    (-.4,0) to (0+.4,0);
 \draw
   (-.4,0) to (0+.4,0);

 \draw (.8,0) node
  {$
    \{z^i\}
  $};

  \end{scope}

\end{scope}

\end{tikzpicture}

\vspace{0mm}

\hspace{0cm}
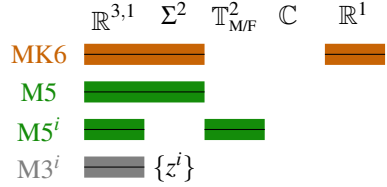
{\footnotesize {\bf Figure 3.} Defect brane configurations in M-theory. Compare the F-theoretic configurations in \hyperlink{FigureD7BraneConfiguration}{\it Figure 2} and their duality in  \hyperlink{StringDualitiesTurningD7D3BranesAtAIntoMTheory}{\it Figure 4}, \hyperlink{DualityBraneNumberOrbifoldOrder}{\it Figure 5}.}
\end{center}
\noindent

These $\mathrm{M3} \subset \mathrm{M5}$-defect branes in M-theory are thought (\cite{HoweLambertWest97b}\cite{LambertWest97}\cite{LambertWest98a}\cite{LambertWest98b}) to engineer the same SW-gauge theory as the $\mathrm{D3} \subset \mathbb{A}$-defect branes in F-theory (as per Rem. \ref{DSevenBraneChargesExpeectedInFTheory} (v) above); and they are thought to be controlled by WZW-model CFTs on their transversal complex curve, vaguely akin to Rem. \ref{DSevenBraneChargesExpeectedInFTheory} (vi) above:

\begin{remark}[{\bf AGT-correspondence and $\SpecialLinearLieAlgebra{}$-WZW theory}]
\label{AGTCorrespondenceAndWZW}
When the M5-brane worldvolume transverse to the M3-brane(s) is wrapped on a Riemann surface as
in \hyperlink{MBraneConfigurations}{\it Figure 3}, then the {\it AGT correspondence} (\cite{AGT10}, review in \cite{Tachikawa16}\cite{LeFloch20}\cite[\S 3]{Akhond21}) suggests that the M3-branes appear as punctures in this transversal space, labeled by vertex operators of a conformal field theory on $\Sigma^2$ which dually encodes at least some aspects of the (super-)Yang-Mills theory on the M3-brane worldvolume
(\cite[p. 2]{AldayTachikawa10}\cite[\S 3]{GaiottoMaldacena12}\cite[Table 1]{OSTY15}, see also \cite{CDT12}).
Specifically has been argued
that this CFT may be expressed via (cosets of) the $\SpecialLinearLieAlgebra{}$-WZW CFT
(\cite[\S 3]{Giribet10}\cite{AldayTachikawa10}\cite{NishiokaTachikawa11}, see in particular \cite[\S 1.3]{FMMW20}\cite{Manabe20}).
\end{remark}

Together, this suggests that defect branes in F-theory are equivalent (``dual'') to defect branes in M-theory:

\medskip

\noindent
{\bf M/F-Duality of defect branes.}
We describe sequences of T-dualities and M/IIA-dualities that relate defect brane configurations in F-theory (\hyperlink{FigureD7BraneConfiguration}{\it Figure 2}) with defect brane configurations in M-theory (\hyperlink{MBraneConfigurations}{\it Figure 3}), either at $\mathbb{A}_{\ShiftedLevel-1}$-singularities. This is essentially a recap of string theory folklore (e.g. \cite{Johnson97}\cite[\S 3]{Cherkis99}\cite[\S 6.3.3]{Smith02}\cite[\S 2.2]{Tachikawa14}\cite[\S 3]{DHTV14}) but may be worth spelling out some more.

\medskip

\noindent
{\it Duality rules}
for branes under T-duality (pointers in \cite{FSS16TDuality}) and M/IIA-duality (pointers in \cite[\S 1]{BMSS19}):

\medskip
\hspace{-.6cm}
{\small
\def\arraystretch{3}
\begin{tabular}{|c|l|}
\hline
  \rowcolor{lightgray}
\!\!{$\mathbb{A}_{\ShiftedLevel-1}/\ShiftedLevel\mathrm{NS5}/\ShiftedLevel\mathrm{M5}$}\!\!
&
\begin{minipage}{14.4cm}
The T-dual of an $\mathbb{A}_{\ShiftedLevel-1}$-singularity in IIA/B-theory along the $S^1$-fiber of the blowup of the transverse orbifold to an ALE-space is $\ShiftedLevel$ NS5-branes in IIB/A-theory
(\cite[\S 3]{OoguriVafa96}\cite{Kutasov96}\cite{GHM97}\cite[\S 4]{ACL98}).

In turn, the  M/IIA-dual of $\ShiftedLevel$ coincident NS5-branes in IIA-theory are $\ShiftedLevel$ coincident M5-branes in M-theory.
$\mathclap{\phantom{\vert_{\vert_{\vert}}}}$
\end{minipage}
\\
\hline
{{$\mathrm{D}3/\mathrm{D}4/\mathrm{M5}$}}
&
\begin{minipage}{14.4cm}
The T-dual of a D3-brane in IIB-theory along a transverse circle is of course a D4-brane. The further M/IIA dual is an
M5-brane whose worldvolume completes that of the M5-branes dual to NS5-branes (as in the previous item) to the Seiberg-Witten (SW) curve (\cite[\S 2.3]{Witten97Solutions}, see also \cite{FSS19SuperExceptionalGeometry}).
$\mathclap{\phantom{\vert_{\vert_{\vert}}}}$
\end{minipage}
\\
\hline
  \rowcolor{lightgray}
{{$\mathrm{D}7/\mathrm{D6}/\mathrm{MK6}$}}
&
\begin{minipage}{14.4cm}
The T-dual of D7-branes along a parallel circle are D6-branes in IIA-theory, whose M/IIA-dual are KK-monopoles in M-theory (``$\mathrm{MK5}$-branes'').
As long as the D7-branes do not coincide, neither do these KK-monopoles, so that each separately does not induce an orbi-singularity.
$\mathclap{\phantom{\vert_{\vert_{\vert}}}}$
\end{minipage}
\\
\hline
{{$\rho\mathrm{D}5/\rho\mathrm{D6}/\mathbb{A}_{\rho-1}$}}
&
\begin{minipage}{14.4cm}
The T-dual of $\rho$ D5-branes along a transverse circle are $\rho$ coincident D6-branes in IIA-theory, whose
M/IIA dual are $\rho$ coincident KK-monopoles ($\mathrm{MK6}$), whose far-horizon geometry (or ``blowdown'') is an $\mathbb{A}_{\ShiftedLevel-1}$-singularity (e.g.
\cite[around (18)]{Asano00}\cite[\S 2.2.5]{HSS18}).
$\mathclap{\phantom{\vert_{\vert_{\vert}}}}$
\end{minipage}
\\
\hline

\end{tabular}
}

\bigskip

\noindent
These duality rules predict for instance that the F-brane configuration from \hyperlink{FigureD7BraneConfiguration}{\it Figure 2} is equivalent to the following M-brane configuration shown in the following \hyperlink{StringDualitiesTurningD7D3BranesAtAIntoMTheory}{\it Figure 3}
(essentially discussed this way in, e.g., \cite[\S 2]{Tachikawa14}\cite[Table 1]{OSTY15}):

\begin{center}
\hypertarget{StringDualitiesTurningD7D3BranesAtAIntoMTheory}{}
\begin{tikzpicture}

\begin{scope}
\hspace{-.5em}
  \draw (0,.1) node
    {$
      \overset{
      }{
        \RealNumbers^{3,1}
      }
    $};

  \draw (.8,.7) node
    {$
      \overset{
      \mathllap{
        \raisebox{+0pt}{
        \hspace{-23pt}
        \rotatebox{-36}{
        \tiny
        \color{darkblue}
        \bf
        \def\arraystretch{.9}
        \begin{tabular}{l}
          transverse
          \\
          complex curve
        \end{tabular}
        \hspace{-6pt}
        }
      }
      }
      }
      {
        \Sigma^2
      }
    $};

  \draw (2,.1) node
    {
      $
        \HomotopyQuotient
          { \Quaternions }
          {\CyclicGroup{\ShiftedLevel}}
        \mathrlap{
          \hspace{-17pt}
          \raisebox{6pt}{
          \rotatebox{36}{$
            \mathrlap{
            \mbox{
            \tiny
            \color{darkblue}
            \bf
            \def\arraystretch{.9}
            \begin{tabular}{c}
              $\mathbb{A}$-type
              \\
              singularity
            \end{tabular}
            }
            }
          $}
          }
        }
      $
    };

  \draw (2,+1.4) node
    {
      $
        \overset{
          \mathclap{
          \raisebox{3pt}{
            \tiny
            \color{darkblue}
            \bf
            ALE space
          }
          }
        }{
          S^1_{\mathrm{f}}
          \,\times\,
          \RealNumbers^3
        }
      $
    };

  \draw[->]
    (2, 1)
    to
    node[left]{ \mbox{\tiny\color{greenii} \bf blowup} }
    (2, .4);

 \begin{scope}[shift={(0,-.5)}]

  \draw
    (-1, -.05) node
    {
      \color{gray}
      $\mathrm{D7}_{{}_I}$
    };

  \draw[
    line width=8pt,
    gray
  ]
    (0-.4,0) to (0+.4,0);
 \draw
   (-.4,0) to (0+.4,0);

 \draw (.8,0) node
   {$
     \{z_{{}_I}\}
   $};

  \draw[
    line width=8pt,
    gray
  ]
    (1.6-.4,0) to (2.4+.4,0);
 \draw
   (1.6-.4,0) to (2.4+.4,0);

 \end{scope}

 \begin{scope}[shift={(0,-1)}]

  \draw
    (-1, 0) node
    {
      \color{greenii}
      $\mathbb{A}_{\ShiftedLevel-1}$
    };

  \draw[
    line width=8pt,
    greenii
  ]
    (-.4,0) to (.8+.4,0);
 \draw
   (-.4,0) to (.8+.4,0);

 \draw (2, -0) node
  {$
    \{0\}
  $};

 \end{scope}

 \begin{scope}[shift={(0,-1.5)}]
  \draw
    (-1, 0) node
    {
      \color{gray}
      $\mathrm{D3}^i$
    };

  \draw[
    line width=8pt,
    gray
  ]
    (-.4,0) to (0+.4,0);
 \draw
   (-.4,0) to (0+.4,0);

 \draw (.8,0) node
   {$
     \{z^i\}
   $};

 \draw (2,0) node
   {$
     \{0\}
   $};

 \end{scope}

\end{scope}

\draw (3.8,0) node
  {
    $\xleftrightarrow
      [\scalebox{.7}{along $S^1_{\mathrm{f}}$}]
      {\scalebox{.7}{T-duality}}$
  };

\begin{scope}[shift={(6,0)}]
\hspace{.5em}
  \draw (0,.1) node
    {$
      \overset{
      }{
        \RealNumbers^{3,1}
      }
    $};

  \draw (.8,.05) node
    {
      $\Sigma^2$
    };

  \draw (1.6,.55) node
    {$
      \overset{
      \mathllap{
        \hspace{-32pt}
        \rotatebox{-36}{
        \tiny
        \color{darkblue}
        \bf
        \def\arraystretch{.9}
        \begin{tabular}{l}
          T-dual circle
        \end{tabular}
        \hspace{-6pt}
        }
      }
      }
      {
        S^1_{\mathrm{f}}
      }
    $};

  \draw (2.4,.04) node
    {
      $\RealNumbers^3$
    };

  \begin{scope}[shift={(0,-.5)}]

  \draw
    (-1, -.05) node
    {
      \color{gray}
      $\mathrm{D6}_{{}_I}$
    };

  \draw[
    line width=8pt,
    gray
  ]
    (-.4,0) to (0+.4,0);
 \draw
   (-.4,0) to (0+.4,0);

 \draw (.8,0) node
   {$
     \{z_{{}_I}\}
   $};

 \draw (1.6,0) node
   {$
     \{x\}
   $};

  \draw[
    line width=8pt,
    gray
  ]
    (2.4-.4,0) to (2.4+.4,0);
 \draw
   (2.4-.4,0) to (2.4+.4,0);

 \end{scope}

  \begin{scope}[shift={(0,-1)}]

  \draw
    (-1, 0) node
    {
      \color{greenii}
      $\ShiftedLevel \!\cdot\! \mathrm{NS5}$
    };

  \draw[
    line width=8pt,
    greenii
  ]
    (-.4,0) to (.8+.4,0);
 \draw
   (-.4,0) to (.8+.4,0);

 \draw (1.6, 0) node
  {$
    \{x\}
  $};

 \draw (2.4, 0) node
  {$
    \{0\}
  $};

 \end{scope}

 \begin{scope}[shift={(0,-1.5)}]

  \draw
    (-1, 0) node
    {
      \color{gray}
      $\mathrm{D4}^i$
    };

  \draw[
    line width=8pt,
    gray
  ]
    (-.4,0) to (0+.4,0);
 \draw
   (-.4,0) to (0+.4,0);

 \draw (.8,0) node
   {$
     \{z^i\}
   $};

 \draw (2.4,0) node
   {$
     \{0\}
   $};

  \draw[
    line width=8pt,
    gray
  ]
    (1.6-.4,0) to (1.6+.4,0);
 \draw
   (1.6-.4,0) to (1.6+.4,0);

 \end{scope}

\end{scope}

\hspace{.7em}
\draw (9.6,0) node
  {
    $\xleftrightarrow[\mbox{\small duality}]{\mbox{\small M/IIA}}$
  };

\hspace{.5em}
\begin{scope}[shift={(11.8,0)}]

  \draw (0,.1) node
    {$
      \overset{
      }{
        \RealNumbers^{3,1}
      }
    $};

  \draw (.8,.05) node
    {$
      {
        \Sigma^2
      }
    $};

  \draw (1.6,.0) node
    {
      $S^1_{\mathrm{f}}$
    };

  \draw (2.4,.0) node
    {
      $S^1_{\mathrm{m}}$
    };

  \draw (2,.35) node
    {
      $
        \overset{
          \mathclap{
          \raisebox{-1pt}{
            \tiny
            \color{darkblue}
            M-F torus $\MFTorus$
          }
          }
        }{
        \overbrace{
          \phantom{-----}
        }
        }
      $
    };

  \draw (1.7,.8) node
    {
      $
        \overset{
          \mathclap{
          \raisebox{-1pt}{
            \tiny
            \color{darkblue}
            elliptic fibration
          }
          }
        }{
        \overbrace{
          \phantom{-------}
        }
        }
      $
    };

  \draw (3.2,.04) node
    {
      $\RealNumbers^3$
    };

 \begin{scope}[shift={(0,-.5)}]

  \draw
    (-1, -.05) node
    {
      \color{gray}
      $\mathrm{MK6}_{{}_I}$
    };

  \draw[
    line width=8pt,
    gray
  ]
    (-.4,0) to (0+.4,0);
 \draw
   (-.4,0) to (0+.4,0);

 \draw (.8,0) node
   {$
     \{z_{{}_I}\}
   $};

 \draw (2, 0) node
  {$
    \{x \;+\; \ImaginaryUnit y\}
  $};

  \draw[
    line width=8pt,
    gray
  ]
    (3.2-.4,0) to (3.2+.4,0);
 \draw
   (3.2-.4,0) to (3.2+.4,0);

 \end{scope}

 \begin{scope}[shift={(0,-1)}]

  \draw
    (-1, 0) node
    {
      \color{greenii}
      $\ShiftedLevel \!\cdot\! \mathrm{M5}$
    };

  \draw[
    line width=8pt,
    greenii
  ]
    (-.4,0) to (.8+.4,0);
 \draw
   (-.4,0) to (.8+.4,0);

 \draw (2,0) node
   {$
     \{x \;+\; \ImaginaryUnit y\}
   $};

 \draw (3.2, 0) node
  {$
    \{0\}
  $};

 \end{scope}

 \begin{scope}[shift={(0,-1.5)}]

  \draw
    (-1, 0) node
    {
      \color{gray}
      $\mathrm{M5}^i$
    };

  \draw[
    line width=8pt,
    gray
  ]
    (-.4,0) to (0+.4,0);
 \draw
   (-.4,0) to (0+.4,0);

 \draw (.8,0) node
   {$
     \{z^i\}
   $};

  \draw[
    line width=8pt,
    gray
  ]
    (1.6-.4,0) to (2.4+.4,0);
 \draw
   (1.6-.4,0) to (2.4+.4,0);

 \draw (3.2,0) node
   {$
     \{0\}
   $};

 \draw[opacity=.6]
   (.8-.3, -.75+1.5) rectangle (2.4+.5, -1.78+1.5);

 \draw
   (1.8,-1.9+1.5) node
   {
     \tiny
     \color{darkblue}
     \bf SW-curve
   };

 \end{scope}

\end{scope}

\end{tikzpicture}
{
  \footnotesize
  {\bf Figure 4.}
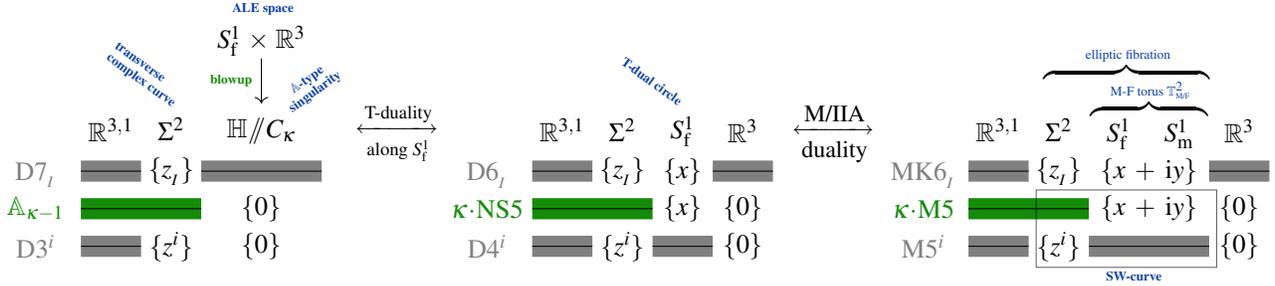
  A sequence of stringy dualities
  turning the F-theoretic configuration from
  \hyperlink{FigureD7BraneConfiguration}{\it Figure 2}
  into an M-theoretic configuration
  as in \hyperlink{MBraneConfigurations}{\it Figure 3}.
}

\end{center}

\medskip

\noindent
{\bf Fiber/Base duality and brane-number/orbifold-order duality.}
Notice that \cite[\S 2.2.1]{Tachikawa14} considers the worldvolume theory of what above is called $\mathrm{M5}^i$, while \cite[p. 23]{OSTY15} also consider the worldvolume theory of the $\mathrm{M5}$ above. The latter wraps the base of the elliptic fibration, while the former wraps its fiber: \footnote{The elliptic fibration \eqref{MFEllipticFibration} is a trivial fiber product only topologically, while the complex structure on the fibers depends non-trivially on the base via the axio-dilaton field $\tau$ \eqref{TheAxioDilaton}.}
\begin{equation}
  \label{MFEllipticFibration}
  \begin{tikzcd}[row sep=small]
    \MFTorus
    \ar[r]
    &
    \Sigma^2 \times \MFTorus
    \ar[d]
    \\
    &
    \Sigma^2
  \end{tikzcd}
\end{equation}
Accordingly, switching perspective between regarding the M3-brane defects from within one or the other intersecting ambient M5-brane species corresponds
(e.g. \cite[p. 5]{HaghighatSun18})
to what is known in F-theory as {\it fiber-base duality} (e.g. \cite{HKYY18}). In the quiver gauge theories on the compactified M5-branes, this is a duality interchanging rank and multiplicity of the gauge group $\mathrm{SU}(\rho)^{\kappa-1}$
(e.g \cite{BPTY11}); while in terms of the $\SpecialLinearLieAlgebra{}$-WZW theory on the (either) transverse complex curve, this is
(see \cite[(1.1) with (1.4)]{FMMW20} and \cite[p. 2, 11]{Manabe20})
the {\it level-rank duality} (e.g. \cite{NakanishiTsuchiya}). Explicitly, the above duality rules yield the brane-number/orbifold-order duality $\kappa \leftrightarrow \rho$ as shown in the following \hyperlink{DualityBraneNumberOrbifoldOrder}{\it Figure 5}:\footnote{Including the D7-branes from
\hyperlink{FigureD7BraneConfiguration}{\it Figure 2}
on the left of
\hyperlink{DualityBraneNumberOrbifoldOrder}{\it Figure 5} shows that under $\mathrm{T}_{\mathrm{b}}$ they turn into $\mathrm{D}8$-branes, whose IIA/M-dual is subtle (cf. pointers and discussion in \cite{BMSS19}) beyond the scope of the present discussion. Therefore we disregard the D7-branes at this point, focusing on the remaining data constituted by the D3-branes of codimension=2 inside the $\mathbb{A}_{\ShiftedLevel-1}$-singularity.}

\begin{center}
\hypertarget{DualityBraneNumberOrbifoldOrder}{}
\begin{tikzpicture}

\draw[<->]
  (8.3, 0)
  node[above]
  { \hspace{19pt}\scalebox{.7}{M/IIA} }
  to (8.3+.7,0);

\draw[<->]
  (8.3, -3)
  node[above]
  { \hspace{19pt}\scalebox{.7}{M/IIA} }
  to (8.3+.7,-3);

\draw[<->]
  (2.5, -1.2)
  to (3.2,-.2);
\draw
({2.5 + (3.2-2.5)/4-.35}, {-1.2 + (-.2+1.2)/4 +.35})
node
 {
   \hspace{19pt}\scalebox{.7}{$\mathrm{T}_{\mathrm{f}}$}
 };

\draw[<->]
  (2.5, -1.9)
  to (3.2,-3);
\draw
({2.5+(3.2-2.5)/4-.3}, {-1.9+(-3+1.9)/4-.4})
 node
 {
   \hspace{19pt}\scalebox{.7}{$\mathrm{T}_{\mathrm{b}}$}
 };

\draw (0,-1.5) node {
\begin{tikzpicture}[framed, xscale=.8]

  \begin{scope}

  \draw (0,.1) node
    {$
      \overset{
      }{
        \RealNumbers^{3,1}
      }
    $};

  \draw (.8,.1) node
    {$
      \RealNumbers^1
    $};

  \draw (1.6,.1) node
    {$
      S^1_{\mathrm{b}}
    $};

  \draw (2.4,.1) node
    {$
      S^1_{\mathrm{f}}
    $};

  \draw (3.2,.1) node
    {$
      \RealNumbers^3
    $};

\end{scope}

 \begin{scope}[shift={(0,-.5)}]

  \draw (-1,0) node
    {$
      \rho \mathrm{D5}
    $};

  \draw[line width=8pt, gray]
    (0-.4,0) to (1.6+.4,0);

 \end{scope}

 \begin{scope}[shift={(0,-1)}]

  \draw (-1,0) node
    {$
      \mathbb{A}_{\ShiftedLevel-1}
    $};

  \draw[line width=8pt, gray]
    (0-.4,0) to (1.6+.4,0);

 \end{scope}

 \begin{scope}[shift={(0,-1.5)}]

  \draw (-1,0) node
    {$
      \mathrm{D3}
    $};

  \draw[line width=8pt, gray]
    (0-.4,0) to (0+.4,0);

 \end{scope}

\end{tikzpicture}
};

\draw (5.7,0) node {
\begin{tikzpicture}[framed, xscale=.8]

  \begin{scope}

  \draw (0,.1) node
    {$
      \overset{
      }{
        \RealNumbers^{3,1}
      }
    $};

  \draw (.8,.1) node
    {$
      \RealNumbers^1
    $};

  \draw (1.6,.1) node
    {$
      S^1_{\mathrm{b}}
    $};

  \draw (2.4,.1) node
    {$
      S^1_{\mathrm{f}}
    $};

  \draw (3.2,.1) node
    {$
      \RealNumbers^3
    $};

  \end{scope}

 \begin{scope}[shift={(0,-.5)}]

  \draw (-1,0) node
    {$
      \rho \mathrm{D6}
    $};

  \draw[line width=8pt, gray]
    (0-.4,0) to (2.4+.4,0);

 \end{scope}

 \begin{scope}[shift={(0,-1)}]

  \draw (-1,0) node
    {$
      \ShiftedLevel \mathrm{NS5}
    $};

  \draw[line width=8pt, gray]
    (0-.4,0) to (1.6+.4,0);

 \end{scope}

 \begin{scope}[shift={(0,-1.5)}]

  \draw (-1,0) node
    {$
      \mathrm{D4}
    $};

  \draw[line width=8pt, gray]
    (0-.4,0) to (0+.4,0);

  \draw[line width=8pt, gray]
    (2.4-.4,0) to (2.4+.4,0);

 \end{scope}

\end{tikzpicture}
};

\draw (12,0) node {
\begin{tikzpicture}[framed, xscale=.8]
  \begin{scope}

  \draw (0,.1) node
    {$
      \overset{
      }{
        \RealNumbers^{3,1}
      }
    $};

  \draw (.8,.1) node
    {$
      \RealNumbers^1
    $};

  \draw (1.6,.1) node
    {$
      S^1_{\mathrm{b}}
    $};

  \draw (2.4,.1) node
    {$
      S^1_{\mathrm{f}}
    $};

  \draw (3.2,.1) node
    {$
      S^1_{\mathrm{m}}
    $};

  \draw (4,.1) node
    {$
      \RealNumbers^3
    $};

  \end{scope}

 \begin{scope}[shift={(0,-.5)}]

  \draw (-1,0) node
    {$
      \color{greenii}
      \mathbb{A}_{\rho-1}
    $};

  \draw[line width=8pt, greenii]
    (0-.4,0) to (2.4+.4,0);
  \draw
    (0-.4,0) to (2.4+.4,0);

 \end{scope}

 \begin{scope}[shift={(0,-1)}]

  \draw (-1,0) node
    {$
      \color{orangeii}
      \ShiftedLevel \mathrm{M5}
    $};

  \draw[line width=8pt, orangeii]
    (0-.4,0) to (1.6+.4,0);
  \draw
    (0-.4,0) to (1.6+.4,0);

 \end{scope}

 \begin{scope}[shift={(0,-1.5)}]

  \draw (-1,0) node
    {$
      \mathrm{M5}
    $};

  \draw[line width=8pt, gray]
    (0-.4,0) to (0+.4,0);

  \draw[line width=8pt, gray]
    (2.4-.4,0) to (3.2+.4,0);

 \end{scope}

\end{tikzpicture}
};

\draw (5.7,-3) node {
\begin{tikzpicture}[framed, xscale=.8]
  \begin{scope}

  \draw (0,.1) node
    {$
      \overset{
      }{
        \RealNumbers^{3,1}
      }
    $};

  \draw (.8,.1) node
    {$
      \RealNumbers^1
    $};

  \draw (1.6,.1) node
    {$
      S^1_{\mathrm{b}}
    $};

  \draw (2.4,.1) node
    {$
      S^1_{\mathrm{f}}
    $};

  \draw (3.2,.1) node
    {$
      \RealNumbers^3
    $};

  \end{scope}

 \begin{scope}[shift={(0,-.5)}]

  \draw (-1,0) node
    {$
      \rho \mathrm{D4}
    $};

  \draw[line width=8pt, gray]
    (0-.4,0) to (.8+.4,0);

 \end{scope}

 \begin{scope}[shift={(0,-1)}]

  \draw (-1,0) node
    {$
      \mathbb{A}_{\ShiftedLevel-1}
    $};

  \draw[line width=8pt, gray]
    (0-.4,0) to (1.6+.4,0);

 \end{scope}

 \begin{scope}[shift={(0,-1.5)}]

  \draw (-1,0) node
    {$
      \mathrm{D4}
    $};

  \draw[line width=8pt, gray]
    (0-.4,0) to (0+.4,0);

  \draw[line width=8pt, gray]
    (1.6-.4,0) to (1.6+.4,0);

 \end{scope}

\end{tikzpicture}
};

\draw (12,-3) node {
\begin{tikzpicture}[framed, xscale=.8]
  \begin{scope}

  \draw (0,.1) node
    {$
      \overset{
      }{
        \RealNumbers^{3,1}
      }
    $};

  \draw (.8,.1) node
    {$
      \RealNumbers^1
    $};

  \draw (1.6,.1) node
    {$
      S^1_{\mathrm{b}}
    $};

  \draw (2.4,.1) node
    {$
      S^1_{\mathrm{f}}
    $};

  \draw (3.2,.1) node
    {$
      S^1_{\mathrm{m}}
    $};

  \draw (4,.1) node
    {$
      \RealNumbers^3
    $};

  \end{scope}

 \begin{scope}[shift={(0,-.5)}]

  \draw (-1,0) node
    {$
      \color{greenii}
      \rho \mathrm{M5}
    $};

  \draw[line width=8pt, greenii]
    (0-.4,0) to (.8+.4,0);
  \draw
    (0-.4,0) to (.8+.4,0);

  \draw[line width=8pt, greenii]
    (3.2-.4,0) to (3.2+.4,0);
  \draw
    (3.2-.4,0) to (3.2+.4,0);

 \end{scope}

 \begin{scope}[shift={(0,-1)}]

  \draw (-1,0) node
    {$
      \color{orangeii}
      \mathbb{A}_{\ShiftedLevel-1}
    $};

  \draw[line width=8pt, orangeii]
    (0-.4,0) to (1.6+.4,0);
  \draw
    (0-.4,0) to (1.6+.4,0);

  \draw[line width=8pt, orangeii]
    (3.2-.4,0) to (3.2+.4,0);
  \draw
    (3.2-.4,0) to (3.2+.4,0);

 \end{scope}

 \begin{scope}[shift={(0,-1.5)}]

  \draw (-1,0) node
    {$
      \mathrm{M5}
    $};

  \draw[line width=8pt, gray]
    (0-.4,0) to (0+.4,0);

  \draw[line width=8pt, gray]
    (1.6-.4,0) to (1.6+.4,0);
  \draw[line width=8pt, gray]
    (3.2-.4,0) to (3.2+.4,0);

 \end{scope}
\end{tikzpicture}
};

\end{tikzpicture}
  {
    \footnotesize
    {\bf Figure 5.}
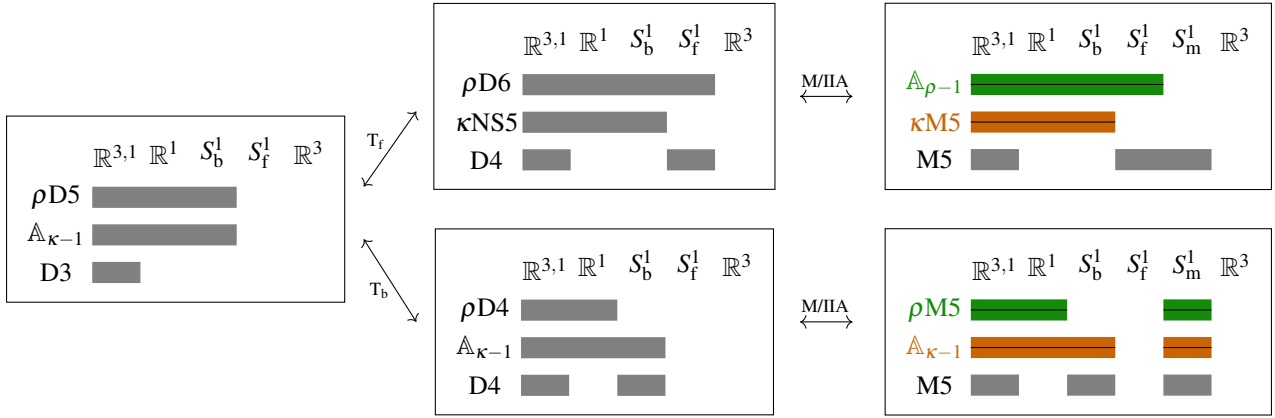
    Sequence of stringy dualities exhibiting
    duality between the number of
    M5-branes and order of
    the orbifold singularity that they probe.
  }
\end{center}

The configuration on the bottom right
of \hyperlink{DualityBraneNumberOrbifoldOrder}{\it Figure 5}
is the fiber-base-dual configuration which we are after, where the $\ShiftedLevel$-M5-branes
from \hyperlink{StringDualitiesTurningD7D3BranesAtAIntoMTheory}{\it Figure 4}
have turned into a $\CyclicGroup{\ShiftedLevel}$-orbi-singularity, and we see M5-branes intersecting over this singular locus in an M3-brane \eqref{OrderedConfigurationSpaceAsSpaceOfIntersections}.
It is this configuration on which we now measure quantum brane charges according to {\HypothesisH}.

Before we turn to this in Rem. \ref{MThreeBraneModuliViaHypothesisH} below, here is a loose end worth mentioning:

\begin{remark}[\bf Level and rank in the $\SpecialLinearLieAlgebra{}$-WZW models appearing in the AGT correspondence]
  \label{LevelAndRankInTheAGTCorrespondence}
  In \cite{NishiokaTachikawa11}\cite{FMMW20} it is the
  rank and hence dually the
  level $\Level$
  of the AGT $\SpecialLinearLieAlgebra{}$-WZW theory (from Rem. \ref{AGTCorrespondenceAndWZW})
  which is identified with the order of the orbifold singularity probed by the M5-branes, while the results in \cref{ConformalBlocksAsTEdKTheory} instead identify this with the {\it shifted} level $\ShiftedLevel = \Level + 2$.

\noindent {\bf (i)}  While we currently do not see how to resolve this mismatch, it may be worthwhile to highlight the general subtlety with quantum corrections to the level.
  For example, \cite[p. 3]{NishiokaTachikawa11} quotes the argument \cite{DHSV08}, whose analysis identifies brane number with a {\it classical} Chern-Simons level (in \cite[p. 15-16]{DHSV08}). The latter is well-known to pick up quantum corrections by exactly dual Coxeter number
  \eqref{ShiftedLevel}
  that shifts $\Level$ to $\ShiftedLevel$
  (e.g. \cite{AGLR90}\cite{Shifman91}),
  but this correction seems not to have been considered in these arguments.
  The analogous issue arises in discussion the ABJM model
  for M2-branes at $\CyclicGroup{\ShiftedLevel}$-singularities
  (\cite{ABJM08}), where the quantum correction to the level is
  also generally disregarded -- but here an exception is \cite[(3.72)]{Marino11}.

\noindent {\bf (ii)}   Indeed, the careful computation
  in \cite[p. 10-12]{OoguriVafa96}\cite[p. 4]{LLS00}
  shows that such a
  shift
  {\it is}
  present in the type II-theory discussion \eqref{ShiftOfLevelViaStringsOnALE} and thus would be expected to translate under type-II/M-duality.

\end{remark}

\medskip

\noindent
{\bf Measuring M3-Brane charges at $\mathbb{A}$-type singularities.}
The above discussion of dualities means to have shown that the transverse orbi-geometries of codim=2 defect branes in F-theory and in M-theory correspond to each other under the expected stringy dualities. Therefore we conclude now with discussing defect M3-brane charges in M-theory according to {\HypothesisH}, amplifying that this makes (the respective dual of) the transverse space $\Sigma^2$ be generalized to its configuration spaces of points, as required to apply the full strength of Thm. \ref{ConformalBlockInTEdKTheory} above:

\medskip
\noindent
\begin{remark}
[{\bf M3-brane moduli via {\HypothesisH}.}]
\label{MThreeBraneModuliViaHypothesisH}
According to {\HypothesisH}, moduli spaces of M-branes
in the 11-dimensional M-theory bulk are cocycle space for \TED-4-Cohomotopy theory,
where the degree 4 corresponds to that of the flux density $G_4$ of the M-theory C-field \cite{SS19ConfigurationSpaces} (cf. \cite{Sati19} from a TD-generalized cohomology perspective).
Moreover, after localization to the 7-dimensional locus of an
MK6-brane (an ADE-singularity in 11d), this situation repeats, but now with
respect to 3-Cohomotopy, where the degree 3 corresponds to that of the flux
density $H_3$ \cite{FSS20TwistedString}\cite{FSS19HopfWZ} (see Rem. \ref{MoreOnHypothesisHOnM5Branes} for more on how this works).

\medskip
Specifically, the moduli space of a flat non-singular codimension=$d-p$ brane
with respect to plain $n$-Cohomotopy
is the pointed mapping space
$\ConstrainedMaps{\big}{\RealNumbers^{d-p}_{\mathrm{cpt}}}{S^n}$:

\vspace{-1mm}
\begin{itemize}[leftmargin=*]
\setlength\itemsep{-1pt}

\item In sufficiently large codimension $d - p \geq n$,
this is a homotopy type controlled by the homotopy groups of spheres \cite{SS21MF}.

\item On the other hand, in small codimension
$d-p < n$ of interest here, the homotopy type of this space is presented by the configuration space of un-ordered points
in $\RealNumbers^{n}$ which are distinct already in their projection to $\RealNumbers^{d-p}$
and which may escape to infinity along the remaining $n-d+p$ directions (as follows by a classical theorem due to P. May and G. Segal, recalled in \cite[Prop. 2.5]{SS19ConfigurationSpaces}):
$$
  \overset{
    \mathclap{
    \raisebox{2pt}{
      \tiny
      \color{darkblue}
      \bf
      \begin{tabular}{c}
        Cohomotopy cocycle space
      \end{tabular}
    }
    }
  }{
  \ConstrainedMaps{}
    { \RealNumbers^{d-p}_{\compact} }
    { S^n}
  }
  \;\;
  \underset{
    \homotopy
  }{
  \simeq
  }
  \;\;
  \coprod_{\NumberOfProbeBranes \in \NaturalNumbers}
  \overset{
    \mathclap{
    \raisebox{3pt}{
      \tiny
      \color{darkblue}
      \bf
      \def\arraystretch{.9}
      \begin{tabular}{c}
        configuration space of points
        in $\RealNumbers^n$
        \\
        which are distinct already in $\RealNumbers^{d-p}$ and
        \\
        may escape to $\infty$ along $\RealNumbers^{n-p-d}$
      \end{tabular}
    }
    }
  }{
  \mathrm{Conf}_{\NumberOfProbeBranes}
  \big(
    \RealNumbers^{d-p}
    ;\,
    \RealNumbers^{n-d+p}_{\compact}
  \big)
  }
  \,.
$$
\end{itemize}

\vspace{-1mm}
\noindent Accordingly, the corresponding moduli space for {\it intersections} of
branes in codimension $d-p$ and $d-p'$ is to be the fiber product of these
respective configuration spaces \cite[(9)]{SS19ConfigurationSpaces}. For example, when $d - p' = 1$, then configurations of codimension=1 branes are specified, up to homotopy of configurations, by the {\it order} in which they are arranged in their transverse space, and therefore this fiber product moduli space of flat $p \!\!\perp\!\! p'$-brane intersections for $d-p = n-1$ is presented by the configuration space of {\it ordered} points in $\RealNumbers^{d-p}$ \cite[Prop. 2.4]{SS19ConfigurationSpaces}.

\medskip
In \cite[Prop. 2.11]{SS19ConfigurationSpaces} this situation is considered for $n = 4$, applicable to the M-theory bulk; while here we are concerned with the case $n = 3$, $d = 7$, $p = 5$, $p' = 1$, corresponding to charges of intersecting probe M5-branes inside M5-domain walls inside MK6 worldvolumes, as expected:

\vspace{-6mm}
\hypertarget{FigureConfigurationSpaceForM3Branes}{}
\begin{equation}
\label{OrderedConfigurationSpaceAsSpaceOfIntersections}
\hspace{-1cm}
\def\arraystretch{7}
\begin{array}{rcl}
\overset{
  \raisebox{5pt}{
    \tiny
    \color{darkblue}
    \bf
    \def\arraystretch{.9}
    \begin{tabular}{c}
      Configuration space of
      \\
      ordered points
      in the plane
      \\
      \phantom{a}
    \end{tabular}
  }
}{
{\coprod}_{\NumberOfProbeBranes}
\ConfigurationSpace{\NumberOfProbeBranes}(\ComplexNumbers)
}
&
\simeq
&
\underset{
  \mathclap{
  \raisebox{-4pt}{
    \tiny
    \color{darkblue}
    \bf
    \def\arraystretch{.9}
    \begin{tabular}{c}
    Fiber product of respective configuration spaces
    \\
    (of un-ordered points escaping to transverse infinity)
    \\
    reflecting the brane intersections
    \end{tabular}
  }
  }
}{
\overset{
\overset{
  \mathclap{
  \raisebox{3pt}{
    \tiny
    \color{darkblue}
    \bf
    \def\arraystretch{.9}
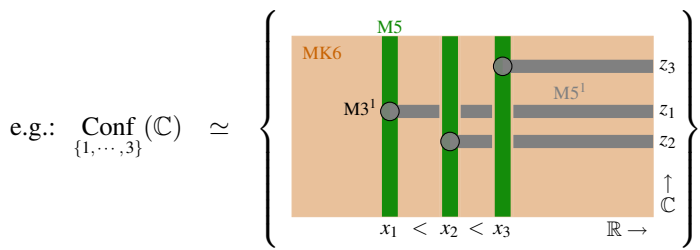
    \begin{tabular}{c}
      3-Cohomotopy cocycle space
      \\
      for codim=1 branes
    \end{tabular}
  }
  }
}{
\ConstrainedMaps{}
  {
    \RealNumbers_{+}
    \wedge
    \ComplexNumbers_{\compact}
  }
  {
    S^3
  }
}
  \;
  \simeq
}{
  \overbrace{
  {\coprod}_{\NumberOfProbeBranes}
  \mathrm{Conf}_{\NumberOfProbeBranes}
  \big(
    \ComplexNumbers
    ;\,
    \RealNumbers_{\compact}
  \big)
  }
}
\qquad \qquad
\underset{
  \mathclap{ \small
    \raisebox{-6pt}{$
    \coprod_{\NumberOfProbeBranes}
    \mathrm{Conf}_{\NumberOfProbeBranes}
    \big(
      \ast;\,
      (\RealNumbers
        \times
      \ComplexNumbers)_{\compact}
    \big)
    $}
  }
}{
  \;\;
  \bigtimes
  \;\;
}
\qquad \qquad
\overset{
\overset{
  \mathclap{
  \raisebox{3pt}{
    \tiny
    \color{darkblue}
    \bf
    \def\arraystretch{.9}
    \begin{tabular}{c}
      3-Cohomotopy cocycle space
      \\
      for codim-2 branes
    \end{tabular}
  }
  }
}{
\ConstrainedMaps{}
  {
    \RealNumbers_{\compact}
    \wedge
    \ComplexNumbers_{+}
  }
  {
    S^3
  }
}
  \;\simeq
}
{
  \overbrace{
  {\coprod}_{\NumberOfProbeBranes}
  \mathrm{Conf}_{\NumberOfProbeBranes}
  \big(
    \RealNumbers
    ;\,
    \ComplexNumbers_{\compact}
  \big)
  }
}
}
\\
\mathllap{
  \mbox{e.g.:}
  \;\;\;
}
\ConfigurationSpace{3}(\ComplexNumbers)
&
\simeq
&
\left\{
\raisebox{9pt}{
\begin{tikzcd}
  \draw[
    draw=orangeii,
    fill=orangeii,
    draw opacity=.4,
    fill opacity=.4
  ]
    (-1.8,1) rectangle (3,-1.4);

  \draw (3, -1.66)
    node
    {
      \llap{
      \scalebox{.8}{
      \color{black}
      $\RealNumbers \to $
      }}
    };

  \draw (3.2, -1.2)
    node
    {
      {
      \scalebox{.8}{
      \color{black}
      $
        \def\arraystretch{.9}
        \begin{array}{c}
          \uparrow
          \\
          \ComplexNumbers
        \end{array}
      $
      }}
    };

  \draw (-1.4, .7)
    node
    {
      \scalebox{.7}{
      \color{orangeii}
      MK6
      }
    };

  \draw[
    line width=6pt,
    greenii
  ]
    (-.5, 1) to (-.5,-1.4);
  \draw
    (-.5,-1.63) node
    {
      \scalebox{.8}{$x_1$}
    };

  \draw (-.5, +1.05)
    node
    {
      \scalebox{.7}{
      \color{greenii}
      M5
      }
    };

  \draw[
    line width=5pt,
    gray
  ]
    (-.5, 0) to (3, 0);
  \draw[
    fill=gray,
    fill opacity=.9
  ]
    (-.5, 0) circle (3.5pt);

  \draw
    (3.2, -.04) node
    {
      \scalebox{.8}{
        $z_1$
      }
    };
  \draw
    (1.9, -.04+.2) node
    {
      \scalebox{.7}{
        \color{gray}
        $\mathrm{M5}^1$
      }
    };

  \draw (-.89, -.05)
    node
    {
      \scalebox{.7}{
      \color{black}
      $\mathrm{M3}^1$
      }
    };

  \draw[
    line width=8pt,
    white,
  ]
    (+.3, 1) to (+.3,-1.4);
  \draw[
    line width=8pt,
    orangeii,
    draw opacity=.4
  ]
    (+.3, 1) to (+.3,-1.4);
  \draw[
    line width=6pt,
    greenii
  ]
    (+.3, 1) to (+.3,-1.4);
  \draw
    (+.3,-1.63) node
    {
      \scalebox{.8}{$x_2$}
    };

  \draw
    (-.1,-1.63) node
    {
      \scalebox{.8}{$<$}
    };

  \draw[
    line width=5pt,
    gray
  ]
    (+.3, -.4) to (+3, -.4);
  \draw[
    fill=gray,
    fill opacity=.9
  ]
    (+.3, -.4) circle (3.5pt);
  \draw
    (3.2, -.44) node
    {
      \scalebox{.8}{
        $z_2$
      }
    };

  \draw[
    line width=8pt,
    white,
  ]
    (+1, 1) to (+1,-1.4);
  \draw[
    line width=8pt,
    orangeii,
    draw opacity=.4
  ]
    (+1, 1) to (+1,-1.4);
  \draw[line width=6pt, greenii]
    (+1, 1) to (+1,-1.4);

  \draw
    (+1,-1.63) node
    {
      \scalebox{.8}{$x_3$}
    };

  \draw
    (+.65,-1.63) node
    {
      \scalebox{.8}{$<$}
    };

  \draw[
    line width=5pt,
    gray
  ]
    (+1, +.6) to (+3, +.6);
  \draw[
    fill=gray,
    fill opacity=.9
  ]
    (+1, +.6) circle (3.5pt);
  \draw
    (3.2, +.56) node
    {
      \scalebox{.8}{
        $z_3$
      }
    };

  \draw
    (1.9, +.46+.2) node
    {
      \scalebox{.7}{
        \color{gray}
      }
    };

\end{tikzcd}
}
\hspace{-12pt}
\right\}
\end{array}
\end{equation}
\begin{center}
  \footnotesize
  {\bf Figure 6.} The moduli space of flat M3-branes according to {\HypothesisH} is the configuration space of ordered points in their transverse plane.
\end{center}

Therefore, taking $\NumberOfPunctures$ of the M3 as background branes and considering the moduli spaces of configurations of $\NumberOfProbeBranes$ further M3 branes around these yields exactly the configuration spaces of points in the punctured plane \eqref{ConfigurationSpaceAsFiberProduct}.
Now Thm. \ref{ConformalBlockInTEdKTheory} implies that
the corresponding quantum states, conceptualized as in \cite{SS19ConfigurationSpaces}\cite{CSS21} but now measured in full \TED-K-theory, reflects the $\slTwoAffine{\Level}$-conformal blocks in any degree $\NumberOfProbeBranes$.
\end{remark}

\begin{remark}[Charge quantization on M5-branes at $\mathbb{A}$-type singularities according to {\HypothesisH}]
\label{MoreOnHypothesisHOnM5Branes}
We recall and spell out in more detail how
{\HypothesisH} implies that the charges of
M3-brane charges inside $\mathbb{A}$-type singularities
(\hyperlink{FigureD7BraneConfiguration}{\it Figure 3})
have coefficients in 3-Cohomotopy, thus leading to the moduli spaces \eqref{OrderedConfigurationSpaceAsSpaceOfIntersections}:

\noindent
{\bf (i)}
On the ambient 11d orbi-spacetime, the tangentially J-twisted 4-Cohomotopy coefficient
(according to \cite{FSS19b}) as
applicable to an $\mathbb{A}_{\ShiftedLevel-1}$-type orbi-singularity $\RealNumbers^{3,1} \times \ComplexNumbers_{\mathrm{cpt}} \times \mathbb{R}^1 \times \HomotopyQuotient{\Quaternions}{\CyclicGroup{\ShiftedLevel}}$
(by \cite{SS19TadpoleCancellation}\cite[Def. 5.28 (ii)]{SS20OrbifoldCohomology})
is the representation 4-sphere of the left multiplication action $\CyclicGroup{\ShiftedLevel} \acts \, \Quaternions$ (by \cite[Thm. 5.16, Ex. 5.29 (ii)]{SS20OrbifoldCohomology}).

\noindent
{\bf (ii)}
On the $\mathrm{M5} \subset \mathrm{MK6}$-worldvolume the induced twisted 7-Cohomotopy coefficient
(by \cite[(118)]{FSS19b}\cite[(43)]{FSS19HopfWZ}\cite[(12)]{FSS20a})
is the lift of this action through the quaternionic Hopf fibration (e.g. \cite[(11)]{FSS20a}) --
which is the representation 7-sphere of the left multiplication action $\CyclicGroup{\ShiftedLevel} \acts \Quaternions \times \RealNumbers^4$ on the left factor,
\vspace{-3mm}
\begin{equation}
  \label{ThreeCohomotopyChargeOnM5Brane}
  \CyclicGroup{\ShiftedLevel}
  \acts
  \;
  \ShapeOfSphere{7}_L
  \;\;
  \coloneqq
  \;\;
  S
  \Big(
    \!\!\!\!
    \begin{tikzcd}[column sep=-5pt]
      \mathbb{H}
      \ar[out=180-60, in=60, looseness=3.8, "\scalebox{.77}{$\mathclap{
        \scalebox{.9}{$\CyclicGroup{\ShiftedLevel}$}
      }\;\;\;$}"{pos=.41, description},shift right=1]
      &\oplus&
      \mathbb{R}^4
    \end{tikzcd}
    \!\!\!
  \Big)
  \xrightarrow{
    \phantom{--}
    h_{\Quaternions}
    \phantom{--}
  }
  S
  \Big(
    \!\!\!\!
    \begin{tikzcd}[column sep=-5pt]
      \mathbb{H}
      \ar[out=180-60, in=60, looseness=3.8, "\scalebox{.77}{$\mathclap{
        \scalebox{.9}{$\CyclicGroup{\ShiftedLevel}$}
      }\;\;\;$}"{pos=.41, description},shift right=1]
    \end{tikzcd}
    \!\!\!
  \Big)
  =:
  \CyclicGroup{\ShiftedLevel}
  \acts \,
  S^4_L
  \,,
\end{equation}

\vspace{-2mm}
\noindent in that (the topological sector of) the Hopf WZ term on the M5-brane  (\cite[(46)]{FSS19HopfWZ}\cite[(1)]{FSS20TwistedString})
is given by dashed sections shown in the following pullback diagram on the left (as in \cite[(9)]{FSS20TwistedString}):
\begin{equation}
  \adjustbox{raise=-15pt}{
  \begin{tikzcd}[column sep=70pt, row sep=36pt]
    S^3
    \ar[r]
    &[-68pt]
    \widehat{
      \mathbb{R}^{3,1}
      \times
      \ComplexNumbers
      \times
      \RealNumbers
    }
    \ar[d]
    \ar[r]
    \ar[dr, phantom, "\mbox{\tiny\rm(pb)}"]
    &
    \ShapeOfSphere{7}_L
    \ar[
      d,
      "{
        h_{\mathbb{H}}
      }"
    ]
    \\
    &
    \underset{
      \mathclap{
      \raisebox{-3pt}{
        \tiny
        \color{darkblue}
        \bf
        \def\arraystretch{.9}
        \begin{tabular}{c}
          $\mathrm{M}5 \!\subset\! \mathrm{MK}6$
          \\
          worldvolume
        \end{tabular}
      }
      }
    }{
      \underbrace{
        (
        \mathbb{R}^{3,1}
        \times
        \ComplexNumbers
        \times
        \RealNumbers^1
        )
      }
    }
    \times
    \underset{
      \mathclap{
      \raisebox{-3pt}{
        \tiny
        \color{darkblue}
        \bf
        \def\arraystretch{.9}
        \begin{tabular}{c}
          ADE-
          \\
          singularity
        \end{tabular}
      }
      }
    }{
      \underbrace{
        \mathbb{H}_{\mathrm{cpt}}
      }
    }
    \ar[
      r,
      "{
        \underset{
          \mathclap{
          \raisebox{-8pt}{
            \tiny
            \def\arraystretch{.9}
            \begin{tabular}{c}
              \color{greenii}
              \bf
              unit M5/MK6
              \\
              \color{greenii}
              \bf
              charge
              \\
              \color{gray}
              \cite{SS19TadpoleCancellation}
            \end{tabular}
          }
          }
        }{
          \mathrm{pr}_2
        }
      }"{below}]
    \ar[
      u,
      bend left=30,
      dashed,
      "{
        \mbox{
          \tiny
          \def\arraystretch{.9}
          \begin{tabular}{c}
            \color{gray}
            \cite{FSS19HopfWZ}
            \\
            \color{gray}
            \cite{FSS20TwistedString}
          \end{tabular}
          \hspace{-12pt}
          \def\arraystretch{.9}
          \begin{tabular}{c}
            \color{greenii}
            \bf
            charge inside
            \\
            \color{greenii}
            \bf
            M5/MK6
          \end{tabular}
        }
      }"{xshift=6pt}
    ]
    &
    \ShapeOfSphere{4}_L
  \end{tikzcd}
  }
  \;\;
  \adjustbox{raise=-9pt}{
  \begin{tikzcd}[column sep=16pt]
    {}
    \ar[
      rr,
      |->,
      "{
        (-)^{\scalebox{.8}{$\CyclicGroup{\ShiftedLevel}$}}
      }",
      "{
        \mbox{
          \tiny
          \color{orangeii}
          \bf
          \def\arraystretch{.9}
          \begin{tabular}{c}
            restriction to
            \\
            $\CyclicGroup{\ShiftedLevel}$-fixed locus
          \end{tabular}
        }
      }"{below, yshift=-4pt}
    ]
    &&
    {}
  \end{tikzcd}
  }
  \hspace{-10pt}
  \adjustbox{raise=-15pt}{
  \begin{tikzcd}[column sep=55pt, row sep=36pt]
    \mathllap{(}
      \mathbb{R}^{3,1}
      \times
      \ComplexNumbers
      \times
      \RealNumbers^1
    )
    \mathrlap{
      \times
      S^3}
    \ar[
      r,
      start anchor={[xshift=15pt]}
    ]
    \ar[d]
    \ar[dr, phantom, "\mbox{\tiny\rm(pb)}"]
    &
    \ShapeOfSphere{3}
    \ar[d]
    \\
    \underset{
      \mathclap{
      \raisebox{-3pt}{
        \tiny
        \color{darkblue}
        \bf
        \def\arraystretch{.9}
        \begin{tabular}{c}
          4d
          \\
          spacetime
        \end{tabular}
      }
      }
    }{
      \underbrace{
        \mathbb{R}^{3,1}
      }
    }
    \;\times\;
    \underset{
      \mathclap{
      \raisebox{-3pt}{
        \tiny
        \color{darkblue}
        \bf
        \def\arraystretch{.9}
        \begin{tabular}{c}
          cylinder over
          \\
          surface
        \end{tabular}
      }
      }
    }{
      \underbrace{
        \ComplexNumbers
        \times
        \RealNumbers
      }
    }
    \ar[r]
    \ar[
      u,
      bend left=30,
      dashed,
      "{
        \mbox{
          \tiny
          \color{greenii}
          \bf
          \def\arraystretch{.9}
          \begin{tabular}{c}
            3-Cohomotopy
            \\
            cocycle
          \end{tabular}
        }
      }"{xshift=6pt}
    ]
    &
    S^0
    \,.
  \end{tikzcd}
  }
\end{equation}

\vspace{-3mm}
\noindent
{\bf (iii)}
Under passage to the singularity by
restricting this diagram to the $\CyclicGroup{\ShiftedLevel}$-fixed locus,
as shown on the right of
\eqref{ThreeCohomotopyChargeOnM5Brane},
the pullback diagram manifestly reduces to a direct product with the $S^3$-fiber of the quaternionic Hopf fibration; and hence the dashed sections inside the singularity are equivalently plain maps to $S^3$, hence cocycles in 3-Cohomotopy (quantizing the $H_3$-flux on the M5, under {\HypothesisH}).
\end{remark}

\medskip

\noindent
{\bf In conclusion} we have thus made plausible the following:
\vspace{-1mm}
\begin{itemize}
\setlength\itemsep{-1pt}

\item[{\bf (1)}]
The charges of flat D3/D7-branes in F-theory
(\hyperlink{FigureD7BraneConfiguration}{\it Figure 2})
ought to be measured, according to the widely-expected {\HypothesisK} (Rem. \ref{HypothesesAboutBraneChargeQuantization}), by the (secondary) \TED-K-theory
of \eqref{SESForD7BraneChargeOnSigmaTwo}
their transverse complex curve inside an $\mathbb{A}$-type orbi-singularity; and the latter, computed via Prop. \ref{Degree1ConformalBlockInTEdKTheory} as shown in \eqref{ConcludingD7BraneCharges}, indeed neatly matches
(as per Rem. \ref{ListOfExoticD7BraneChargeAspects})
a host of expectations (listed in Rem. \ref{DSevenBraneChargesExpeectedInFTheory}) about exotic defect brane charges in F-theory.

The only shortcoming at this point is that the \TED-K-theory of the F-theoretic transverse space sees of the expected $\suTwoAffine{\Level}$-CFT \eqref{ShiftOfLevelViaStringsOnALE} only the conformal blocks of degree 1. But:

\item[{\bf (2)}]
Dualizing the situation to that of M3-branes in M-theory (\hyperlink{MBraneConfigurations}{\it Figure 3})
allows to enhance the traditional {\HypothesisK} to our more recently proposed {\HypothesisH} (Rem. \ref{HypothesesAboutBraneChargeQuantization}). This has the consequence
(\hyperlink{TableOfCohomologyOfCohomotopyCocycleSpaces}{\it Table 1})
that the orbi-space to evaluate the \TED-K-theory on is enhanced from the plain transverse space to its 3-Cohomotopy cocycle space, which is given, more generally, by the configuration space of points \eqref{OrderedConfigurationSpaceAsSpaceOfIntersections} inside the transverse space. On this, the \TED-K-theory is computed now via Thm. \ref{ConformalBlockInTEdKTheory}, which completes the F-theoretic analysis by including now also the conformal blocks in all higher degrees.
\end{itemize}

Notice that the arguments about brane charges in this one section here are {\it necessarily} non-rigorous, since the folklore about F/M-theory of branes and dualities which they refer to does not yet exist as a complete well-defined theory. But conversely, to the extent that this folklore plausibly matches to the well-defined \TED-K-theory of Cohomotopy cocycle spaces, it supports the {\HypothesisH} that the latter is (at least partly) the missing rigorous definition of F/M-theory. Assuming this, we may then turn this around and rigorously explore F/M-theory by purely mathematical analysis of the \TED-K-theory of Cohomotopy cycles spaces.

\newpage

\section{Anyonic quantum computation via {TED-K}}
\label{Outlook}

We close by highlighting some
potentially far-reaching consequences of the above analysis.
The following is an outline of a research program which will be laid out in
more detail, beginning with \cite{SS22TQC}.

\medskip

\noindent
{\bf 1. Quantum field theory of Defect branes.}
It is in fact an old observation that the correlators of some Euclidean quantum field theories are encoded in the de Rham cohomology
of a configuration space of points: For 3d Chern-Simons theory this goes back to \cite{AxelrodSinger94}, further discussed
in \cite[\S 3]{AF96}\cite[Rem. 3.6]{BottCattaneo98}, leading to Kontsevich's graph complexes; a discussion for more general
quantum field theories is in \cite{Berghoff14}.
The evident suggestion that therefore the {\it generalized} cohomology (such as the K-theory) of configuration spaces might
reflect yet more details of quantum field theory seems not to have been explored much yet (an exception being the
note \cite{Zhou18} -- the {\it Hilbert schemes} considered there are essentially the algebro-geometric version of configuration
spaces of points). The discussion here and its higher dimensional analog (see \hyperlink{TableOfCohomologyOfCohomotopyCocycleSpaces}{\it Table 1})
suggests that a fair bit of deep structure in quantum field theory is reflected in the (twisted, equivariant, differential, generalized, ...)
cohomology of configuration spaces of points, of which here we only discussed the simplest examples.

\smallskip
One immediate interesting generalization of the present discussion would be to consider M5-brane topologies more complex
than the $\RealNumbers^{3,1} \times \Sigma^2$ considered here. If we generalize $\mathbb{R}^{3,1}$ to a topologically
non-trivial spacetime manifold $M^{3,1}$, then the \TED-K-theory
$$\mathrm{KU}^{\NumberOfProbeBranes + [\FlatConnectionForm]}\Big( \ConfigurationSpace{\NumberOfProbeBranes}\big( M^{3,1} \times \Sigma^2 \times \HomotopyQuotient{\ast}{\CyclicGroup{\ShiftedLevel}}\big) \Big)
$$
will extend the $\slTwoAffine{\Level}$-conformal blocks associated to the $\Sigma^2$-factor by Prop.
\ref{Degree1ConformalBlockInTEdKTheory} by data attached to $M^{3,1}$. If the discussion in \cref{QuantumStatesAndObservables}
is anything to go by, this cohomological approach could reveal further fine-structure of the 6d CFT in its class-S sector.

\medskip

 \noindent
 {\bf 2. Topological quantum computation on Defect branes.}
 The idea of topological quantum computation with anyons
 (\cite{Kitaev03}\cite{FKLW01}, review in \cite{NSSFS08}\cite{Wang10})
 is traditionally thought of as implemented within solid state physics,
 where anyons are to be realized, in one way or another, as effective codim=2 defects
 in some quantum material whose gapped ground state is governed by an effective 2+1-dimensional Chern-Simons theory with WZW boundaries (e.g. \cite{Lerda92}\cite{Kitaev06}\cite{Rao16}).
Meanwhile, it has become understood that relevant solid state quantum systems (potentially) supporting quantum computation tend to be analogs of {\it intersecting brane models} in string theory/M-theory, a statement known as the
{\it AdS/CMT correspondence}
(\cite{ZaanenLiuSunSchalm15}\cite{HartnollLucasSachdev18}\cite{Zaanen21},
 closely related to the {\it AdS/QCD correspondence} \cite[\S 15-18]{RhoZahed16},
 both of which being variants of the more famous {\it AdS/CFT correspondence} \cite{AGMOO99},
 i.e., of the ``holographic principle'' in string/M-theory).

\smallskip
While the realization of anyons in condensed matter theory seems possible (\cite{NLGM20}) but remains somewhat elusive,
the analogous anyonic defects in string theory/M-theory would theoretically be ubiquitous, as soon as it is clarified that and how defect branes such as D7-branes obey anyon statistics (cf \cite[p. 65]{deBoerShigemori12}). But this is just what our discussion in \cref{QuantumStatesAndObservables} argues for (see Rem. \ref{ExoticDefectBranesAndAnyonStatistics}).

\smallskip
This suggests that topological quantum computation may have its natural conceptual home not
in ``mesoscopic'' condensed matter physics, but in the truly ``microscopic'' high energy physics
according to string/M-theory, specifically in the dynamics of defect branes,
with the former only being the approximate image or analog of the latter under the AdS/CMT-correspondence.
A similar state of affairs is already believed to hold for quantum error correction,
which in recent years has been argued
(\cite{ADH14}\cite{PYHP15})
to naturally reflect the quantum information theory inherent in the AdS/CFT correspondence (review in \cite{Harlow17}\cite{SSW20}\cite{JahnEisert21}).

\smallskip
In the case of quantum error correction, this seemingly remote identification of aspects of quantum computation with
aspects of stringy quantum gravity
has been argued to help with substantial practical problems in quantum computation
(see \cite[p. 14]{WVSB20}\cite[p. 16(4)]{Harlow20}\cite[p. 1]{CDCW21}).
Similarly, an understanding of anyon braiding as fundamentally describing defect brane dynamics might usefully
inform the discussion of topological quantum computation.

For example, it is striking that
among all mathematically possible anyon species (braid representations),
the \TED-K-theory in Prop. \ref{Degree1ConformalBlockInTEdKTheory} \&
Thm. \ref{ConformalBlockInTEdKTheory}
singles out
specifically monodromy braid representations realized on conformal blocks in conformal field theory (via Rem. \ref{ExoticDefectBranesAndAnyonStatistics}),
and here specifically the $\slTwoAffine{\Level}$-conformal blocks: These are known as ``$\mathrm{SU}(2)_k$-anyons'', such as Majorana/Ising-anyons for $k=2$, and Fiboniacci-anyons for $k = 3$ (e.g. \cite{TTWL08}\cite{GATHLTW13}\cite[p. 11]{SarmaFReedmanNayak15}\cite[\S III]{JohansenSimula20}) and include just those anyon species plausibly realized in nature, notably via the fractional quantum Hall effect (see \cite{MooreRead91}\cite{EPSS12}\cite{NLGM20}, review in \cite{Su18}) or via Majorana modes bound to vortices in topological superconductors (see \cite{Ivanov01}\cite{Beenakker11}\cite{SunJia17}), and hence in particular in quantum computers (e.g., \cite{KBWS21}). In fact, as we speak an experimental proof of principle for topological qbits based on Majorana anyons (a special form of $\suTwoAffine{2}$ Ising anyons \cite{SarmaFReedmanNayak15}) is being claimed by a major quantum computing lab (\cite{Nayak22}, following \cite{Pikulin21}).

\begin{center}
\label{SUTwoANyonSpecies}
\def\arraystretch{1.6}
\begin{tabular}{|c||c|c|c|c|}
  \hline
    \bf
  \begin{tabular}{c}
    WZW CFT
  \end{tabular}
  &
  $\suTwoAffine{2}$
  &
  $\suTwoAffine{3}$
  &
  $\suTwoAffine{\,k \geq 4}$
  &
  $\cdots$
  \\
  \hline
  \hline
  \bf
  \def\arraystretch{.9}
  \begin{tabular}{c}
    $\mathclap{\phantom{\vert^{\vert^{}}}}$
    realistic
    \\
    anyon species
  \end{tabular}
  &
  \def\arraystretch{.9}
  \begin{tabular}{c}
    Ising /
    Majorana
  \end{tabular}
  &
  \def\arraystretch{.9}
  \begin{tabular}{c}
    Fibonacci / Potts
  \end{tabular}
  &
  parafermions
  &
  $\cdots$
  \\
  \hline
\end{tabular}
\end{center}

\smallskip
Moreover, the realization of topological quantum computation on defect branes such as D3s or M3s
may not be all that remote from observed physics:
After all, Randall-Sundrum-like ``brane world models'' (e.g. \cite{KLLL03}), where the physically observed spacetime is identified with the worldvolume of 3-brane intersections in an unobserved higher dimensional bulk spacetime, notably with $\mathrm{D3}/\mathrm{D7} @ \mathbb{A}_{\ShiftedLevel-1}$-branes (e.g. \cite{GKP01}\cite{BayntonBurgessNierop09}\cite{MaharanaPalti13}), are the way in which all of type I/II/M/F-theory realizes quasi-realistic particle physics (see \cite{IbanesUranga12}, the only exception to this rule being HET-theory models).
In fact there is decent experimental indication for this notion:

\vspace{-.2cm}
\begin{itemize}
\setlength\itemsep{-1pt}

\item[{\bf (1)}]
The geometric engineering of quantum chromodynamics on such brane world models
(``holographic QCD'') yields a strikingly realistic description of experimentally observed confined hadron physics (eg \cite{Erlich10}\cite{KimYi11}\cite[\S 15]{RhoZahed16}).

In fact, seen among the (actually observed) hadrons (as opposed to their un-observed constitutent quarks), the notorious super-partners do seem to be experimentally realized as baryon/meson pairs (``hadron supersymmetry'', see  \cite{Lichtenberg99}\cite{KlemptMetsch12} for the general phenomenon and, eg., \cite{BTDL16}\cite{Brodsky21} for its emergence in holographic QCD).

\item[{\bf (2)}]
The ongoing LHC measurements of B-meson anomalies
(see e.g. \cite{CFFMIN21}\cite{ILOS21})
have been argued to be experimental signatures of just such Randall-Sundrum-type 3-brane models (\cite{FMILSS22}\cite{FMIPS20}, review in \cite[\S 3]{Lizana22}\cite[\S 5.4]{AltmannshoferZupan22}, see also \cite{BlankeCrivellin18}).

\end{itemize}
\vspace{-.15cm}
Under this identification of 3-branes with observed spacetime, the model of topological quantum computation on defect branes should amount to operating
the hypothetical axio-dilaton-field \eqref{TheAxioDilaton}. One could go on to quote speculations that this string theoretic axio-dilaton field is secretly already observed (e.g. \cite[\S A.1]{CGRW21}) namely in the form of {\it fuzzy dark matter} (review in \cite{Khoury21}) . While this indicates
rich possibilities for
further phenomenological exploration of the idea of quantum computation
with defect branes, here is not the place to discuss this further.

\medskip

\noindent
{\bf 3. Topological quantum computation in Topological phases of matter.}
More concretely, our identification of realistic $\suTwoAffine{\Level}$-anyons in the \TED-K-theory of the {\it punctured plane} is reminiscent of the well-known fact that twisted equivariant K-theory of {\it tori} classifies topological phases (gapped phases) of quantum materials (\cite{Kitaev09}\cite{FreedMoore12}\cite{Thiang14}\cite{Thiang15}). These, of course, are just the systems thought to host anyonic defect excitations, generally (e.g. \cite[p. 4]{Kitaev06}).

We suggest that this is not a coincidence: While these tori are often thought of as Brillouin tori of quasi-momenta, their twisted equivariant K-theory is in fact equivalent (namely T-dual \cite{MathaiThiang15}\cite{GomiThiang18}) to dually-twisted equivariant K-theory of the actual position-space tori which represent cells in the quantum material's crystalline structure. But this means that anyon defects will appear as punctures in these position-space tori, see, e.g., \cite{Einarsson90}\cite{HosotaniHo92}\cite{GreiterWilczek92}\cite{PuJain21}.
(Notice that the dual situation of punctures in the momentum-space torus is known to correspond to Weyl points in Weyl semimetals, see \cite{MathaiThiang16}.)

In view of the above discussion, this suggests that the \TED-K-theory of {\it punctured tori} is in fact the natural language for discussing topological quantum computation on anyonic defects in topological phases of matter:\footnote{
   \label{PhysicsOfBraiding}
   While here we are all focused on abstract theory,
   it is interesting to highlight the practical viability of this approach:

    On the one hand, it is a general open problem in topological quantum computation of how to actually move anyons around, once they have been realized, hence how to actually braid their worldlines. This is certainly an engineering problem but also a theoretical problem (e.g. \cite[p. 8]{Kitaev06}\cite[p. 7-8]{SarmaFReedmanNayak15}), given that the defining properties of anyons revolve around them being inert to interactions and tending to behave like classical defects. This issue is only further amplified by the above identification/analogy of anyons with
    with defect branes (in \cref{QuantumStatesAndObservables}).
    In other words, while it is mathematically most natural to consider curves in the configuration space of points (braids), it remains generally unclear and in fact implausible (at least if these points represent defects, see also \cite[p. 7]{SarmaFReedmanNayak15}) that these may be realized as actual motions of anyons in time, let alone as a quantum propagations that admit coherent superposition.

  But a neat solution to this fundamental issue has been presented in \cite{BondersonFReedmanNayak08}\cite{BondersonFreedmanNayak08}\cite{Bonderson12}\cite{ZhengDuaJiang16}, where it is claimed that the braid group actions on anyon quantum states may be implemented, to suitable accuracy, by performing certain sequences of measurements of their topological charges, without actually moving the anyons (review in \cite[p. 8]{SarmaFReedmanNayak15}\cite{Beenakker19}), hence by ``braiding without braiding'' (\cite{VijayFu16}), in fact by braiding via ``quantum teleportation'' (as in \cite{Zhang06}).
  All this is understood fairly concretely for Majorana $\suTwoAffine{2}$-anyons,
  whose experimental realization has just been announced (\cite{Nayak22}, following \cite{Pikulin21}).
}
The \TED-K theory of the underlying un-punctured torus classifies the topological phase of the ambient quantum material, and its corrections around the punctures encode the topological order reflected by the presence of anyons. We will further discuss this in \cite{SS22TQC}:

\begin{center}
\def\arraystretch{2}
\begin{tabular}{|c||ccc|}
  \hline
  \def\arraystretch{.9}
  \bf
  \begin{tabular}{c}
    crystalline
    position-space /
    \\
    transverse toroidal orbifold
  $\mathclap{\phantom{\vert_{\vert_{\vert_{\vert}}}}}$
  \end{tabular}
  &
  $
  \overset{
    \raisebox{3pt}{
      \tiny
      \color{darkblue}
      \bf
      plain
    }
  }{
  \HomotopyQuotient
    {\Torus{2}}
    {G}
  }
  $
  &
  \hspace{-50pt}
  \begin{tikzcd}
    {}
    \ar[from=r, hook']
    &
    {}
  \end{tikzcd}
  \hspace{-50pt}
  &
  $
  \overset{
    \mathclap{
    \raisebox{3pt}{
      \tiny
      \bf
      \color{darkblue}
      \begin{tabular}{c}
        punctured /
        with defects
      \end{tabular}
    }
    }
  }{
  (
  \HomotopyQuotient
    {\Torus{2}}
    {G}
  )  \setminus \{\vec z\}
  }
  $
  \\
  \hline
  \def\arraystretch{.9}
  \bf
  \begin{tabular}{c}
    TED-K
    \\
    cohomology
    $\mathclap{\phantom{\vert_{\vert_{\vert}}}}$
  \end{tabular}
  &
  \def\arraystretch{.9}
  \begin{tabular}{c}
    topological
    \\
    phases
    $\mathclap{\phantom{\vert_{\vert_{\vert}}}}$
  \end{tabular}
  &&
  \def\arraystretch{.9}
  \begin{tabular}{c}
    anyonic
    \\
    topological orders
    $\mathclap{\phantom{\vert_{\vert_{\vert}}}}$
  \end{tabular}
  \\
  \hline
\end{tabular}
\end{center}

\noindent
From the point of view of defect branes, this corresponds to allowing the transverse space $\Sigma^2$ in \hyperlink{FigureD7BraneConfiguration}{\it Figure 2}, \hyperlink{MBraneConfigurations}{\it Figure 3}
to be a (punctured) toroidal orbifold, which is exactly what one wants to consider also in F/M-theory, notably in the context of the AGT correspondence, discussed in \cref{QuantumStatesAndObservables}.

\bigskip
\noindent  Hisham Sati, {\it Mathematics, Division of Science,
\\
and Center for Quantum and Topological Systems (CQTS), NYUAD Research Institute,
\\
New York University Abu Dhabi, UAE.
}
\\
{\tt hsati@nyu.edu}
\\
\\
\noindent  Urs Schreiber, {\it Mathematics, Division of Science,
\\
and Center for Quantum and Topological Systems (CQTS), NYUAD Research Institute,
\\
New York University Abu Dhabi, UAE;
\\
on leave from Czech Academy of Science, Prague.}
\\
{\tt us13@nyu.edu}

\end{document}